%% file: quantumpaper_may2014.tex
\documentclass[a4paper, 10pt]{article} 

\usepackage[centertags]{amsmath}
\usepackage{amsmath}
\usepackage[dutch,british]{babel}
\usepackage{amsfonts}
\usepackage{amssymb} 
\usepackage{amsthm}
\usepackage{newlfont}
\usepackage{color}
 \usepackage{bbm}
 \usepackage{float}
 \usepackage{import}
  \usepackage{epsfig}
\usepackage{tikz}
\usepackage{caption}

\usepackage{stmaryrd}
\usepackage{mathrsfs}
\usepackage{pstricks,pst-node}
\usepackage{fancyhdr}
\usepackage{graphicx}
\usepackage{fancybox}
\usepackage{geometry}
 \usepackage{psfrag}
 
 \usepackage{multirow}

\theoremstyle{plain}
  \newtheorem{theorem}{Theorem}[section]
  \newtheorem{corollary}[theorem]{Corollary}
  \newtheorem{proposition}[theorem]{Proposition}
  \newtheorem{lemma}[theorem]{Lemma}
  
\theoremstyle{definition}
  \newtheorem{definition}{Definition}[section]

\theoremstyle{remark}

\numberwithin{equation}{section}

 \newcounter{smallarabics}
\newenvironment{arabicenumerate}
{\begin{list}{{\normalfont\textrm{\arabic{smallarabics})}}}
  {\usecounter{smallarabics}\setlength{\itemindent}{0cm}
  \setlength{\leftmargin}{5ex}\setlength{\labelwidth}{4ex}
  \setlength{\topsep}{0.75\parsep}\setlength{\partopsep}{0ex}
   \setlength{\itemsep}{0ex}}}
{\end{list}}

\newcounter{smallroman}

\newcommand{\ben}{\begin{arabicenumerate}}
\newcommand{\een}{\end{arabicenumerate}}

\DeclareMathOperator{\Tr}{Tr}

\newcommand\otimesal{\mathop{\hbox{\raise 1.6 ex
  \hbox{$\scriptscriptstyle\mathrm{al}$}
\kern -0.92 em \hbox{$\otimes$}}}}
\newcommand\oplusal{\mathop{\hbox{\raise 1.6 ex
  \hbox{$\scriptscriptstyle\mathrm{al}$}
\kern -0.92 em \hbox{$\oplus$}}}}
\newcommand\Gammal{\hbox{\raise 1.7 ex
\hbox{$\scriptscriptstyle\mathrm{al}$}\kern -0.50 em $\Gamma$}}
\renewcommand\i{\mathrm{i}}

\let\al=\alpha \let\be=\beta  

  \let\ga=\gamma 
\let\ka=\kappa  \let\om=\omega 
\let\si=\sigma

 \let\Ga=\Gamma  \let\Om=\Omega

\newcommand{\caB}{{\mathcal B}}

\newcommand{\caD}{{\mathcal D}}

\newcommand{\caF}{{\mathcal F}}

\newcommand{\caH}{{\mathcal H}}
\newcommand{\caI}{{\mathcal I}}
\newcommand{\caJ}{{\mathcal J}}
\newcommand{\caK}{{\mathcal K}}

\newcommand{\caN}{{\mathcal N}}

\newcommand{\caP}{{\mathcal P}}

\newcommand{\caR}{{\mathcal R}}
\newcommand{\caS}{{\mathcal S}}

\newcommand{\caZ}{{\mathcal Z}}

\newcommand{\scrA}{{\mathscr A}}
\newcommand{\scrB}{{\mathscr B}}

\newcommand{\scrD}{{\mathscr D}}
\newcommand{\scrE}{{\mathscr E}}
\newcommand{\scrF}{{\mathscr F}}
\newcommand{\scrG}{{\mathscr G}}
\newcommand{\scrH}{{\mathscr H}}

\newcommand{\scrK}{{\mathscr K}}

\newcommand{\bbC}{{\mathbb C}}

\newcommand{\bbE}{{\mathbb E}}

\newcommand{\bbH}{{\mathbb H}}

\newcommand{\bbN}{{\mathbb N}}

\newcommand{\bbP}{{\mathbb P}}

\newcommand{\bbR}{{\mathbb R}}
\newcommand{\bbS}{{\mathbb S}}

\newcommand{\bbZ}{{\mathbb Z}}

\newcommand{\opunit}{\text{1}\kern-0.22em\text{l}}

\newcommand{\frB}{{\mathfrak B}}

\newcommand{\frG}{{\mathfrak G}}

\newcommand{\bsS}{{\boldsymbol S}}

\newcommand{\norm}{ \|}
\newcommand{\str}{ |}

\newcommand{\e}{{\mathrm e}}

\renewcommand{\d}{{\mathrm d}}

\newcommand{\spann}{\mathrm{Span}}
\newcommand{\beq}{ \begin{equation} }
\newcommand{\eeq}{ \end{equation} }
\newcommand{\bet}{ \begin{theorem} }
\newcommand{\eet}{ \end{theorem} }

\newcommand{\lone}{\mathbbm{1}}

\newcommand{\links}{\mathrm{L}}
\newcommand{\rechts}{\mathrm{R}}
\newcommand{\adjoint}{\mathrm{ad}}
\newcommand{\ad}{\adjoint}

\newcommand{\sgn}{\mathrm{sgn}}

\newcommand{\diam}{\mathrm{diam}}
\newcommand{\distance}{\mathrm{dist}}

\newcommand{\tr }{\mathrm{tr}}
\newcommand{\proj }{\mathrm{proj}}
\newcommand{\degree}{\mathrm{deg}}
\newcommand{\pot}{X}

\newcommand{\bigvolume}{\Lambda}
\newcommand{\conn}{\#_c}

\begin{document}

\title{
Asymptotic quantum many-body localization from thermal disorder}
\author{W. De Roeck, F. Huveneers}
\author{
\vspace{0.1cm}\\
Wojciech De Roeck\\
\normalsize
Instituut voor Theoretische Fysica, K.U.Leuven\\
\normalsize
Celestijnenlaan 200 D\\
\normalsize
3001 Leuven, Belgium\\
\normalsize
E-mail: \texttt{wojciech.deroeck@fys.kuleuven.be }\\
\vspace{0.5cm}\\
Fran\c cois Huveneers\\
\normalsize
CEREMADE\\
\normalsize
Universit\' e Paris-Dauphine\\
\normalsize
Place du Mar\' echal De Lattre De Tassigny\\
\normalsize
75775 Paris CEDEX 16, France\\
\normalsize
E-mail: \texttt{huveneers@ceremade.dauphine.fr}
}
\date{}

\maketitle
\vspace{0.5cm}

\begin{abstract}
We consider a quantum lattice system with infinite-dimensional on-site Hilbert space, very similar to the Bose-Hubbard model.  We investigate many-body localization in this model, induced by thermal fluctuations rather than disorder in the Hamiltonian.  We provide evidence that  the Green-Kubo conductivity $\kappa(\beta)$, defined as the time-integrated current autocorrelation function,  decays faster than any polynomial in the inverse temperature $\beta$ as $\beta \to 0$. 
More precisely,  we define approximations $\kappa_{\tau}(\beta)$ to $\kappa(\beta)$ by integrating the current-current autocorrelation function up to a large but finite time $\tau$ and we rigorously show that  $\be^{-n}\kappa_{\be^{-m}}(\beta)$ vanishes as $\be \to 0$,  for  any $n,m \in \bbN$ such that $m-n$ is sufficiently large. 

\end{abstract}

\vspace{\stretch{1}}

\section{Introduction}

\subsection{Localization and its characterization}
The phenomenon of localization was introduced in the context of non-interacting electrons in random lattices in \cite{andersonabsence}.  It is now widely accepted that in such systems, a delocalization-localization (or metal-insulator) transition occurs as the disorder strength is increased.   This transition is often discussed by referring to the nature of the one-particle wavefunctions that are exponentially localized in space in the insulator, but delocalized in the metallic regime.   The localized phase has been studied with mathematical rigor starting with \cite{frohlichspencerabsence}, whereas for the delocalized regime, this has not been successful up to now. 

The  natural question how interactions modify this transition has received renewed attention lately.  Both theoretical  \cite{baskoaleineraltschuler,vosk2013many} and numerical \cite{palhuse,oganesyanhuse} work suggests that the localization-delocalization transition persists, at least for short range interactions. When talking about models with interaction, most authors choose a model where the localization is manifest in the absence of interaction (whereas, in the original model of \cite{andersonabsence}, it was a highly nontrivial result). For example, a simple model from \cite{palhuse} is the random-field Ising chain
\beq \label{eq: ham spin chain}
  H = \sum_{x=1}^{L}  h_x S^{(3)}_x + J  \bsS_x \cdot \bsS_{x+1} 
  \eeq
where $\bsS_x=(S^{(1)}_x,S^{(2)}_x,S^{(3)}_x)$ are the Pauli-matrices at site $x$ and $h_x$ are i.i.d. random variables with $\bbE(h_x)=0$.
We  think of many-body localization as the property that a local in space excess of energy  does not spread into the rest of the system. However, before formalizing this intuition, we give another possbile definition of many-body localization, used e.g.\ by \cite{palhuse, imbriespencer}, in the model defined by \eqref{eq: ham spin chain}.   Let  $\Psi$ label eigenfunctions of $H$, then `many-body localization' at infinite temperarure $\be=0$ ($\be$ is the inverse temperature) could be defined as the occurrence of the inequality
\beq \label{eq: definition mbl}
 \lim_{L \to \infty} \frac{1}{2^L} \sum_{\Psi} \bbE_h(\str\langle \Psi, S_{L/2}^{(3)} \Psi \rangle\str^2)  \neq      \bbE_h( \str\langle S_{L/2}^{(3)} \rangle_{\be=0} \str^2)=0.
 \eeq
where $\langle \cdot \rangle_\be$ on the right hand side refers to the thermal average and $\bbE_h(\cdot)$ refers to disorder average.  Of course, one can also ask whether this inequality holds at $\be >0$, in which case the average over eigenfunctions $\frac{1}{2^L} \sum_{\Psi}$ on the left-hand side should be restricted to those eigenfunctions with an energy density corresponding to the inverse temperature $\be$, and the right hand side does not automatically vanish. 
Depending on the disorder strength, the validity of \eqref{eq: definition mbl} can then depend on the temperature as well. The appeal of the inequality \eqref{eq: definition mbl} is that it violates the so-called Eigenstate Thermalization Hypothesis (ETH) which  states that most eigenvectors of the Hamiltonian define an ensemble that is equivalent to the standard (micro)-canonical ensemble; i.e.\ with the notation as in \eqref{eq: definition mbl}, it states  that, for  for any $\delta>0$, the bound
\beq \label{eq: definition mbl consequence}
\big\str \langle \Psi, S^{(3)}_{L/2} \Psi \rangle -   \langle S^{(3)}_{L/2} \rangle_{\be=0} \big\str \leq \delta
 \eeq
is satisfied for a  fraction of eigenfunctions $\Psi$ that approaches $1$ as $L \to \infty$. 
  Even though the ETH has not been proven for any interesting non-integrable system (the difficulty of doing so is related to the difficulty of proving delocalization), it has nevertheless been accepted by the theoretical physics community, starting with the works \cite{lindenpopescushortwinter,lebowitzgoldsteinmastrodonatotumulka}.  
It is however important to point out that the ETH also fails for ballistic systems like the ideal crystal for which there is surely no localization in the sense of non-spreading of energy excess.

There is at present no mathematical proof of many-body localization.  Some progress was made for the (one-particle) Anderson model on a Cayley tree in  \cite{aizenmanwarzel}, which is often quoted as a toy model for many-body localization and, recently, an approach via iterative perturbation theory for the model \eqref{eq: ham spin chain} was initiated by \cite{imbriespencer} (see \cite{imbriespencernote} for an outline of their strategy in the one-particle setting).

As already indicated, we prefer a characterization that stresses the dynamics of energy fluctuations, and therefore we consider the Green-Kubo formula for the heat conductivity
\beq
\ka(\be)=  \frac{\be^2}{2} \int_{-\infty}^\infty  \, \d t   \lim_{L \to \infty} \sum_{x} \langle j_{L/2}(t)j_{x}(0)\rangle_\be
\eeq
where $j_x(t)$ are local energy currents at site $x$. Many-body localization is then understood as the vanishing of $\ka(\be)$. The picture underlying such a definition is that $\ka(\be)=0$ means that energy excitations do not spread diffusively (or faster than diffusively) through the system.  
  Let us bypass the question of the relation between these two characterizations of many-body localization; in the few cases where there exists up to date a convincing argument for many-body localization, those arguments  would imply $\ka(\be)=0$, as well.   In any case, it seems to us that the characterization via the conductivity is clearly physically relevant. 
  
In classical mechanics, one can consider models of the same flavour: One-particle localization occurs in a chain of harmonic oscillators with random masses. Adding anharmoncity to this setup yields a model that is a candidate for many-body localization, but the expectation seems to be that these models do not exhibit strict many-body localization. However, the phenomenology can still manifest itself through the dependence of $\ka(\be,g)$ on the anharmonicity $g$. Form the works \cite{oganesyanpalhuse, basko,fishman2009perturbation,huveneers}, one conjectures that,  
\beq  \label{eq: asymptotic localization}
\lim_{g \to 0} g^{-n}\ka(g,\be) \to  0, \qquad \text{for any $n >0$}.
\eeq
In other words, the conductivity has a non-perturbative origin for small $g$. Below, we refer to this scenario as 'asymptotic localization'.

\subsection{Thermal disorder instead of quenched disorder}

Whereas all the models hinted at above have disorder in the Hamiltonian, this paper is concerned with the question whether one can in principle replace the disorder by thermal fluctuations, i.e.\ disorder due to the thermal Gibbs state. As far as we see, this question does not have any one-particle analogue but it is natural in many-body systems. 
Indeed, whereas disorder can model defects, it is also sometimes used as a model for slow degrees of freedom that are, in principle, influenced by the rest of the system. 

The fact that randomness in the strict sense of the word is not necessary for localization had up to now been investigated by replacing the random field in the Hamiltonian by a quasi-random field, which is quite different from what we do.   In the one-particle setup, this led to the study of models like the Aubry-Andr{\'e} model \cite{aubry1980analyticity}, and recently it was argued \cite{iyer} that also in the many-body setting, quasi-randomness suffices for many-body localization.
To explain our setup and question, we now introduce our model. We consider a variant of the Bose-Hubbard model: 
\begin{equation} \label{eq: our model}
H
\; =\; \sum_{x =1 }^L N_x^q  + g  (a_x^{*} a_{x+1} + a_x a_{x+1}^{*}), \qquad q >2
\end{equation}
where $a_x,a^*_x$ are annihilation/creation operators of a boson at site $x$ and $N_x=a^*_xa_x$.   For $q=2$, this model is exactly the Bose-Hubbard model.  In fact, the model we study is slightly more general than  \eqref{eq: our model} to avoid conceptual complications related to conserved quantities and nonequivalence of ensembles, see Section \ref{sec: ham}, however this is not relevant for the discussion here. 
W.r.t.\ the thermal state at $g=0$, the occupations $N_x$ behave as i.i.d.\ random variables whose distribution is given by
\beq \label{eq: def prob}
\mathrm{Prob}(N_x=\eta(x)) = \frac{1}{Z_0(\be)} \e^{-\be \eta(x)^q}, \qquad \text{with}\,Z_0(\be)\, \text{a normalizing constant}
\eeq 
We split our Hamiltonian as 
\beq
H =  E_{0}  + \tilde g V,  \qquad \textrm{with $  E_0= \sum_{x} N^q_x $ and $\tilde g, V$ defined in \eqref{def: tilde g}}
\eeq
and we treat $ \tilde g V$ as a perturbation of $E_0 $. 
Intuitively, a perturbative analysis is possible, if for a pair of eigenstates $\eta,\eta' $ of $E_0 $, we have the non-resonance condition
\beq \label{eq: condition perturbation theory}
\str \langle \eta,  \tilde g V  \eta' \rangle \str  \ll   \str  E_0(\eta) -  E_0(\eta')  \str
\eeq
where $E_0(\eta):=\langle \eta, E_0  \eta \rangle $.
Since the distance between consecutive eigenvalues (level spacing) of the operator $N_x^q$ grows roughly as $N^{q-1}_x$ and the matrix elements of $ \tilde g V$, locally at site $x$, grow as $N_x$ (since they are quadratic  in the field operators), the condition  \eqref{eq: condition perturbation theory} seems  satisfied for most pairs $\eta,\eta'$ if $q>2$, that is, with high probability w.r.t.\ the probability measure \eqref{eq: def prob} when $\be$ is sufficiently small.   This is the basic intuition why this model should exhibit  some localization effect at high temperature\footnote{One should not confuse this with the situation at $\be=\infty$, where one expects a quantum phase transition between a conducting superfluid phase and an insulating Mott phase. This has nothing to do with our results.}.
 However, because of the many-body setup, it is not straightforward that the above claims make sense. In particular, it is certainly \emph{false} that one could apply perturbation theory directly to the eigenstates $\eta$ of $E_0$.  Indeed, since the number of eigenstates should be thought of as $C^{L}$ and the range of energies has width $CL$, the level spacing (difference between nearest levels) vanishes fast as $L \to \infty$.  Therefore, the locality of the operators is a crucial issue that should be used in making the above heuristics precise.  Instead of having resonant and non-resonant configurations $\eta$, we will assign to any $\eta$ 'resonance spots' (where a local version of \eqref{eq: condition perturbation theory} fails). 
 
 Up to now, the heuristic reasoning is in fact no different from the one that would develop for the disordered Ising chain, except that we replaced the 'disorder distribution' by 'distribution in the uncoupled Gibbs state'.  The difference kicks in when one realizes that the non-resonance condition is not static but it can change as the dynamics changes the occupations $\eta$.  Therefore, it is not sufficient to argue that resonant spots are sparse, but one should investigate the dynamics of these resonance spots and exclude that this dynamics induces a current. The most intuitive part of this issue takes the form of a question in graph theory: The vertices of the graph are the configurations $\eta$ and the edges are pairs of configurations that satisfy some resonance condition.  If the connected components of this graph are small, i.e.\ they typically consist of a few configurations, then this hints at localization.   The main problem to be overcome in the present article is to show that, indeed, typical graphs decompose into many small disconnected components. Our analysis is however only valid in the limit $\beta \rightarrow 0$, and for this reason, we do not know yet, even at an heuristic level, whether our model exhibits many-body localization in the strict sense (see also Section \ref{sec: sketch of the proof} and the recent paper \cite{Schiulaz, deroeckhuveneerssearch}), that is, whether the conductivitiy $\ka(\be)=0$ for $\be<\be_c$ with $\be_c >0$, or whether the localization is only asymptotic as in \eqref{eq: asymptotic localization}, i.e. 
\beq  \label{eq: asymptotic localization quantum}
\lim_{\be \to 0} \be^{-n} \ka(\be) =0, \qquad \text{for any}\,\, n>0
\eeq

In this paper, we give a strong indication why at least \eqref{eq: asymptotic localization quantum} should hold, even in higher dimensions $d>1$, see Theorem \ref{thm: vanishing conductivity}.  
This is done by approximating the current-current correlation function by truncation at times that grow like an arbitrary polynomial in $\be^{-1}$ and proving \eqref{eq: asymptotic localization quantum} for these approximations.    We refer to Section \ref{sec: sketch of the proof} for a more detailed overview of the main ideas. 

Similar reasoning was developed earlier in \cite{huveneers} for disordered classical systems, and in \cite{deroeckhuveneers}, for classical systems where the setup is analogous to the present paper, i.e.\ disorder is replaced by thermal fluctuations.

\subsection{Outline of the paper}

In Section \ref{section: Model}, we introduce the model in precise terms and we state our results and in Section \ref{sec: sketch of the proof} we outline the strategy and we present a glossary of the most important symbols used throughout the proof.  
Section \ref{section: Diagonalization} deals with the iterative diagonalization of our Hamiltonian, excluding the resonant configurations (see explanation above).  The sum of all terms that were not treated by iterative diagonalization is called 'the resonant Hamiltonian', indicated by the symbol $Z$. 
Sections \ref{sec: analysis invariant subspaces} and \ref{sec: left right splitting}  contain the analysis of the resonant Hamiltoninian $Z$. As such, they are fully independent of Section \ref{section: Diagonalization} and they form the main part of our work. In Section \ref{sec: proofs}, we finally combine the results of Section \ref{section: Diagonalization} with the analysis of Sections \ref{sec: analysis invariant subspaces} and \ref{sec: left right splitting} to prove our results. In the appendix, we establish exponential decay of correlations at small $\be$ for our model.

\paragraph{Acknowledgements.}

We benefited a lot from discussions with John Imbrie and David Huse and we also thank them for their encouragement regarding this work. 

W.D.R thanks the DFG for financial support and the University of Helsinki for hospitality.
F.H. thanks the University of Helsinki and Heidelberg University for hospitality,
as well as the ERC MALADY
and the ERC MPOES for financial support.

\section{Model and result}\label{section: Model}

\subsection{Preliminaries} \label{sec: preliminaries}
Let $\bigvolume \subset \bbZ^d$ be a finite set. We define the Hilbert space
\beq
\caH_{} := \otimes_{x \in \bigvolume} \ell^2(\bbN) \sim \ell^2(\bbN^{\bigvolume}),
\eeq
i.e.\ at each site there is an infinite-dimensional 'spin'-space.
  For an operator $O$ acting on $\caH_\bigvolume$ we denote by $s(O)$  ('support' of $O$) the minimal set $A$ such that $O= O_A \otimes \lone_{\bigvolume \setminus A}$ for some $O_A$ acting on $\caH_A$, and $\lone_{A'}$ the identity on $\caH_{A'}$ for any $A' \subset \bigvolume$.  We do not distinguish between $O_A$ and $O$, and we will denote them by the same symbol. 

 Let $a,a^*$  be the bosonic annihilation/creation operators on $\ell^2(\bbN)$:
\beq
(a f)(n) = \sqrt{n+1} f(n+1), \qquad  (a^* f)(n+1) = \sqrt{n+1} f(n), \qquad \text{for $n \in \bbN$}
\eeq
We write $a_x,a_x^*$ for the annihilation/creation operators acting on site $x$, and, as announced above, we do not distinguish between $a_x $ and $a_x \otimes \lone_{\bigvolume \setminus \{x\}}$. We also define the number operators
\beq
N_x : = a^*_x a_x
\eeq
The vectors diagonalizing the operators $N_x$ play a distinguished role in our analysis. For a finite set $A$, we define the \emph{phase space} $\Om_A:=  \bbN^{A}$ 
with elements
\beq
\eta = (\eta(x))_{x \in A}, \qquad \eta(x) \in \bbN
\eeq
such that $\caH_A \sim \ell^2(\Om_A)$ and we often use $\eta$ as a label for the function $\delta_\eta$ i.e.\ $\delta_\eta(\eta')=\delta_{\eta,\eta'}$ for $ \eta',\eta \in \Om_A$. 

\subsection{Hamiltonian}\label{sec: ham}
We introduce the Hamiltonian of our model in finite volume $\bigvolume \subset \bbZ^d$ and with free boundary conditions;
\begin{equation} \label{def: hamiltonian}
H_{\bigvolume} 
\; =\; \sum_{x\in \bigvolume } N_x^q +g_1 (a_x + a_x^{*})^2 +  \sum_{x,x' \in \bigvolume, x \sim x'} g_2  (a_x^{*} a_{x'} + a_x a_{x'}^{*})
\end{equation}
where $
N_x \; = \; a_x^{*} a_x$, $x \sim x'$ means that $x,x'$ are nearest neighbours, and the exponent  $q > 2$.   
By standard methods (e.g.\ Kato-Rellich), one checks that $H_\bigvolume$  is self-adjoint on the domain of $\sum_{x\in \bigvolume} N_x^q$.
  The term $g_1 (a_x + a_x^{*})^2$ destroys the conservation of the total occupation number $\sum_{x \in \bigvolume} N_x$.
In the sequel, we will assume that $g_1, g_2 \sim 1$, so that total energy is the only conserved quantity. Nonetheless, all our results remain valid when $g_1$ or $g_2$ vanish.
The  reason  why we find it important to destroy the second conserved quantity is that  similar models with two conserved quantities typically exhibit non-equivalence of ensembles. As explained in \cite{rumpf}, one expects in a microcanonical ensemble equilibrium states where a macroscopic part of the particles  (the total number of particles would correspond  to $\sum_{x \in \bigvolume} N_x$ in our model) is concentrated on a single lattice site.  We want to stress that this type of 'statistical localization' has nothing in common with the localisation mechanism in the present paper. 

To avoid constants later on, we demand that $\str g_1\str, \str g_2 \str \leq 1$. 


\subsection{States} \label{sec: states}

The thermal equilibrium state $\om_{\be, \bigvolume}$ of the system at inverse temperature $\be$ and in finite volume $\bigvolume$ is defined as
\beq
\om_{\be,\bigvolume} (O) =  \frac{\Tr O \e^{-\be H_\bigvolume}  }{\Tr \e^{-\be H_\bigvolume}  }, \qquad    O \in \caB(\caH_\bigvolume)
\eeq
We are interested in the high-temperature regime, where the finite-volume states $\om_{\be,\bigvolume}$ have a unique infinite-volume limit (for, say, $\bigvolume \nearrow \bbZ^d$ in the sense of Van Hove), independent of boundary conditions.  Morally speaking, this results belongs to standard knowledge, but, literally, it does not, because of the infinite one-site Hilbert space. In principle, we deal with this issue in the appendix, but, since we in fact only need exponential decay of correlations, uniformly in $\bigvolume$, we will not explicitly address the construction of the infinite volume state. 
We drop the volume $\bigvolume$ and inverse temperature $\be$ from the notation for the time being, writing simply $\om(\cdot)$.  It is understood that sums over $x,x'$ are always restricted to the volume $\bigvolume$.


\subsection{Currents} \label{sec: currents}
We fix once and for all the vector $e_1= (1,0,\ldots, 0) \in \bbZ^d$ and we study the current  in this direction. 
First, we decompose the Hamiltonian as
\beq
H=\sum_{x } H_x
\eeq
where
\beq
H_x = N_x^q +g_1 (a_x + a_x^{*})^2 + \tfrac{1}{2}\sum_{x': x' \sim x} g_2  (a_x^{*} a_{x'} + a_x a_{x'}^{*}) 
\eeq
We define local current operators $J_x$ by
\beq \label{def: local currents}
J_x =    \i \sum_{x': x'_1 > x_1} [H_{x'},H_{x}]
\eeq
Since the operators $H_x$ act on at most $2d+1$ sites, all $x'$ that contribute a nonzero term to the sum in \eqref{def: local currents} are nearest neighbours of $x$. 
One way to convince oneself that this is a meaningful definition is to consider first the total current through the (restriction of a) hyperplane $\bbH_a= \{x \in \bigvolume:\, x_1=a\}$ as the time-derivative of  the total energy to the left of this hyperplane, i.e.\
\beq \label{current accross an hyperplane}
J_{\bbH_a}  :=  \i [H, H^{(\links)} ]  = \frac{\d}{\d t}  H^{(\links)}(t)\big\str_{t=0}, \qquad  \textrm{with}\,\,   H^{(\links)}=  \sum_{x: x_1 \leq a} H_x
\eeq
Then it follows that 
\beq
J_{\bbH_a}  = \sum_{x: x_1=a}  J_x 
\eeq
Note that, by the time-invariance of the equilibrium state, $\om(O(t))=\om(O)$, we have
\beq \label{eq: invariance current hyperplane}
 \om(J_{\bbH_a} )=0
\eeq

\subsection{Green-Kubo formula} \label{sec: green kubo}

To study the Green-Kubo formula, we introduce an empiric average of the local current over  space and time:
\beq\label{current appearing in Green Kubo}
\caJ_{\tau} =   \frac{1}{\sqrt{\tau\str \bigvolume \str}} \int_0^\tau  \d t  \sum_{x }   J_x(t)
\eeq
where the scaling anticipates a central limit theorem, relying on the fact that  the equilibrium expectation of $\caJ_{\tau}$ vanishes:
\beq
 \om (\caJ_{\tau} ) =    0.
 \eeq
This follows directly from \eqref{eq: invariance current hyperplane} by using the decomposition  $\sum_{x }= \sum_a \sum_{x:x_1=a}$. 
 We introduce the finite-time conductivity
\beq
 \ka_{\tau}(\be)= \be^2 \lim_{\bigvolume \nearrow \bbZ^d} \om(   \caJ^{*}_{\tau }  \caJ_{\tau} )
\eeq
A basic intuition in transport theory states that in systems with normal (diffusive) transport, the current-current correlations decay in an integrable way, resulting in the convergence of the finite-time conductivity to the conductivity $\ka:= \lim_{\tau \to \infty}  \ka_{\tau}$ with $0 <\ka< \infty$.  At present, this has however not been proven in any interacting Hamiltonian system. 
Instead, we study the behaviour of the approximants  $ \ka_{\tau}$ for arbitrarily large $\tau$ (polynomial in $\be^{-1}$) and we show that at all such times, the conductivity vanishes:

\bet[Conductivity in small $\be$ limit] \label{thm: vanishing conductivity} There is a $C>0$ such that
for any $0<n<m-C$, 
\beq
 \lim_{\be \to 0}    \be^{-n} \ka_{\be^{-m}}(\be)      =0
\eeq
\eet
As already explained in the introduction, we take this result as a strong indication that also
\beq \label{eq: conjecture}
 \lim_{\be \to 0}    \be^{-n} \ka(\be)      =0, \qquad \text{for any}\,\, n>0
\eeq
To make this precise, we should understand what type of processes dominate the dynamics after very long times, i.e. superpolynomial in $\be^{-1}$.  In \cite{deroeckhuveneers}, we argued  for models of classical mechanics that in the case that the dynamics becomes chaotic at such large times, the conjecture \eqref{eq: conjecture} is definitely true.   This was done by introducing an energy-conserving stochastic term in the dynamics of arbitrarily small strength and proving that the conductivity (which in that case can be shown to be finite) has the same order of magnitude as the stochastic term.  This is not attempted in the present paper. 
On the other hand, without such a stochastic term, it remains an enormous task to prove that the conductivity is even finite and nonzero, see for example \cite{bonettolebowitzreybellet} for an exposition of this problem. 

An alternative way to view our results, is to compare them to Nekhoroshev estimates in classical systems. Such estimates typically establish results very reminiscent of ours, but they are restricted to a finite number of degrees of freedom. We refer to \cite{deroeckhuveneers} for a more thorough discussion of this point and for relevant references. 

\subsection{Splitting of the current} \label{sec: splitting of the current}

From a technical point of view, the key result in this paper is  a splitting of the current  $J_{\bbH_a} $ into an oscillatory part and a small part. 
To describe it, let us introduce a  multi-dimensional strip (whose width is called $2r^2$) containing the  hyperplane $\bbH_a$;
\beq \label{def: strip first}
\bbS_{a,r^2}:= \{x \in \bigvolume: \str x_1 -a \str < r^2 \}
\eeq
and we often drop the parameters by simply writing $\bbS=\bbS_{a, r^2}$. 
\bet[Splitting of current] \label{thm: splitting of current}
For any $r>0$, and sufficiently small $\be$, depending on $r$, the following holds uniformly in the volume $\bigvolume$ and the choice of $a$:
There are collections of operators $(O_A)_{A \subset \bbS_{a,r^2}}, (I_A)_{A \subset \bbS_{a,r^2+2}}$, such that
\beq\label{explicit splitting of the current}
J_{\bbH_a}= \sum_{A \subset \bbS_{a,r^2}} \i [H, O_A]  +     \sum_{A \subset \bbS_{a,r^2+2}} I_A
\eeq
and 
\begin{enumerate}
\item  The operators $O_A$ and $I_A$ are supported in $A$, i.e.\ $s(I_A), s(O_A) \subset A$, and  $O_A=I_A=0$ whenever $A$ is not connected. 
\item   $O_A$ and $ I_A$ have zero average: $\om(O_A)= \om(I_A)=0$
\item  They  are bounded as
\begin{align} \label{eq: hs norms}
 \om(O^*_{A} O_{A})   \leq    C(r) \be^{-C +c(r)\str A \str  }, \qquad  \om(I^*_{A} I_A )     \leq    C(r) \be^{-C + cr  +c(r)\str A\str } 
\end{align}
\end{enumerate}
Here, $C,c$ denote constants with $C < \infty, c>0$ that depend only on the dimension $d$, and the exponent $q$. The parameters $C(r),c(r)$ can additionally depend on $r$. 
\eet
The relevance of this theorem in establishing asymptotic energy localization is explained in more details below. 



\section{Overview of the method}\label{sec: sketch of the proof}

Before embarking into the proof of our results, let us informally describe the main steps leading to them. 
Let us first observe that Theorem \ref{thm: vanishing conductivity} is readilly deduced from Theorem \ref{thm: splitting of current}, as detailed in Section \ref{sec: proof of main theorem}.
Indeed, to start with, the first sum in the right hand side of \eqref{explicit splitting of the current} just represents local energy oscillations; 
the contribution of such an oscillation to the current \eqref{current appearing in Green Kubo} is given by
\begin{equation*}
\frac{1}{\sqrt\tau}\int_0^\tau i [H,O](t) \d t = \frac{O (\tau) -  O(0)}{\sqrt\tau} \;\;\rightarrow\;\; 0 \qquad \text{as} \qquad \tau\;\; \rightarrow\;\; \infty. 
\end{equation*}
Next, the terms in the second in sum in the right hand side of \eqref{explicit splitting of the current} possibly contribute to the conductivity, 
but are very small in the Hilbert-Schmidt norm $\norm I_A \norm_\om:= \om(I^*_{A} I_A )^{1/2} $ based on the thermal state $\om$.
In fact, they are seen to decay as an arbitrary large power in $\beta$, if $r$ is taken large enough, 
thanks to the presence of the term `$cr$' in the exponent of the bound in \eqref{eq: hs norms}. 
Fianally, the terms $c(r)\str A\str $ in the exponents in \eqref{eq: hs norms} ensure that we can perform sums over the connected sets $A$.

We can thus now focus on the derivation of Theorem \ref{thm: splitting of current}. 
Let us start by explaining the origin of the oscillatory term in \eqref{explicit splitting of the current}.
For the sake of the argument, let us consider a strongly localized solid. So we imagine that
the unitary change of basis $U$ that diagonalizes $H$ is written as $U= \e^{-K}$, where the anti-hermitian matrix $K$ is a sum of almost local terms 
(see Section \ref{subsec: Interaction potentials} for a precise definition of what almost local means). 
The diagonalized Hamiltonian $\Delta$ then takes the form 
\begin{equation}\label{Delta operator}
\Delta \; = \; U^\dagger H U \; = \;  \sum_x \Delta_x \; = \;  \sum_x \Big\{ f_1(N_x) + f_2(N_x,N_{x+1}) + f_3(N_{x-1},N_x,N_{x+1}) + \dots  \Big\}
\end{equation}
where the terms $f_k$ quickly decay to $0$ as $k\rightarrow \infty$ (we took $d=1$ here for simplicity). 
We now could say that $H^{(L)}$ defined in \eqref{current accross an hyperplane} was the naive left part of the total energy. We define
\begin{equation*}
\widetilde{H}^{(L)} \; = \; U \Delta^{(L)} U^\dagger \qquad \text{with} \qquad \Delta^{(L)} \; = \; \sum_{x:x_1\le a} \Delta_x. 
\end{equation*}
But then, from \eqref{current accross an hyperplane}, we find
\begin{equation}\label{the relation between two currents}
J_{\mathbb H_a} \; = \; i [H,H^{(L)}] \; = \; i [H,H^{(L)} - \widetilde{H}^{(L)}] + i [H,\widetilde{H}^{(L)}]. 
\end{equation}
On the one hand, the locality properties of $U$ allow to conclude that $H^{(L)} - \widetilde{H}^{(L)}$ is localized near the hyperplane $\mathbb H_a$, 
so that the first term in this last equation may be identified with the first sum in the right hand side of \eqref{explicit splitting of the current}. 
On the other hand $[H,\widetilde{H}^{(L)}]$ would here vanish. 
In reality, we will however not be able to fully diagonalize $H$, so that a ``rest term" $\sum_A I_A$ appears in \eqref{explicit splitting of the current}. 
Technically, this step consisting in deriving Theorem \ref{thm: splitting of current}
once the change of basis $U$ and the opertaor $\Delta$ are known, is performed in Sections \ref{sec: splitting of ham in strip}-\ref{subsec: Classification of current operators}. 
This leads to heavy computations, as the operator $\Delta$ that we manage to obtain is far less simple than \eqref{Delta operator}; 
this issues from both conceptual (resonances) and technical questions (high energies).  

We now need to find a change of  basis $U$ that will remove as much oscillations as possible, and then analyze the Hamiltonian in the new basis. 
The construction of the change of basis is performed in Section \ref{section: Diagonalization} 
(the notation $U$ does not appear yet in Proposition \ref{thm: renormalization}; it only shows up in Section \ref{sec: unitary trafo} when we restrict our attention to finite volumes). 
As already stressed in the introduction, the interaction between atoms can be treated as a perturbation at high temperature, thanks to the choice $q > 2$ in the Hamiltonian \eqref{eq: our model}: 
resonances are only met in some exceptional places in the solid (see figure \ref{figure: resonant and non resonant interaction}). 
To make this a bit more transparent at this level of the discussion, we can rewrite $H$ given by \eqref{eq: our model} as 
\begin{equation}  \label{def: tilde g}
H \; = \; E_0 \, + \, \tilde{g} V \qquad \text{with} \qquad E_0 \; = \; \sum_{x}N_x^q, \qquad \tilde{g} = \beta^{1- 2/q}g, \qquad V \; = \; \beta^{-1+ 2/q}\sum_x  (a_x^* a_{x+1} + a_x a^{*}_{x+1}).
\end{equation}
With these notations, both typical self-energy differences and terms in $V$ are of order $\beta^{-(1-1/q)}$, so that $\tilde g$ is indeed a perturbative parameter.  
We will however not explicitly make use of these notations in the proofs.

\begin{figure}[h!]

\vspace{0.5cm}

\begin{center}

\begin{tikzpicture}[scale=1]
\draw [ultra thick] (-2,1.5) -- (-1,1.5);
\draw [ultra thick] (-1,0.5) -- (0,0.5);
\draw [ultra thick] (0,3) -- (1,3);
\draw [ultra thick] (1,0) -- (2,0);
\draw [ultra thick] (2,2) -- (3,2);
\draw [ultra thick] (3,1) -- (4,1);
\draw [ultra thick] (4,1.5) -- (5,1.5);
\draw [ultra thick] (5,2.5) -- (6,2.5);

\draw [>=stealth,->] (-0.5,0.5) -- (-0.5,0.9);
\draw [>=stealth,->] (0.5,3) -- (0.5,2.6);
\draw [>=stealth,->] (3.5,1) -- (3.5,1.4);
\draw [>=stealth,->] (4.5,1.5) -- (4.5,1.1);

\draw [ultra thick, dashed] (-1,1) -- (0,1);
\draw [ultra thick, dashed] (0,2.5) -- (1,2.5);
\draw [ultra thick, dashed] (3,1.5) -- (4,1.5);
\draw [ultra thick, dashed] (4,1) -- (5,1);


\draw [>=stealth,->] (-2.2,-0.5) -- (6.3,-0.5);
\draw [>=stealth,->] (-2,-0.7) -- (-2,4);

\draw (6,-0.6) node[below]{$x$} ;
\draw (-2,3.6) node[left]{$N_x$} ;
\end{tikzpicture}
\end{center}

\caption{
\label{figure: resonant and non resonant interaction}
Resonances in first order in perturbation. For simplicity we assume $d=1$. 
The situation on the left is typical at high temperature, and non-resonant, as the self-energy difference is much larger than the interaction energy: 
$\big( N_x^q + N_{x+1}^q \big) - \big( (N_x +1)^q + (N_{x+1} - 1)^q \big) \; >> \; g \sqrt{N_x (N_{x+1} - 1)}.$
The situation on the right is rare and resonant: the self-energy difference even vanishes in this case.  
}

\end{figure}
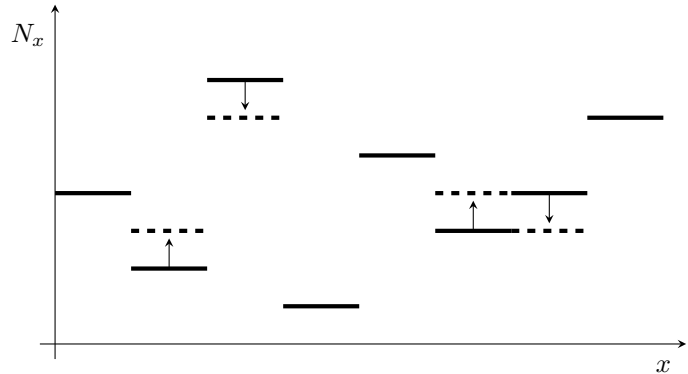

We construct $U$ via an iterative KAM-like scheme, recently developed by Imbrie and Spencer \cite{imbriespencer} in the contex of quenched disordered systems. 
Naively, the scheme works as follows. 
In a first step, we determine $U$ so that $H' := U^\dagger H U$ takes the form $H'= E'_0 + \tilde{g}^2 V'$, for some new self energy $E_0' = E_0 + \mathcal O(\tilde{g}^2)$ and some new perturbation $V'$. 
For this, we write $U = \e^{-K}$ and, assuming that $K$ is a sum of local terms of order $\tilde{g}$, we expand $U^\dagger H U$ in powers of $\tilde g$: 
\begin{equation}\label{naive expansion renormalization}
U^\dagger H U \; = \; \e^{K} (E_0 + \tilde g V) e^{-K} \; = \; E_0 \, +  \, \big( \tilde g V + [K,E_0] \big) \, + \, \mathcal O (\tilde g^{2}). 
\end{equation}
Writing $V = \sum_x V_x$ and $K = \sum_x K_x$, we get rid of the first order in $\tilde{g}$ by setting 
\begin{equation}\label{cohomological equation}
\langle \eta | K_x | \eta' \rangle \; = \; \tilde{g} \,\frac{\langle \eta | V_x | \eta' \rangle}{E_0(\eta) - E_0(\eta')} \qquad \Rightarrow \qquad  \tilde g V + [K,E_0] \; = \; 0
\end{equation}
($\langle \eta | K_x | \eta' \rangle=0$ if $\langle \eta | V_x | \eta' \rangle=0$ by definition). 
The fact that resonances are rare precisely means that for most of the pairs of states $\eta$ and $\eta'$, the matrix element $\langle \eta | K_x | \eta' \rangle$ is well defined and of order $\tilde{g}$.
Let us ignore resonances for the moment. 
We would then conclude that the expansion \eqref{naive expansion renormalization} was justified. 
So, we would also have determined a renormalized Hamiltonian $H'$ with a perturbation of order $\tilde{g}^2$. 
This strategy could then be iterated, constructing a sequence of change of variables $\e^{K^{(j)}}$, where $K^{(j)}$ is a sum of local terms of order $\tilde{g}^{2^{j-1}}$.
Doing so, we would readilly conclude that $H$ is strongly localized in the sense of \eqref{Delta operator}. 

As an aside, let us observe that it is only possible to find $K$ so that \eqref{cohomological equation} holds if the pertubation $V$ has no diagonal element.
This explains that in principle we need to renormalize the self-energy ($E_0' = E_0 + \mathcal O(\tilde{g}^2)$ as mentionned), 
in order to absorbe all the diagonal part of the new Hamiltonian. 
It is however not what we do in practice, as we simply treat these extra diagonal elements as resonances. 
However, in order to investigate possible true localization in translation-invariant models, there would be a deeper reason to take account of the self-energy renormalization. 
Indeed, if this phenomenon is ignored, it is readilly seen that, for fixed $\beta$, resonances would eventually become typical as one goes on to higher orders in perturbation theory. 
Since, at higher orders, atoms may change level by more than one unit, the interaction could now just swap the levels of any two near atoms
(whereas at the first order this was only possible if the energies were nearly the same as depicted on the right interaction of figure \ref{figure: resonant and non resonant interaction}).
So all atoms would be in resonance with their neighbours, allowing energy to travel into the solid.  
On the other hand, such a drastic conclusion could not be reached if the renormalization of the self-energy was taken into account. 
It has in fact been suggested in \cite{Schiulaz} that this effect could guarantee that resonances rarefy as one moves to higher orders. 
To support this view, 
we indeed observe that the perturbative splitting of the levels could and should be exploited to show localization in the one-body Anderson model when the disorder only takes a finite numner of values, 
a model for which localization is clearly expected to hold. 

Let us come back to the description of the scheme initiated in \eqref{naive expansion renormalization}. 
It is clear that resonances, even if very rare, cannot just be ignored as we pretended up to now. 
We just do as much as we can: the perturbation $V$ is splitted into a resonant and non-resonant term (see \eqref{resonant non resonant splitting}), 
and \eqref{cohomological equation} is only solved with $V$ replaced by the non-resonant part of $V$.
While the change of variables $\e^{K^{(j)}}$ are now well defined and enjoy good decay properties, this replacement comes with a price.
A first, technical, consequence is that the speed of the iteration procedure is much slowed down.     
Indeed, in this version of the scheme, we just let the resonant part as it is, so that at each step, resonant terms of order $\tilde{g}$ are present in the perturbation. 
Though they do not create any trouble as such, it is seen that, itarating the scheme once more, 
non-resonant terms are generated that would be too large for a superexponential bound like $\tilde{g}^{2^{j-1}}$ to survive. 
Instead, we can only obtain that $K^{(j)}$ is a sum of terms of order $\tilde{g}^j$ (so we do not progress faster than in usual perturbation theory).

The true problem is however that, after a large but finite number of iterations, 
we are left with a Hamiltonian containing still a perturbation of order $\tilde{g}$ 
(see the term $\scrG^{(r)}$ in \eqref{eq: change of variables}, and, later on, the resonant Hamiltonian $Z$ defined in \eqref{eq: def z}).
The resonant Hamiltonian is well sparse, but not as much as needed to get our results: 
a look at figure \ref{figure: resonant and non resonant interaction} shows indeed that the probability of two atoms to be resonant is at best bounded by $\beta^{1/q}$.  
Before indicating how we will get off the hook, let us stress here that the analysis of resonances reveals a fundamental difference between quenched and thermal disorder.   

To see this, let us for example consider the first order resonances in a quenched disorder spin chain, as studied by \cite{palhuse} \cite{imbriespencer}. 
In this model, it is possible determine bonds on the lattice such that resonances can only occur on these bonds. 
Moreover, if the disorder is strong enough, these potentially resonant bonds form small isolated islands. 
In this case, it is then in fact possible to completely get rid of the resonant Hamiltonian at each step of the procedure. 
Indeed, one can diagonalize the Hamiltonian ``on the resonant islands", meaning that we conjugate it with a change of basis that affects only the terms in $H$ that act inside the islands.  
This rotation is non-perturbative, but does not entail any delocalization, as the resonant spots do not percolate. 
At the opposite, in the translation invariant set-up, it is no longer possible to visualize resonances on the physical lattice. Instead, we directly need to analyze a percolation problem in the full set of states
(it should however be noticed that the eigenstates of the resonant Hamiltonian could still be localized even in the presence of a giant percolation cluster, 
but we are not aware of any convincing argument supporting this view). 
This is a rather delicate problem, illustrated on figure \ref{figure: motion flemish mountain}.

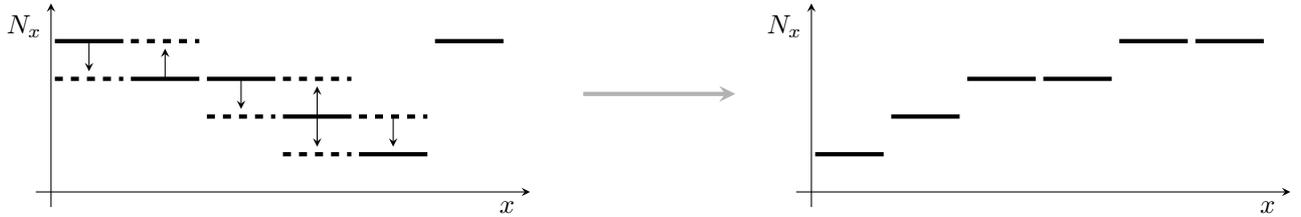
\begin{figure}[h!] 

\vspace{0.5cm}

\begin{center}

\begin{tikzpicture}[scale=1]
\captionsetup{singlelinecheck=off}
\draw [ultra thick] (0.05,2) -- (0.95,2);
\draw [ultra thick] (1.05,1.5) -- (1.95,1.5);
\draw [ultra thick] (2.05,1.5) -- (2.95,1.5);
\draw [ultra thick] (3.05,1) -- (3.95,1);
\draw [ultra thick] (4.05,0.5) -- (4.95,0.5);
\draw [ultra thick] (5.05,2) -- (5.95,2);

\draw [ultra thick,dashed] (0.05,1.5) -- (0.95,1.5);
\draw [ultra thick,dashed] (1.05,2) -- (1.95,2);
\draw [ultra thick,dashed] (2.05,1) -- (2.95,1);
\draw [ultra thick,dashed] (3.05,1.5) -- (3.95,1.5);
\draw [ultra thick,dashed] (3.05,0.5) -- (3.95,0.5);
\draw [ultra thick,dashed] (4.05,1) -- (4.95,1);

\draw [>=stealth,->] (0.5,2) -- (0.5,1.6);
\draw [>=stealth,->] (1.5,1.5) -- (1.5,1.9);
\draw [>=stealth,->] (2.5,1.5) -- (2.5,1.1);
\draw [>=stealth,->] (3.5,1) -- (3.5,1.4);
\draw [>=stealth,->] (3.5,1) -- (3.5,0.6);
\draw [>=stealth,->] (4.5,1) -- (4.5,0.6);

\draw [>=stealth,->] (-0.2,0) -- (6.3,0);
\draw [>=stealth,->] (0,-0.2) -- (0,2.5);

\draw[>=stealth,->,ultra thick,color=gray!60] (7,1.3) -- (9,1.3);

\draw[ultra thick] (10.05,0.5) -- (10.95,0.5);
\draw[ultra thick] (11.05,1) -- (11.95,1);
\draw[ultra thick] (12.05,1.5) -- (12.95,1.5);
\draw[ultra thick] (13.05,1.5) -- (13.95,1.5);
\draw[ultra thick] (14.05,2) -- (14.95,2);
\draw[ultra thick] (15.05,2) -- (15.95,2);

\draw [>=stealth,->] (9.8,0) -- (16.3,0);
\draw [>=stealth,->] (10,-0.2) -- (10,2.5);

\draw (16,0) node[below]{$x$} ;
\draw (10,2.2) node[left]{$N_x$} ;

\draw (6,0) node[below]{$x$} ;
\draw (0,2.2) node[left]{$N_x$} ;

\end{tikzpicture}
\end{center}
\captionsetup{singlelinecheck=off}
\caption[foo bar]{\label{figure: motion flemish mountain}
In translation invariant chains, resonances do travel into the system.  Let us assume that next to nearest neighbor level swapping is allowed (which anyway occurs in second order in perturbation). More precisely, this means that a configuration
 $\ldots, N_x,N_{x+1},N_{x+2},\ldots $ can be transformed into
  \begin{enumerate}
    \item[$i.$] $\ldots,N_{x+1}, N_x,N_{x+2},\ldots \quad $ if $\str N_{x+1}-N_{x} \str \leq 1$
    \item[$ii.$] $\ldots, N_x,N_{x+2},N_{x+1},\ldots \quad$ if $\str N_{x+2}-N_{x+1} \str \leq 1$ 
       \item[$iii.$]  $\ldots,N_{x+2},N_{x+1}, N_x\ldots \quad$ if $\str N_{x+2}-N_{x} \str \leq 1$
  \end{enumerate}
 With a bit of trial and error, we discover that the left configuration can be transformed into the right configuration in a few steps.  This means that 
the time evolution of the state on the left under the dynamics generated by the resonant Hamiltonian can have an overlap with the state on the right. 
We see that the most right atom can enter in resonance with the other ones, though it was not initially so. 
  }
  

\end{figure}

We will not attempt to diagonalize the resonant Hamiltonian. 
Instead, the total energy will be separated into a left and right part, in a state dependent way, by a surface close to $\mathbb H_a$ that  ``slaloms" between the resonances. This is described in Section \ref{sec: left right splitting}, see in particular Figure \ref{fig: leftright1} where the spirals indicate the resonant spots. 
So, we will  arrive in the situation described  by \eqref{the relation between two currents}: 
the second term in the right hand side of this equation will now be sufficiently sparse current, while the first term still is just an oscillation. 

To see how to define this surface, we need to analyze the motion of resonances (see Section \ref{sec: analysis invariant subspaces}). 
Let us first restrict the Hamiltonian to a large but fixed volume $V$ around a point on $\mathbb H_a$ (a volume that will not be sent to infinity).
We show the following. Let us pick up a state $\eta$ in $V$, 
and let us collect all the other states in $V$ that could have an overlap with the time evolution of $\eta$ under the dynamics generated by the resonant Hamiltonian. 
We show that for an overwhelming majority of states $\eta$, there exists small isolated islands in $V$ such that any of the state that we have collected, only differ from $\eta$ on these islands. 
The set of states for which this does not hold is small enough to be neglected. 
On the one hand, we can convince ourselves of the validity of this statement by looking at figure \ref{figure: resonant and non resonant interaction}. 
To simplify, let us assume that resonances are first order, and only occur when two levels are swapped as it is the case for the interaction on the right. 
Then on that example, it is seen that the only resonant island is located on the sites 5,6,7, assuming that atoms have been labeled from 1 to 8.  
On the other hand, a look at figure \ref{figure: motion flemish mountain} hints that this statement could be violated if $V$ was sent to infinity for fixed $\beta$. 
Indeed, as the volume gets larger and larger, configurations that are rare locally, eventually occur. 
It is thereofre concivable that a big resonant spot starts invading the full space, connecting configurations that would have remained separted if the perturbation was confined to the volume $V$.  
 
So we have found a way to construct the surface close to $\mathbb H_a$ in the volume $V$, 
but this is not completely satisfactory as we take the thermodynamic limit $\Lambda \rightarrow \infty$ before sending $\beta\rightarrow 0$.
Two issues are raised. 
First, if the dimension is larger than one, we may take a volume $V$ around each point in $\mathbb H_a$ and construct a piece of surface in each of these volumes, 
but we then have to glue them together.  
Second, even in one dimension, where the surface just reduces to a single point, we must analyze what extra-current is produced if the Hamiltonian is now defined on the full space. 
Let us bypass here the first question, that leads to intricate constructions (see Section \ref{sec: left right splitting}), as the second one appears to us as more fundamental. 
We actually observe that the set of states for which an extra current is produced when reintroducing the interaction at the border is extremely small. 
Indeed, a non zero current could only be created if a small energy change at the border, induced by the perturbation, could completely modify the island picture up to the center of $V$. 
However, in most cases, the configuration of the islands is far less fragile: 
a very atypical configuration would be required for a single change at the border to propagate in the bulk of $V$ 
(too few atoms appear on figure \ref{figure: resonant and non resonant interaction} to see this neatly, but one can be readily become convinced by adding a few sites). 
We thus see that the current is indeed very sparse. 

This summarizes most of the conceptual points addressed in this article. 

\section*{Glossary}
Here is an overview of symbols that appear in different parts of the article (excluding the appendix). The middle column gives the page where the symbol appears for the first time. \vspace{0.5cm}\\
\noindent Potentials (script fonts: $\scrA,\scrB,\ldots$); 
\begin{center}
  \begin{tabular}{ p{2.3cm} p{1cm} p{10cm} }
    $\scrE (\scrE_0)$& \pageref{def: full potential}&Potential of the model Hamiltonian (without interaction). \\ 
   $ \scrF,\scrG, \scrD$   & \pageref{eq: change of variables}& Renormalized  potential:  nonresonant, resonant, diagonal.\\ 
   $ \hat \scrF,\hat \scrG$   & \pageref{eq: hat decomp} & Finite-range approximations to renormalized potentials. \\ 
    \end{tabular}
\end{center}
Operations on potentials (Calligraphic fonts);
\begin{center}
  \begin{tabular}{ p{2.3cm} p{1cm} p{10cm} }
    $\caP_{\leq M}, $   &  \pageref{def: volume cutoff}   & Cutoff in occupation number.   \\ 
      $\caP_{\mathrm{Res}},\caP_{\mathrm{NRes}} $  & \pageref{resonant non resonant splitting}  &Projection onto resonant, nonresonant parts. \\ 
    $\caK$      & \pageref{def: total trafo potentials}& Total renormalization transformation. \\
        $\caD$      & \pageref{def: diagonal restriction}& Restriction to diagonal. \\
      $\caI_A$   & \pageref{def: volume restriction potentials} &Restriction to volume $A$. \\
      $\caR_n$    & \pageref{def: range restriction potentials} &Restriction to range $n$. \\
    \end{tabular}
\end{center}
Notions from the analysis of the resonant Hamiltonian, for configurations $\eta$ and components $\mu$;
\begin{center}
  \begin{tabular}{ p{2.3cm} p{1cm} p{10cm} }
    $\caP= \caP^{(V)} $   & \pageref{def: moves}  &The set of moves, in volume $V$. \\ 
       $\caP_A(\eta), \caP_A(\mu) $   & \pageref{def: partition} &   Moves with support in $A$ that are active from $\eta$, $\mu$.  \\ 
      $\caP'_A(\eta)  $   & \pageref{def: capprime} &Moves with support in $A$  that are not too far from $\eta$ to be active.\\ 
       $\caP''(\eta)  $   & \pageref{def: capdoubleprime}  &Slight modification  of $\caP'(\eta)$. \\  
      $\caF^{(V)},\caF^{(y)} $   & \pageref{def: partition} &Partition of phase space in volume $V$, $B_y$ into components $\mu$. \\
            $\links(\mu),\rechts(\mu) $   & \pageref{def: left right} &Left, right regions depending on component $\mu \in \caF$. \\
                        $Z_\links,Z_\rechts$   & \pageref{def: z links and z rechts} &Left, right resonant Hamiltonian. \\
                                                $U_A$   & \pageref{eq: cak and unitary} &Unitary restriction of transf.\ $\caK$ to volume $A$.  \\
                                                                                      $B_y, \tilde B_y$   & \pageref{def: balls}&Balls (within $\bbS$) centered at $(a,y)$.  \\
    \end{tabular}
\end{center}
Important parameters; 
\begin{center}
  \begin{tabular}{ p{2.3cm} p{1cm} p{10cm} }
    $\delta $   &  \pageref{thm: renormalization} &resonance threshold, set to $M^{-\ga_1}$ in \eqref{eq: def of m via gamma}. \\ 
      $M  $   &  \pageref{def: volume cutoff}  & occupation cutoff, 
set to $M= \be^{-(1+c(r))/q}$ in Thm.\ \ref{thm: splitting of hamiltonian}.  \\ 
    $\gamma_1,\gamma_2$      &  \pageref{eq: def of m via gamma},  \pageref{prop: reduced mobility} & Exponents of $M$. \\
    \end{tabular}
\end{center}
Norms, with $\ka, \ka_j \geq 1$ and $\nu$ a state (density matrix);
\begin{center}
  \begin{tabular}{ p{2.3cm} p{1cm} p{10cm} }
    $\str \cdot \str $   & \pageref{def: range restriction potentials}& Euclidian norm.  \\ 
        $\norm \cdot \norm $   & \pageref{def: def kappa operator norm}& operator norm. \\ 
                $\norm \cdot \norm_\kappa $   & \pageref{def: def kappa operator norm} & (non standard) weighted operator norm.  \\ 
                $\interleave \cdot \interleave_{\kappa_1,\kappa_2}, \interleave \cdot \interleave_{\kappa}    $   & \pageref{def: interleave norms} &weighted potential norm.   \\
                                $\norm \cdot \norm_{\nu} $   & \pageref{def: hs norm trace} &Hilbert-Schmidt norm from scalar product $\langle A, B \rangle= \nu(A^*B)$.  \\ 
    \end{tabular}
\end{center}

%
%
%
%

\section{Perturbative diagonalization of $H$}\label{section: Diagonalization}

In this section, we introduce the formalism of interaction potentials and we implement an iterative diagonalization scheme, acting on interaction potentials.

 \subsection{Energy cutoff} \label{sec: energy cutoff}
 In our analysis, we find it convenient to introduce a high-energy cutoff, even though, in principle, the main reasoning of the paper is the more applicable, the higher the energy.
Given a number $M >0$ and an operator $O$ with finite range $s(O)$, we set  
\beq \label{def: volume cutoff}
\caP_{\leq M}(O) :=    \left(\mathop{\otimes}\limits_{x \in s(O)}  \chi(N_x \leq M) \right)  O    \left(\mathop{\otimes}\limits_{x \in s(O)}  \chi(N_x \leq M) \right) 
\eeq
and, analogously, we define  $\caP_{> M}(O)$ by replacing $N_x \leq M$ by $N_x >M$.  Note that in general, $O\neq \caP_{> M}(O)+ \caP_{\leq M}(O)$.
The cutoff will be chosen, at the end of the analysis, to be $M = \be^{-(1+\gamma_c)/q}$, for some small $\gamma_c>0$
 
 \subsection{Interaction potentials}\label{subsec: Interaction potentials}
The Hamiltonian $H$ is strictly local, i.e.\ it is a sum of terms that act  on at most  two lattice sites. When performing an iterative diagonalization, this will no longer be true and hence we first introduce a weaker notion of locality by introducing \emph{interaction potentials}. 
\begin{definition}\emph{
An interaction potential $\scrA$ is a map from finite, connected sets $ A \subset \bbZ^d$ to bounded operators $\scrA(A)$ on $\caH_A$. 
A Hamiltonian in finite volume ${V}$ associated to a potential $\scrA$ is defined by 
\beq  \label{eq: def sum potential}
\pot_{V}(\scrA) = \sum_{A \, \text{connected}:A \subset {V}} \scrA(A)
\eeq 
For simplicity, we henceforth assume that, for any interaction potential $\scrA$,  $\scrA(A)=0$ if $A$ is not connected and we omit the restriction to connected $A$ from  sums like \eqref{eq: def sum potential}. }
\end{definition}
In the literature, one almost always uses the notation $H_{{V}}(\scrA)$ but we have chosen $\pot_{V}(\scrA) $ to avoid confusion with the Hamiltonian $H_\bigvolume$ defined in \eqref{def: hamiltonian}.
Obviously, the denomination 'Hamiltonian'  is a misnomer in case the operators $\scrA(A)$ are not Hermitian. 
For a potential $\scrA$, we define the cutoff potential
\beq
\left(\caP_{\leq M} (\scrA)\right)(A) :=  \caP_{\leq M} \left(\scrA(A)\right)
\eeq
and analogously for $\caP_{>M}(\scrA)$. An important example of a potential is the potential  $\scrE$ specifying our model Hamiltonian itself, with an energy cutoff. It is defined by 
\beq \label{def: full potential}
\scrE(A) := \begin{cases} \caP_{\leq M} (H_x)  &   \text{if}\,\,     A=\{x' \in \bigvolume: \str x'  -x \str \leq 1\} \,\, \text{for some $x$}     \\  0  & \text{otherwise} \end{cases}
\eeq
We also define the potential of the free Hamiltonian 
\beq\label{eq: def free potential}
\scrE_0(\{x\}) = \caP_{\leq M} (N^q_x), \qquad \text{and} \qquad \scrE_0(A)=0, \quad \text{whenever}\quad \str A \str >1. \eeq
so that indeed
 \beq \pot_\bigvolume(\scrE)= \sum_{x \in \bigvolume }  \caP_{\leq M} (H_x), \qquad   \pot_\bigvolume(\scrE_0) =  \sum_{x \in \bigvolume}  \caP_{\leq M} (N^q_x).   \eeq 
 Note however that other choices are possible for $\scrE$; different potentials can define the same Hamiltonian.

 \subsubsection{Norms} \label{sec: norms}
 
Note that interaction potentials form a linear space under the addition $(\scrA+\scrA')(A):= \scrA(A)+\scrA'(A)$. 
We introduce a family of suitable norms on interaction potentials, based on the following weighted operator norms: For an operator  $O $ on $ \caH_A$,  we define an associated operator $\breve O $ on $ \caH_A$ by 
\beq
\langle \eta, \breve O  \eta' \rangle :=   \str \langle \eta, O  \eta' \rangle\str, \qquad \eta,\eta' \in \Om_{A}
\eeq
such that, in particular, $\norm O \norm \leq \norm \breve O \norm$ where $\norm \cdot \norm$ is the standard operator norm. Further, for $\ka > 1$,  we set
\beq \label{def: def kappa operator norm}
\norm O \norm_{\ka} :=     \sup_{\substack{w \in \bbR_+^{A} \\ \ka^{-1} \leq w(x) \leq \ka }}\norm w^N \breve O w^{-N}  \norm,\qquad \text{with}\,\,  w^N= \prod_{x \in A}    w(x)^{N_x}
\eeq
 For $\ka=1$, we define simply $\norm O \norm_{1} := \norm O \norm$ and we note that 
\beq
\norm O \norm_{\ka'}  \leq \norm O \norm_{\ka}, \qquad \text{for} \, \, 1 \leq \ka' \leq \ka 
\eeq 
 Note that these definitions are independent of $A$ provided $s(O) \subset A$.  For $\ka>1$, the $\norm \cdot \norm_{\ka} $-norm penalizes off-diagonal elements in the number basis.  The corresponding class of norms on interaction potentials is
\beq \label{def: interleave norms}
\interleave \scrA \interleave_{\ka_1,\ka_2} :=  \sup_{x \in \bbZ^d}  \sum_{A: A \ni x}    \ka_1^{ \str A \str}  \norm  \scrA(A) \norm_{\ka_2},\qquad   \interleave \scrA \interleave_{\ka} := \interleave \scrA \interleave_{\ka,\ka} 
\eeq
There is no compelling reason to consider  $\ka_1=\ka_2$,  but we often do so for reasons of simplicity. 

\subsection{Operations on interaction potentials}
Given two interaction potentials $\scrA, \scrB$ we define a new potential 
\beq
[\scrA,\scrB] (A) := \sum_{A_1,A_2: A_1 \cup A_2=A}  [\scrA(A_1),\scrB(A_2)]
\eeq
 and we note that every term in the sum on the right hand side vanishes unless $A_1 \cap A_2 \neq \emptyset$. In particular,  if $\scrA,\scrB$ assign zero to every non-connected set $A$, then so does  $[\scrA,\scrB]$.
  The motivation for this definition is of course that, for any volume ${V}$ 
\beq
{\pot}_{V} ([\scrA,\scrB] ) = [{\pot}_{V}(\scrA), {\pot}_{V}(\scrB) ]
\eeq
Often, we prefer to use the notation
\beq
\adjoint_{\scrA}(\scrB) = - \adjoint_{\scrB}(\scrA) = [\scrA,\scrB] 
\eeq
If one imagines that $\i {\pot}_{V}(\scrA)$ is an anti-Hermitian operator and hence that it generates a time evolution, then one might ask how this time-evolution affects a potential $\scrB$. To address such questions, we define
(for the moment as a formal series)
\beq \label{eq: definition unitary trafo}
\e^{ \adjoint_\scrA}(\scrB) := \sum_{n \geq 0}   \frac{1}{n!}   \adjoint^n_\scrA (\scrB) 
\eeq
Provided this series converges (in one of the norms $\interleave \cdot \interleave_{\ka}$), we can conclude that 
\beq \label{eq: unitary trafo}
{\pot}_{V}(\e^{ \i \adjoint_\scrA}(\scrB)) =    \e^{\i {\pot}_{V}(\scrA)} {\pot}_{V}(\scrB)   \e^{ - \i  {\pot}_{V}(\scrA)} 
\eeq 
In particular, for any time $t$, we can consider the time-evolution
\beq
\scrB_t :=  \e^{\i t \adjoint_\scrA} (\scrB) 
\eeq
The intuition that ${\scrB_t}$ is still a bonafide interaction potential, though with range growing with $t$, is captured by the so-called Lieb-Robinson bounds that have received a lot of attention lately \cite{nachtergaele2010lieb}.  In some sense, we rederive such bounds  in the following lemma (in particular $3)$), which helps us to handle multiple commutators of potentials.  We do not require Hermiticity, but we are restricted to small potentials, corresponding to small time $t$ in the setup above. 

\begin{lemma} \label{lem: composing norms}
Let $\ka_1 > \ka_1' \geq 1$ and $\ka,\ka_2 \geq 1$,  let $\scrA,\scrB$ be interaction potentials and let $O_1,O_2$ be bounded operators.  In all inequalities below, both sides can be infinite.
\begin{enumerate}
\item 
\beq \label{eq: factorization norm}
\norm O_1 O_2 \norm_{ \ka}  \leq   \norm O_1 \norm_{ \ka} \norm O_2 \norm_{\ka}
\eeq
\item
\beq
\interleave \ad_{\scrA}(\scrB) \interleave_{\ka'_1, \ka_2}  \leq 4(\log(\ka_1/\ka_1'))^{-1} \interleave \scrA \interleave_{\ka_1,\ka_2} \interleave \scrB \interleave_{\ka_1,\ka_2}\label{eq: bound norms one}
\eeq 
\item If $4(\log(\ka_1/\ka_1'))^{-1}\interleave \scrA \interleave_{\ka_1,\ka_2} <1$, then, for any bounded sequence $\str g(k) \str \leq 1, k \in \bbN$
\beq
\interleave \sum_{k \geq 0}  \frac{g(k)}{k !} \ad^{k}_{\scrA}(\scrB) \interleave_{\ka'_1,\ka_2} \leq  \frac{1}{1-   4(\log(\ka_1/\ka_1'))^{-1}\interleave \scrA \interleave_{\ka_1,\ka_2} }   \interleave \scrB  \interleave_{\ka_1,\ka_2}   \label{eq: bound norms}
\eeq  In particular, by choosing $g(k)=1$, the potential on the left hand side equals  $\e^{ \adjoint_{\scrA}}(\scrB)$. 
\end{enumerate}
\end{lemma}
\begin{proof} Point $1)$ is trivial. To address points  $2),3)$, we introduce some more structure. 
Let us first define, for a function $F \geq 0$ on finite subsets of $\bbZ^d$, the norm on potentials
\beq
\interleave \scrA \interleave_{F} := \sup_{x} \sum_{A: A \ni x}   F(A) \norm \scrA(A) \norm
\eeq 
The following class of functions $F$ will be of relevance:
\begin{align}
F_{m,\ka}(A)  &:=  \str A \str^{-m}  \ka^{\str A \str}, \qquad    m  \geq 0.
\end{align}
We  establish 
\begin{lemma}\label{lem: norm of commutator}
For any $\ka \geq 1$ and $m\geq 0$,
\beq \label{eq: inequality of fs}
\interleave \ad_\scrA(\scrB) \interleave_{F_{m+1,\ka}} \leq 4 \interleave \scrA \interleave_{F_{0,\ka}} \interleave \scrB \interleave_{F_{m,\ka}}
\eeq
\end{lemma}
\begin{proof}
\begin{align}
\interleave \ad_\scrA(\scrB) \interleave_{F_{m+1,\ka}}  & \leq \sup_x      \sum_{A_1: A_1 \ni x}  \sum_{x' \in A_1}    \sum_{A_2: A_2 \ni x'}       F_{m+1,\ka}(A_1 \cup A_2) \left( \norm [\scrA(A_1), \scrB(A_2) ] \norm  +   \norm [\scrA(A_2), \scrB(A_1) ] \norm  \right)  \label{eq: basic manip f} 
\end{align}
To deal with the first term and second term, we dominate, respectively,  
\begin{align}
F_{m+1,\ka}(A_1 \cup A_2) \leq &   F_{0,\ka}(A_1) F_{m,\ka}(A_2)  \str A_1 \str^{-1} \label{eq: manip f1}  \\
F_{m+1,\ka}(A_1 \cup A_2) \leq &   F_{m,\ka}(A_1) F_{0,\ka}(A_2)  \str A_1 \str^{-1}  \label{eq: manip f2}  
\end{align}
and $\norm [\scrA(A), \scrB(A') ] \norm \leq 2 \norm\scrA(A) \norm \norm \scrB(A') \norm  $. The claim follows. 
\end{proof}
In the same spirit, we now estimate, for $1\leq \ka'<\ka$, 
\begin{align}
  \interleave \sum_{k \geq 0}  \frac{g(k)}{k !} \ad^{k}_{\scrA}(\scrB) \interleave_{F_{0,\ka'}}  & \leq   \sum_{k \geq 0}\sup_x \sum_{A \ni x}    \frac{1}{k !}  (\frac{\ka'}{\ka})^{\str A \str}  F_{0,\ka}(A)      \norm (\ad^{k}_{\scrA}(\scrB))(A)  \norm  \nonumber \\ 
  &  \leq   \sum_{k \geq 0} \sup_x \sum_{A \ni x}    ( \log(\ka/\ka') \str A \str)^{-k}   F_{0,\ka}(A)     \norm (\ad^{k}_{\scrA}(\scrB))(A)  \norm \nonumber \\ 
      &  \leq   \sum_{k \geq 0}    (\log(\ka/\ka'))^{-k}    \interleave\ad^{k}_{\scrA}(\scrB)  \interleave_{F_{k,\ka} } \nonumber \\ 
 &  \leq   \sum_{k \geq 0}    (\log(\ka/\ka'))^{-k}   4^k  \interleave \scrA \interleave^{k}_{F_{0,\ka}}   \interleave \scrB  \interleave_{F_{0,\ka} } \nonumber \\ 
  &  \leq   (1-   4(\log(\ka/\ka'))^{-1}\interleave \scrA \interleave_{F_{0,\ka}})^{-1}   \interleave \scrB  \interleave_{F_{0,\ka} }  \label{eq: long manip f}
\end{align}
where the second inequality follows from 
\beq
\sup_{a>0} a^k \e^{- a }  \leq    k !, \qquad k \in \bbN  
\eeq
and the fourth inequality follows by $k$ applications of Lemma \ref{lem: norm of commutator}. 
 
 This means that we have obtained items $2),3)$ for $\ka_2=1$ because $\norm \cdot \norm_{F_{0,\ka}} = \norm \cdot \norm_{\ka,1}  $.  More precisely, for $2)$, take $m=0$ in \eqref{eq: inequality of fs} and use that 
 \beq \interleave \adjoint_{\scrA}(\scrB) \interleave_{\ka',1} \leq  (\log(\ka/\ka'))^{-1}    \interleave \adjoint_{\scrA}(\scrB) \interleave_{F_{1,\ka}}, \qquad \text{for $1 \leq \ka' < \ka$}. \eeq  
 By inspection of the above estimates we see that the reasoning applies just as well with $\ka_2>1$, so that $2),3)$ are proven.

 \end{proof}

 \subsection{Perturbative diagonalization}

 Let us define the cut-off \emph{phase-space}, for finite $A \subset \bbZ^d$ 
 $$\Om_A^{(M)}  =  \{0,1,2,\ldots M\}^{A}, \qquad \text{with $M$ as in Section \ref{sec: energy cutoff}}   $$
Slightly abusing notation, we denote its elements by $\eta,\eta'$ and we recall that they index eigenvectors of the free Hamiltonian $\sum_{x \in A} N^q_x$, with eigenvalues
\beq
 E_A(\eta)= \sum_{x \in A}\langle \eta, N^{q}_x \eta \rangle =    \sum_{x \in A} \eta_x^q
\eeq
Moreover,  we will decompose interaction potentials in \emph{resonant} and \emph{non-resonant} parts.  For this purpose, we fix some small \emph{resonance threshold} $0< \delta < 1$  (that will be related to the cutoff $M$ in Section \ref{sec: analysis invariant subspaces}) and we define
\begin{align*}
\mathrm{Res}_A \; &:= \; \big\{ (\eta,\eta') \in \Om_A^{(M)} \times \Om_A^{(M)}  : |E_A(\eta) - E_A(\eta') | \; \le \; \delta^{-1} M  \big\},\\
\mathrm{NRes}_A \; &:= \; \big\{ (\eta,\eta') \in \Om_A^{(M)} \times \Om_A^{(M)} : |E_A(\eta) - E_A(\eta') |  \; > \; \delta^{-1} M  \big\}.
\end{align*}
and the linear maps on interaction potentials
\beq \label{resonant non resonant splitting}
(\caP_{\mathrm{Res}}(\scrA))({A}) :=  \sum_{(\eta,\eta') \in \mathrm{Res}_A} P_\eta \scrA(A) P_{\eta'},   \qquad  \caP_{\mathrm{NRes}}(\scrA)(A) :=  \sum_{(\eta,\eta') \in \mathrm{NRes}_A}  P_\eta \scrA(A) P_{\eta'}\eeq
where  $P_{\eta} \in \caB(\caH_A)$ is the one-dimensional orthogonal projection on the space spanned by the vector $\eta$, i.e., by $\delta_\eta(\cdot)$, see  Section \ref{sec: preliminaries}.   The following proposition is inspired by \cite{imbriespencer}:
\begin{proposition}[Perturbative diagonalization] \label{thm: renormalization}
For any $r =0,1,2,\ldots$ and sufficiently small $\delta>0$, depending on $r$, we find interaction potentials $\scrF^{(r)} , \scrG^{(r)}, \scrK^{(r)} $  such that
\beq  \label{eq: change of variables}
\e^{ \adjoint_{\scrK^{(r)}} } \ldots \e^{ \adjoint_{\scrK^{(2)}} }\e^{ \adjoint_{\scrK^{(1)}} } \left(\scrE \right)  = \scrE_0+ \scrF^{(r)} +{\scrG^{(r)}}
\eeq
(where the left hand side is understood to be $\scrE$ for $r=0$),
and the following properties hold with
\beq
\nu= \frac{1}{4(2 d +3)}, \qquad   e(r)= (2r-1)/3,
\eeq
\begin{enumerate}
\item All potentials have the $M$-cutoff;
\beq
\caP_{\leq M}(\scrF^{(r)}) =\scrF^{(r)},\qquad   \caP_{\leq M}(\scrG^{(r)}) =\scrG^{(r)},\qquad  \caP_{\leq M}(\scrK^{(r)}) =\scrK^{(r)}
\eeq
\item The $\scrF^{(r)}$-potential is small and nonresonant
\beq
\interleave    \scrF^{(r)}   \interleave_{\delta^{-\nu}} \leq C(r) M  \delta^{e(r)}, \qquad  \caP_{\mathrm{Nres}}( \scrF^{(r)} )=\scrF^{(r)}. 
\eeq
\item
The $\scrG^{(r)}$-potential is `not too big' and  resonant
\beq
\interleave \scrG^{(r)} \interleave_{\delta^{-\nu}} \leq C(r) \delta^{e(0)} M, \qquad  \caP_{\mathrm{Res}} (\scrG^{(r)})=\scrG^{(r)} 
\eeq 
\item The  $\scrK^{(r)} $-potential is small;
\beq
\interleave \scrK^{(r+1)} \interleave_{\delta^{-\nu}}  \leq  C(r) \delta^{e(r)+1}  
\eeq
\end{enumerate}
\end{proposition}
 Before giving the proof, we slightly reformulate this theorem to put it in the form in which it will be used.  To that order, let us define two additional operations on interaction potentials:   First the operation $\scrA \mapsto \caD(\scrA)$ that selects only the diagonal terms
\beq  \label{def: diagonal restriction}
(\caD(\scrA))({A}) :=  \sum_{\eta \in \Om_A} P_\eta \scrA(A) P_\eta
\eeq
and $\scrA \mapsto \caR_n(\scrA)$ for some $n>0$, the restriction to terms of  range not larger than $n$ on the lattice and in the number-operator basis
\beq  \label{def: range restriction potentials}
\caR_{n}(\scrA)({A}) :=   \chi(\str A \str \leq  n)  \sum_{\eta, \eta' \in \Om_A}  \chi( \str \eta-\eta'\str \leq n)  P_\eta \scrA(A)P_{\eta'}
\eeq
where $\str \eta\str^2 := \sum_{x \in A} \str \eta(x)\str^2$.
Now we define a new decomposition of  potentials:
\begin{align}
\e^{ \adjoint_{\scrK^{(r)}} } \ldots \e^{ \adjoint_{\scrK^{(2)}} }\e^{ \adjoint_{\scrK^{(1)}} } \left(\scrE \right)& =  \left(\caD \caR_{r}(\scrE_0+ \scrG^{(r)}))\right) +\left((1- \caD) \caR_{r} ( \scrG^{(r)})) \right) +  \left(  {\scrF^{(r)}}+ (1-\caR_{r})(\scrG^{(r)})\right) \nonumber \\[1mm]
& = : \scrD^{(r)} +\hat \scrG^{(r)}  +  \hat\scrF^{(r)} \label{eq: hat decomp}
\end{align}
This is indeed a decomposition since $\caR_{r}(\scrE_0) =\scrE_0$ and $\caD(\scrE_0) =\scrE_0$.  Then 
\begin{corollary}\label{cor: splitting} 
\beq  \interleave \hat\scrF^{(r)} \interleave_{\delta^{-\nu/2}} \leq  C(r) M ( \delta^{e(0)+(\nu r)/2} + \delta^{e(r)}),     \eeq
\end{corollary}
\begin{proof} By Proposition \ref{thm: renormalization} 2), 
$\scrF^{(r)}$ satisfies this estimate and we only need to check $(1-\caR_{r})(\scrG^{(r)})$.  We note that, in general, for $n \geq 1$, 
\beq \label{eq: to bound cutoff}
 \norm (1-\caR_n)(\scrA)(A) \norm_{\ka'}   \leq \chi(\str A \str >n)  \,  \norm \scrA(A) \norm_{\ka'}  +  \chi(\str A \str \leq n)   \norm\sum_{\eta,\eta'} \chi(\str \eta-\eta'\str > n) P_{\eta}\scrA(A)P_{\eta'} \norm_{\ka'},
\eeq
To analyze the last term, we introduce, for $\si \in \{1,-1\}^{A}$, 
\beq
O_{\sigma}:=      \sum_{\eta,\eta'}     \chi(\str \eta-\eta'\str > n)   \big(\prod_{x \in A} \chi(\sgn(\eta(x)-\eta'(x))= \si(x)) \big)   P_{\eta}\scrA(A)P_{\eta'}
\eeq
where we use the signum function $\sgn(a) := a/\str a \str$ for $a \in \bbR_0$ and $\sgn(0)=1$. Note that 
$$\sum_{\eta,\eta'} \chi(\str \eta-\eta'\str > n) P_{\eta}\scrA(A)P_{\eta'}= \sum_{\si \in \{1,-1\}^{A}} O_\si. $$
The advantage of the operators $O_\si$ is that we can explicitly  perform the supremum over    $w \in \bbR_+^{A}$ in \eqref{def: def kappa operator norm} to obtain, for $\ka\geq \ka' >1$,
\begin{align}
\norm O_\si \norm_{\ka'}  &=  \norm \sum_{\eta,\eta'} (\ka')^{\str\eta-\eta'\str_1} P_\eta \breve O_\si P_{\eta'}  \norm \\[2mm] 
 & \leq\big( \max_{f \in \bbR^{A}: \str f \str \geq n} (\ka'/\ka)^{\str f \str_1}\big) \norm  \sum_{\eta,\eta'} \ka^{\str\eta-\eta'\str_1} P_\eta \breve O_\si P_{\eta'}  \norm  = (\ka'/\ka)^{n} \norm   O_\si \norm_\ka
\end{align}
where we put $\str g \str_1:= \sum_{x } \str g(x) \str$ for functions $g \in \bbR^{A}$ and we recall  the notation $\str g \str^2 = \sum_{x} \str g(x)\str^2$ so that $\str g \str \leq \str g \str_1$, which we used in the last equality for $g=f$.   Hence the last term on the right hand side of \eqref{eq: to bound cutoff} is bounded by 
\beq
 2^{n} (\ka'/\ka)^{n} \norm\scrA(A) \norm_{\ka}, \qquad \text{for}\,\, \ka >\ka'.  
\eeq
because the number of $\si$'s is no larger than $2^n$. 
Therefore, \eqref{eq: to bound cutoff} yields
\beq
 \interleave (1-\caR_n)(\scrA) \interleave_{\ka'}  \leq     (1+2^{n}) (\ka'/\ka)^{n}    \interleave \scrA \interleave_{\ka}
\eeq
We apply this with $\scrA= \scrG^{(r)}$, $n=r$,  and $\ka=(\ka')^2 = \delta^{-\nu}$ and we use the bound of Proposition \ref{thm: renormalization} 3). 
\end{proof}
\begin{proof}[Proof of Proposition \ref{thm: renormalization}]
Our proof is by induction, but of a slightly different statement than that given in the proposition; namely we replace the norm $\interleave \cdot \interleave_{\delta^{-\nu}}$ by $\interleave \cdot \interleave_{(1+2^{-r}) \delta^{-\nu}}$ such that at each induction step, we can reduce the decay parameter in the norm. This is necessary in view of point $3)$ of Lemma \ref{lem: composing norms}, i.e.\ the necessity of $\ka-\ka'>0$.  Throughout the proof, we denote the potential on the right hand side of \eqref{eq: change of variables} by $\scrH^{(r)}$.
 
To save some writing in the formulas, we abbreviate
\beq
\interleave \cdot \interleave_{m(r)}=\interleave \cdot\interleave_{(1+2^{-r}) \delta^{-\nu}}
\eeq
For $r=0$, we set
\beq
 \scrG^{(0)} := \caP_{\mathrm{Res}} (\scrE-\scrE_0), \qquad  \scrF^{(0)} :=\caP_{\mathrm{NRes}}  (\scrE-\scrE_0), \qquad \scrK^{(0)}:=0
\eeq
We choose $e(0)$ and $\nu$ such that 
\beq
\interleave  (\scrE-\scrE_0)   \interleave_{m(0)}  \leq  C   \delta^{e(0)}  M 
\eeq
To satisfy this, note that $\interleave (\scrE-\scrE_0)   \interleave_\ka \leq CM \ka^{2d+6}$, hence we need the condition
\beq
\delta^{-\nu (2d+6)} \leq \delta^{e(0)} \quad \Rightarrow \quad   \nu (2d+6)+e(0) <0
\eeq 
Then the bounds are satisfied because $\caP_{\mathrm{NRes}}, \caP_{\mathrm{Res}}$ are contractions.
  This  establishes the induction hypothesis for $r=0$. \\
  
We now assume that the result holds for a given $r \geq 0$ and we show it for $r+1$.  We consider a transformation 
\begin{align*}
\scrH^{(r+1)} := \e^{ \adjoint_{\scrK^{(r+1)}}} ( \scrH^{(r)} )
\end{align*}
such that, to lowest order in $\scrK^{(r+1)}$, the nonresonant potential $\scrF^{(r)}$ is eliminated.  
 \begin{equation} \label{eq: perturbation condition}
[\scrK^{(r+1)}, \scrE_0] =  - \scrF^{(r)}.
\end{equation}
A possible choice is
\begin{align}
\langle \eta, \scrK^{(r+1)}(A) \eta' \rangle \; :=  \;  \frac{\langle \eta,  \scrF^{(r)}(A) \eta' \rangle}{E_A(\eta)-E_A(\eta')}       
\end{align} 
where the right hand side is defined to be $0$ whenever $E_A(\eta)=E_A(\eta')$. 
It follows\footnote{Here (and only here) we exploit the fact that the weighted norm $\norm \cdot\norm_\ka$ was defined in Section \ref{sec: norms} by replacing an operator $ O$ by $\breve O$} that for any $\ka>0$
\beq
\norm \scrK^{(r+1)}(A) \norm_\ka \leq  \frac{ \delta}{M }  \norm  \scrF^{(r)}(A) \norm_\ka   \eeq
hence in particular
\beq \label{eq: from psi to upsilon}
\interleave \scrK^{(r+1)}  \interleave_{m(r)} \leq  \frac{ \delta}{M } \interleave \scrF^{(r)}   \interleave_{m(r)}. 
\eeq
Now we calculate
\begin{align*}
\scrH^{(r+1)}  \; &= \; 
\scrE_{0} \; +  \; \sum_{k\ge 1} \frac{1}{k!} \ad_{\scrK^{(r+1)}}^k (\scrE_{0})
+\e^{ \adjoint_{\scrK^{(r+1)}}}  (\scrG^{(r)}  )
 +  \sum_{k\ge 0} \frac{1}{k!} \ad_{\scrK^{(r+1)}}^k (\scrF^{(r)})  \\ 
 \; &= \; 
\scrE_{0} \; 
\; + \e^{ \adjoint_{\scrK^{(r+1)}}}   (\scrG^{(r)}  )  +\sum_{k \ge 0} \frac{(k+1) }{(k+2)!} \ad^{k+1}_{\scrK^{(r+1)} } (\scrF^{(r)}) 
\end{align*}
where we used \eqref{eq: perturbation condition} to get the last line. 
We define 
\begin{align*}
\scrG^{(r+1/2)} & :=  \e^{ \adjoint_{\scrK^{(r+1)}}} ( \scrG^{(r)}), \\
\scrF^{(r+1/2)} &:=  
\sum_{k \ge 0} \frac{(k+1) }{(k+2)!} \ad^{k+1}_{\scrK^{(r+1)}} (\scrF^{(r)} ) \\   
\scrG^{(r+1)} \; & :=   \caP_{\mathrm{Res}}  \left( \scrG^{(r+1/2)}
\; + \;   \scrF^{(r+1/2)}   \right)    \\
\scrF^{(r+1)}\; & :=   \caP_{\mathrm{NRes}}  \left(   \scrG^{(r+1/2)}
\; + \;  \scrF^{(r+1/2)}   \right)   
\end{align*}
so that indeed $\scrH^{(r+1)}=\scrE_0+\scrF^{(r+1)}+\scrG^{(r+1)}$. 
It remains to verify the bounds. Let us first consider $\scrF^{(r+1)}$: 

Note that
\begin{align*}
 \scrF^{(r+1/2)}=    \sum_{k \ge 0} \frac{g(k) }{k!} \ad^{k}_{\scrK^{(r+1)} } (\ad_{\scrK^{(r+1)}}(\scrF^{(r)}) ) \\
   \end{align*}
 for  a bounded sequence $\str g(k) \str \leq 1$. 
Therefore, by Lemma \ref{lem: composing norms} $2),3)$,
\begin{align}
  \interleave  \caP_{\mathrm{NRes}}( \scrF^{(r+1/2)} ) \interleave_{m(r+1)}   & \leq   \interleave  \scrF^{(r+1/2)}\interleave_{m(r+1)}    \nonumber  \\[2mm]
  & \leq    (1- C(r)  \interleave \scrK^{(r+1)} \interleave_{m(r+1/2)} )^{-1}  \interleave \ad_{\scrK^{(r+1)}}\scrF^{(r)} \interleave_{m(r+1/2)}     \nonumber   \\[2mm]
 & \leq  \frac{C(r)}{  1- C(r)  \interleave \scrK^{(r+1)} \interleave_{m(r)} }   \interleave  \scrK^{(r+1)} \interleave_{m(r)} \interleave  \scrF^{(r)}  \interleave_{m(r)}     \label{eq: first contribution}
  \end{align}
where we also used $\interleave \cdot \interleave_{\ka'} \leq \interleave \cdot \interleave_{\ka}$ for $1 \leq \ka' \leq \ka$. 

Next, we estimate the contribution to $\scrF^{(r+1)}$  from $\caP_{\mathrm{NRes}}   \scrG^{(r+1/2)}$.  Proceeding as above, we get, for some sequence $\str g(k) \str \leq 1$, 
\begin{align}
\interleave  \caP_{\mathrm{NRes}}    \scrG^{(r+1/2)}   \interleave_{m(r+1)}  & \leq   \sum_{k \geq 1}   \frac{1 }{k!}   \interleave  \ad^{k}_{\scrK^{(r+1)}} (\scrG^{(r)}) \interleave_{m(r+1)}  
\nonumber  \\[2mm]
& \leq   \sum_{k \geq 0}   \frac{g(k) }{k!}   \interleave  \ad^{k}_{\scrK^{(r+1)}} (\ad_{\scrK^{(r+1)}}(\scrG^{(r)})) \interleave_{m(r+1)} \nonumber  \\[2mm]
& \leq  \frac{C(r)}{  1- C(r) \interleave \scrK^{(r+1)} \interleave_{m(r)} }   \interleave \scrK^{(r+1)} \interleave_{m(r)} \interleave  \scrG^{(r)}  \interleave_{m(r)}       \label{eq: second contribution} 
   \end{align}
The first inequality follows because the induction hypothesis $ \caP_{\mathrm{NRes}}  (\scrG^{(r)}) =0$ allows to drop the $k=0$ term.
   By the induction hypothesis  and \eqref{eq: from psi to upsilon}, we have $\interleave \scrK^{(r+1)} \interleave_{m(r)} \leq C(r) \delta^{e(r)+1}$ and 
   therefore the denominators in the above formulae are of order $1$ since
   \beq
1+e(r) >0.
   \eeq   
Adding the two contributions \eqref{eq: first contribution}  and \eqref{eq: second contribution}, we get
   \begin{align*}
 &  \interleave  \scrF^{(r+1)}\interleave_{m(r+1)}  \leq    C(r) \left(  \delta^{2e(r)+1} + \delta^{e(0)+e(r)+1} \right) 
  \end{align*}
  and hence the bound on $ \scrF^{(r+1)}$ holds because
  \beq
e(r+1) \leq \min (2 e(r)+1, e(r)+e(0)+1)
  \eeq
The bound on the potential $\scrG^{(r+1)}$ is derived by analogous (though simpler) reasoning. 
\end{proof}

 \subsection{Transformations and spatial truncations}\label{sec: unitary trafo}

Proposition \ref{thm: renormalization} is set in the language of transformed potentials. We investigate the question how accurately such transformations can be restricted to small volumes.  The  results are Lemma \ref{lem: bound cutoff pot} and \ref{lem: local unitaries}. These are fairly intuitive technical statements that are necessary in Section \ref{sec: proofs}, but their proofs appear complicated, which is definitely a drawback of the use of interaction potentials.  We think one can safely omit these Lemma's in a first reading.  

First, if two potentials $\scrA, \scrK$ are finite in one of the $\interleave \cdot \interleave_{\ka_1,\ka_2}$-norms, then the equality 
\beq \label{eq: really inverses}
\e^{-\adjoint_{\scrK}} \e^{ \adjoint_{\scrK}} (\scrA) =  \scrA 
\eeq
holds (in a weaker norm). This can be checked explicitly by manipulating the defining series \eqref{eq: definition unitary trafo}. 
Let us abbreviate
\beq \label{def: total trafo potentials}
\caK(\scrA) = \caK^{(r)}(\scrA):= \e^{ -\adjoint_{\scrK^{(1)}} }  \e^{- \adjoint_{\scrK^{(2)}} }  \ldots   \e^{ -\adjoint_{\scrK^{(r)}} } (\scrA). 
\eeq
with $\scrK^{(j)}$ as given in Proposition \ref{thm: renormalization}.  Then, by \eqref{eq: really inverses}, we can invert the operator $\caK$: 
  \beq
 \caK^{-1}(\scrA)=  (\caK^{(r)})^{-1}(\scrA) = \e^{ \adjoint_{\scrK^{(r)}} } \ldots \e^{ \adjoint_{\scrK^{(2)}} }\e^{ \adjoint_{\scrK^{(1)}} }  (\scrA). 
  \eeq
By repeated application of Lemma \ref{lem: composing norms} 3) and Proposition \ref{thm: renormalization} 4), one shows that
\beq \label{eq: cak preserves bounds}
\interleave \caK(\scrA) \interleave_{\tfrac{\ka_1}{2}, \ka_2   } \leq C(r) \interleave \scrA \interleave_{\ka_1, \ka_2   }, \qquad   \interleave \caK^{-1}(\scrA) \interleave_{\tfrac{\ka_1}{2}, \ka_2   } \leq C(r)   \interleave \scrA \interleave_{\ka_1, \ka_2   }.
\eeq
for $\ka_1,\ka_2 \leq \delta^{-\nu}$ and $\delta$ small enough, depending on $r$.
 For a finite set $D$, we define
  \beq
  U_{D}  := \e^{ {\pot}_D(\scrK^{(r)}) } \ldots \e^{{\pot}_D({\scrK^{(2)}}) }\e^{ {\pot}_D({\scrK^{(1)}}) }
  \eeq
  Note that $U_D$ is unitary since ${\pot}_D(\scrK^{(j)})$ are anti-Hermitian matrices, as one checks by inspecting the definitions of $\scrK^{(j)}$ and $\scrF^{(j)}$.  By repeated application of \eqref{eq: unitary trafo} we derive
\beq \label{eq: cak and unitary}
\pot_D(\caK(\scrA)) = U_D \pot_D(\scrA) U^*_D
\eeq
In what follows, we will interpret an operator $O$ as a potential $\scrA_{O}$ such that
\beq
\scrA_O(A)= \begin{cases} O & A=A_O \\  0 & A \neq A_O \end{cases}
\eeq
  for some connected set $A_O$ such that  $s(O) \subset A_O$. If  $A_O \subset D$, the identity \eqref{eq: cak and unitary} reads
\beq  \label{eq: cak and unitary restricted}
{\pot}_{D}(\caK(O)) = U_{D}OU^{*}_{D}
\eeq
where, as announced, $\caK(O)=\caK(\scrA_O)$ on the left hand side. From now on, we write $O$ for $\scrA_O$ without further comment.  
To quantify the dependence on the set $D$ in the above formula, it is helpful to define first the restriction of a potential to some volume: 
  Let 
\beq \label{def: volume restriction potentials}
\caI_{D}(\scrA)(A) := \chi(A \cap D \neq \emptyset) \scrA(A).\eeq
then we have, for $A_O \subset D \subset V$ 
\beq \label{eq: decomp caic and u}
{\pot}_V(\caK(O)) = {\pot}_V(\caI_{D^c}\caK(O) )+   U_{D}OU^{*}_{D}
\eeq
 The upcoming Lemma \ref{lem: bound cutoff pot} provides some bounds. 
In what follows, we will stop keeping track of the precise value of exponents like $\nu$.
 We will also set $\kappa_2=1$ in the norm $\norm \cdot \norm_{\ka_1,\ka_2}$  for simplicity, because, once $\hat\scrF,\hat\scrG$ have been defined, the parameter  $\ka_2$ plays no role anymore. 
\begin{lemma}  \label{lem: bound cutoff pot}
Let $\distance(D^c, A_O)> c\str A_O \str$ for some $c>0$, then
\beq \label{eq: bound cutoff pot}
\interleave \caI_{D^c}\caK(O) \interleave_{\delta^{-c' }, 1} \leq C(r) \delta^{c'' \,  \distance(D^c, A_O) } \norm O \norm
\eeq
for some $c',c''>0$. 
\end{lemma}
\begin{proof}
Trivially, for any $c_1>0$
\beq
\interleave O \interleave_{\delta^{-c_1},1} \leq \delta^{-c_1 \str A_O \str} \norm O \norm
\eeq
and hence, by  \eqref{eq: cak preserves bounds}, for $c_1>0$ small enough, 
 \beq  \label{eq: bound on cak o}
 \interleave \caK(O) \interleave_{\frac{\delta^{-c_1}}{2},1} \leq C(r) \delta^{-c_1 \str A_O \str}  \norm O \norm.  \eeq 
Furthermore, if $(\caK(O))(A) \neq 0$, then $A_O \subset A$ and hence, if  $(\caI_{D^c}\caK(O))(A) \neq 0$, then  $\str A \str \geq \distance(D^c, A_O)+\str A_O\str$. Therefore
\begin{align}
\interleave \caI_{D^c}\caK(O) \interleave_{\delta^{-c'}, 1}  & \leq  \delta^{c_2 (\distance(D^c, A_O)+\str A_O\str)} \sum_{A}  \delta^{-(c'+c_2) \str A \str} \norm (\caI_{D^c}\caK(O))(A) \norm \\[2mm]  & \leq
\delta^{c_2 (\distance(D^c, A_O)+\str A_O\str)}  \interleave  \caK(O)   \interleave_{\delta^{-(c'+c_2)},1}\\[2mm]  
& \leq C(r)
\delta^{c_2 (\distance(D^c, A_O)+\str A_O\str)} \delta^{-(1+c_3)(c'+c_2) \str A_O \str} \norm O \norm.\\[2mm]  
& \leq C(r)
\delta^{c_2 \distance(D^c, A_O)}  \delta^{ -(  (1+c_3)  c' + c_3c_2) \str A_O \str} \norm O \norm. \label{eq: constant circus}
\end{align}
The second inequality follows from $\interleave \caI_{D^c}(\scrA)\interleave_{\ka_1,\ka_2} \leq \interleave \scrA \interleave_{\ka_1,\ka_2} $ and the third inequality follows from  \eqref{eq: bound on cak o}, for $\delta$ small enough such that $\delta^{-(c'+c_2)} \leq  (1/2)\delta^{-(1+c_3)(c'+c_2)}$.    
The claim now follows from \eqref{eq: constant circus} by using $\distance(D^c, A_O)> c\str A_O \str$ and choosing $c',c_3$ small enough such that
\beq
c'':= c_2 -(1/c)(  (1+c_3)  c' + c_3c_2)  >0. 
\eeq

\end{proof}

Obviously, changing the volume $D$ far away from $A_O$ leads to small changes in $U_{D}OU^*_{D}$, as we show next.
We denote  the symmetric difference of sets  by $D \Delta D' : = (D \cup D' )\setminus (D \cap D' )$. 
\begin{lemma}\label{lem: local unitaries}   If $\distance(D \Delta D', A_O) \geq c \str A_O \str$ for some $c>0$, then 
 \beq
\norm U_{D}OU^*_{D} -  U_{D'}OU^*_{D'}  \norm \leq C(r)  \delta^{c' \, \distance(D \Delta D', A_O)}   \norm O \norm. 
 \eeq
 for some $c'>0$. 
\end{lemma}
\begin{proof}
Note that this lemma is not restricted to the case $A_O \subset (D \cap D')$. Let us however first treat this case. Then \eqref{eq: decomp caic and u}, applied to both $D$ and $D'$, yields for large enough $V$,
\beq \label{eq: difference}
 U_{D}OU^*_{D} -  U_{D'}OU^*_{D'}   = {\pot}_V(\caI_{(D')^c}\caK(O))- {\pot}_V(\caI_{D^c}\caK(O) ) = \sum_{A \subset V: A \cap (D \Delta D') \neq \emptyset}  \varsigma(A) \times (\caK(O))(A)
\eeq
where $\varsigma(A)=\pm 1$.  The operator norm of the left-most expression is trivially bounded by
\beq
 \interleave \caI_{(D \Delta D')}\caK(O) \interleave_{\ka,1}, \qquad \text{for any $\ka>1$}, 
\eeq 
and hence the claim follows by Lemma \ref{lem: bound cutoff pot} with $(D \Delta D')$ in the role of $D^c$.\\
Next, we consider the case where $ G := A_O \setminus (D\cap D') $ is not empty.  Note that $G \cap (D \cup D')=\emptyset$ since $(D \Delta D') \cap A_O  =\emptyset$. 
 Set 
$$ \tilde D:= D \cup G, \qquad    \tilde D':= D' \cup  G. $$
and define modified potentials $\tilde \scrK^{(j)}$ by 
\beq
\tilde \scrK^{(j)}(A) := \begin{cases} \scrK^{(j)}(A) & A \subset (D \cup D')  \\   0&   A \not \subset (D \cup D')   \end{cases}
\eeq  
and let  $\tilde U_A$ (for a set $A$) by the modified version of $ U_A$ obtained by replacing $\scrK^{(j)}$ by $\tilde \scrK^{(j)}$.  
Then it is clear that
\beq
\tilde U_{\tilde D} =  U_{ D}, \qquad    \tilde U_{\tilde D'} =  U_{ D'}
\eeq
such that in particular 
\beq \label{eq: difference tilded}
  U_{ D}O  U^*_{ D} -   U_{ D'}O  U^*_{ D'}  =      \tilde U_{\tilde D}O \tilde U^*_{\tilde D} -  \tilde U_{\tilde D'}O \tilde U^*_{\tilde D'}.   \eeq
  For the second expression, the above proof still applies since  $A_O \subset \tilde D \cap \tilde D'$ and hence we conclude that its norm is bounded by
  \beq
  C(r)  \delta^{c' \, \distance(\tilde D \Delta \tilde D', A_O)}   \norm O \norm
  \eeq
 Since however  $\tilde D \Delta \tilde D'=  D \Delta  D'$, we have obtained the claim of the lemma.
\end{proof}

\section{Analysis of the resonant Hamiltonian: Invariant subspaces}\label{sec: analysis invariant subspaces}

We define  the resonant Hamiltonian in the strip $\bbS$ defined in \eqref{def: strip first}:
\beq \label{eq: def z}
Z=Z^{(r)}:= \pot_{\bbS}(\scrD^{(r)}) + \pot_{\bbS}(\hat \scrG^{(r)}) 
\eeq
where $\scrD^{(r)},  \hat \scrG^{(r)} $ were defined preceding Corollary \ref{cor: splitting}.  Note that the potential $\hat \scrG^{(r)}$  depends on the resonance threshold $\delta$ that we choose as
 \beq \label{eq: def of m via gamma}
 \delta=M^{-\gamma_1}, \qquad \text{for some $0<\gamma_1<q-2$} \eeq
It is always understood that $M$ is taken large enough, possibly depending on $r$. This will not be repeated at every step.
The main point of the analysis below is to show that the non-diagonal terms in the Hamiltonian $Z$ are  \emph{sparse}, and therefore, transport induced by this Hamiltonian is small. This goal will be achieved in Proposition \ref{proposition: commutator left right} and one can consider the Sections \ref{sec: analysis invariant subspaces} and \ref{sec: left right splitting} as the proof of this result. 

\subsection{Setup and definition}
In the present section \ref{sec: analysis invariant subspaces}, our analysis will depend on a volume $V \subset \bigvolume$ that should be thought of as being much smaller  than $\bigvolume$.
Even though this is not necessary for most of the statements below, we will always assume that $\str V \str \leq (2r)^{2d}$, as will anyhow be done in Section \ref{sec: left right splitting}.  
We mostly drop the dependence on $r$, for example writing $Z=Z^{(r)}$, but
we write $C(r),c(r)$ for constants $C(r) < \infty, c(r)>0$ that can depend on $r$.  Recall that $\Om_V^{(M)}$ is the phase space in $V$ with a cutoff at $M$.  In what follows we often  abbreviate $\Om_V=\Om_V^{(M)}$ because the high-energy cutoff is always in place.

To write  the Hamiltonian $Z$ in a more explicit way, we introduce
\begin{definition}[Moves] \label{def: moves} \emph{
For a  volume $V \subset \bbS$ and $r  \in \bbN$, we set 
$$\caP^{(V)} := \{ \rho \in \bbZ^{V}:  1 \leq \str \rho \str \leq r,   \str{s(\rho)}\str \leq r \}$$
where $s(\rho)=\{ x :  \rho(x) \neq 0\}$.   We also define the `dependence set'  of a move
\beq \label{eq: dependence set}
 S(\rho)  : = \bigcup_{\substack{A \subset \bbS \\    \str A \str \leq r,    s(\rho) \subset A    } }   A
 \eeq
such that, in particular,  $\diam(S(\rho)) \leq 2 r$. }
 \end{definition}
 To recast the Hamiltonian $Z$ in terms of `moves', we first introduce the `move'-operators
\beq \label{def: move operators}
W_{\rho} : =  \sum_{A \subset \bbS} \sum_{\eta \in \Om_{\bbS}}  P_\eta \hat\scrG^{(r)}(A) P_{\eta+\rho}
\eeq 
They satisfy 
\begin{enumerate}
\item  the high-energy cutoff  $\caP_{\leq M} (W_\rho)=W_\rho$
\item the locality property  $s(W_\rho) \subset S(\rho)$ (In particular, the sum over $A$ in \eqref{def: move operators} can be restricted to subsets of $S(\rho)$).
\item  the bound $ \norm W_\rho \norm \leq C(r) M^C   $.
\item  a resonance condition: $\langle \eta, W_\rho \eta' \rangle=0 $ unless $ \str E_{S(\rho)}(\eta)-E_{S(\rho)}(\eta') \str \leq M^{1+\gamma_1} $.
\end{enumerate}
This is easily checked relying on the locality and bounds on $\hat \scrG^{(r)}$, and \eqref{eq: def of m via gamma}. 
We can now recast  the  Hamiltonian $Z$ as 
\beq
 Z= \pot_{\bbS}(\scrD)+ \pot_{\bbS}(\hat\scrG)= \pot_{\bbS}(\scrD)  + \sum_{\rho \in \caP^{(\bbS)}} W_{\rho}  
\eeq 
Moreover, we recall that $\scrD(A)=0$ unless $\str A \str \leq r$.

Next, we define a partition of the phase space into (possibly delocalized) components such that the resonant Hamiltonian $Z$ cannot induce transport between the components. In the remaining part of this section, we will not need the strip $\bbS$, nor the Hamiltonian $Z$. Instead, we focus on the (joint) structure of the operators $W_\rho$ with $\rho \in \caP^{(V)}$.   When confusion is excluded, we sometimes drop $V$ from our notation.   

\begin{definition}[Partition] \label{def: partition}\emph{
Let  $\eta,\eta' \in  \Om_V$. Define
\beq \label{def: equivalence etas}
\eta \mathop{\sim}\limits_\rho \eta'  \qquad \Leftrightarrow \qquad  \left(  \eta'- \eta \in \{-\rho, \rho\} \quad \text{and}\quad  \str E_V(\eta)-E_V(\eta') \str \leq M^{1+\gamma_1}   \right)
\eeq
and 
\beq
\eta \mathop{\sim} \eta'  \qquad \Leftrightarrow \qquad  \big(  \eta \mathop{\sim}\limits_\rho \eta' \quad \text{for some} \quad \rho \in \caP  \big)
\eeq
Note that the relation $\sim$ is  an adjacency relation, hence it induces a  partition of $\Om_V$ into connected components. We call this partition $\caF=\caF^{(V)}$ and its elements are denoted by $\mu,\mu',\ldots \in \caF$. We write
\beq
\caP(\mu) = \caP^{(V)}(\mu) :=\{  \rho \in \caP\,\str \, \exists \eta,\eta' \in \mu:   \eta \mathop{\sim}\limits_\rho \eta'   \}
\eeq
and, for $A \subset {V}$,
\beq
\caP_A(\mu)=\caP_A^{(V)}(\mu):= \{ \rho \in \caP(\mu): s(\rho)  \subset A\}.
\eeq
and we also write $\caP_A(\eta)= \caP_A(\mu(\eta))$ where $\mu (\eta)$ is the unique $\mu \in \caF$ such that $\eta \in \mu$. 
}

 \end{definition}
 Note  that for $\rho \in \caP^{(V)}$, {it is not guaranteed that $s(W_\rho) \subset V$} because $s(\rho) \subset V$ does not imply $S(\rho)\subset V$.
From Definition \ref{def: partition}, it is immediate that  
$$[P_\mu,W_\rho]=0, \qquad \text{with}\,\, P_\mu  = \sum_{\eta \in \mu} P_\eta \quad \text{and}\,\,  \rho \in \caP^{(V)}, \mu \in \caF^{(V)}$$ and, since $\scrD(A)$ is diagonal in the $\eta$-basis, also $
[\scrD(A), P_\mu]=0$.  Hence we have indeed found invariant subspaces for $Z$.

\subsection{Structure of the partition $\caF$}
The main virtue of this construction is that the partition $\caF$ is rather fine, so that transport by $Z$ can only take place in small sets (in configuration space). We show indeed that if $\eta,\eta'$ belong to the same $\mu$ in the partition, then $\str \eta-\eta'\str \leq M^{1-c}$ for some $c>0$, in other words the size of the sets $\mu$ is small compared to, $M$, the size of the local phase space. In practice, it is more convenient to work with transformed $\eta$'s: 
\begin{equation}\label{definition of theta}
\theta(x) \; = \; (\frac{\eta(x)}{M})^{q-1} 
\quad \text{for} \quad x \in V.
\end{equation}
We write $\theta(\eta)$ for $\theta$ defined in this way.  Note that $\theta \in [0,1]^{\str V \str}$.  Recall that we write $\str \xi \str=\str \xi \str_2=  (\sum_x  \str \xi(x)\str^2)^{1/2}$ for $\xi \in \bbR^{\str V\str}$. 
\begin{proposition} \label{prop: reduced mobility}
Assume that $0<\gamma_1<q-2$ in the resonance condition \eqref{eq: def of m via gamma} and let $\gamma_2$ satisfy $0< {\gamma_2} < \min (1,(q-2) -\gamma_1)$
Then, there is  $ C_0(r)<\infty$ such that, for sufficiently large $M$ (depending on $r$), 
\beq
\max_{\mu \in \caF}  \max_{\eta,\eta' \in \mu} \str \theta(\eta)-\theta(\eta')\str \leq  C_0(r)M^{-\gamma_2}.
\eeq 
\end{proposition}
Recall that we assumed $\str V \str \leq (2r)^{2d}$, which is the reason there is no explicit dependence on  $V$ in the bound on the right hand side.

We define the scalar product 
\beq
\langle \xi,\xi' \rangle = \sum_{x \in V}  \xi(x) \xi'(x), \qquad \xi,\xi'   \in \bbR^{ V}
\eeq
corresponding to the norm $\str\xi\str$ used above. 
In what follows, $x$ always ranges over $V$ and we drop this from the notation. 
It is clear from the definition of the partition $\caF$ that, if two configurations $\eta,\eta'$ belong to the same partitioning set $\mu$, then there must be a finite sequence $(\eta_{n})_{n\ge 1}\subset \Om_V$  such that 
\begin{equation}\label{condition 2 sequence eta}
\eta _{n+1} \mathop{\sim}\limits_{\rho_n} \eta_{n}, \qquad \text{for any $1 \leq n < l$ and $\rho_n \in \caP$},
\end{equation}
and $\eta_1=\eta, \eta_l=\eta'$. 
 In what follows we abbreviate $\theta_n:=\theta(\eta_n)$
We will now show in a series of lemma's that for any such sequence (in particular, for any $l$),   $\str \theta_l-\theta_1 \str$ is bounded as in the statement of Proposition \ref{prop: reduced mobility}. 
The main idea is as follows: The relation $\eta _{n+1} \mathop{\sim}\limits_{\rho_n} \eta_{n}$ imposes a strong constraint on $\eta_n$ or $\theta_n$. As we see in Lemma \ref{lem: condition on gammathree}, it essentially means that $\rho_n \perp \theta_n$.  If we could pretend that  $\eta \parallel \theta(\eta)$, then $\eta_n$ would be in the plane perpendicalur to $\rho_n$ and we see therefore that $\eta_{n+1}= \eta_n+\rho_n$ moves away from this plane; addding $\rho_n$ sufficiently many times, the resulting  $\eta$ will not longer be orthogonal to $\rho_n$.  This is eventually the effect of nonlinearity and it is the main reason why the components $\mu$ are small.  Of course, $\eta \not \parallel \theta(\eta)$ in general, but we clearly see that $\theta(\eta)$ has a component collinear with $\eta$. The task accomplished in the next four lemma's is to make this idea precise, in particular when condition $\rho \perp \theta_n$ holds for several moves $\rho$.

\begin{lemma}\label{lem: condition on gammathree} Let $0< {\gamma_2} < \min (1,(q-2) -\gamma_1)$, then
\begin{equation} \label{eq: theta mobility condition}
\big| \langle \theta_{n}, \rho_{n} \rangle \big|
\; \le \; 
M^{-\gamma_2}.
\end{equation}
\end{lemma}
\begin{proof}
\begin{align}
q \bigl\str \sum_x \rho_n(x) \eta_n(x)^{q-1} \bigr \str
\;&\le\; 
\bigl\str q\sum_x \rho_n(x) \eta_n(x)^{q-1}  - ( E_{V}(\eta_{n+1}) - E_{V}(\eta_n))\bigr \str
\; + \; \big| E_{V}(\eta_{n+1}) - E_{V}(\eta_n) \big| 
\nonumber\\
\; &\le \; 
C(r) \, \big( 1 + \sum_{x} \eta_n(x)^{q-2} \big) + \; \big| E_{V}(\eta_{n+1}) - E_{V}(\eta_n) \big| 
\nonumber\\
\; &\le \; 
C(r) M^{q-2}  \; + \;  M^{1+\gamma_1}.
\end{align}
The second inequality is by  the fundamental theorem of calculus, the third inequality uses $\str\eta_n(x)\str \leq M$ and the resonance condition \eqref{def: equivalence etas}. 
Dividing by $M^{q-1}$ and taking $M$ large enough yields the claim. 
\end{proof}
We will now define regions $\caZ(m)$ in $\bbR_+^{\str V\str}$  such that for all $\theta \in \caZ(m)$, the condition  $\big| \langle \theta, \rho \rangle \big|
\; \le \; 
M^{-\gamma_2}$ is `nearly satisfied' for $m$ linearly independent 'moves' $\rho^{1}, \ldots, \rho^{m}$, but far from satisfied for any move $\rho$ that is not contained in $\spann \{\rho^{1}, \ldots, \rho^{m}\}$.  The construction depends  on a parameter $L>2$ that will be chosen to be large enough later on.    
\begin{definition} \label{def: mobility regions}
\emph{  Let $\caZ(m) \subset \bbR_+^{\str V\str}$ with $1\leq m \leq \str V \str$ be the set of those $\theta$ for which there is a linearly independent collection $\{\rho^{1}, \ldots, \rho^{m}\} \subset \caP$ such that 
\begin{enumerate}
\item   $  |\langle \theta, \rho^j \rangle| \; \le \; L^{m-1} M^{-{\gamma_2}} \qquad \text{for $j=1,\ldots,m $}.  $
\item  $ |\langle \theta, \rho \rangle| \; > \; L^{m} M^{-{\gamma_2}}\qquad  \text{for any } \rho \in \mathcal P \setminus \mathrm{span} \{ \rho^1, \dots , \rho^{m} \} $
\end{enumerate}
For $m=0$, we let $\caZ(0) \subset \bbR_+^{\str V\str}$ be the set of $\theta$ such that $ |\langle \theta, \rho \rangle| \; > \; M^{-{\gamma_2}}$ for any $\rho \in \mathcal P$ (recall that $\str \rho\str \geq 1$).
}
\end{definition}

Note that  $(\bbR_+)^{ V } = \cup_{j=0}^{\str V \str} {\caZ} (j)$ but the regions $\caZ(m)$ are in general not  disjoint.   In what follows we say that $\theta \in \caZ(m)$ by virtue of $\rho^1,\ldots, \rho^m$ if   $\{\rho^{1}, \ldots, \rho^{m}\} \subset \caP$ is one of the linearly independent collections for which the above condition 1.\  holds.  
We first argue that If  $\theta \in {\caZ} (m)$ by virtue of  $\{\rho^1,\ldots, \rho^m\}$, then for any $u \in \spann \{\rho^1,\ldots, \rho^m\}$ with $\str u \str=1$
  \beq \label{eq: has to be bigger}
\str \langle \theta, u \rangle \str \leq  C_1(r) L^{m-1}M^{-{\gamma_2}},  \qquad \text{for some $C_1(r)$.}
\eeq
Let us abbreviate (only here and in Lemma \ref{lem: acrobatics g})    $\caR:=\{\rho^1,\ldots, \rho^m\}$. 
To see \eqref{eq: has to be bigger}, let $e(\caR)$ be the lowest eigenvalue of the symmetric $m\times m$ matrix with entries $\langle \rho^{i},\rho^j \rangle$. Since $\caR$ is linearly independent, $e(\caR)>0$. For $u=\sum_j u_j \rho^j$ with $u_j \in \bbR$, we then have $\sum_j \str u_j\str^2 \leq  (e(\caR))^{-1} \str u \str^2$, hence $\str u_j\str \leq C(\caR) \str u \str$. Since the number of possible collections $\caR$ is $C(r)$, we conclude \eqref{eq: has to be bigger} from condition 1 in Definition \ref{def: mobility regions}.

 \begin{lemma}\label{lem: acrobatics g}
Define the closed set $G_{\caR} \subset \bbR^{\str V\str}$ consisting of $\theta$ such that, for all $x \in V$
\beq
\theta(x)= \nu(x) v(x), \qquad   \text{with $\nu(x) \geq 0, v \in  \spann \caR$}.
\eeq
Then
\beq
 \inf_{\theta \in G_{\caR}, \str \theta\str=1 } \, \,   \sup_{u \in \spann \caR, \str u \str=1} \str \langle \theta, u \rangle \str  \geq c(\caR)  \geq c(r)
\eeq
 \end{lemma}
 \begin{proof}
Assume there is no such $c(\caR) >0$. Since the intersection of $G_{\caR}$ with the unit sphere is compact,  it follows that there is a $\theta \in G_{\caR}, \str \theta\str=1$ such that $\theta \perp \spann \caR$.  However, this is false because
\beq
\langle \theta, v \rangle =  \sum_{x} \nu(x) (v(x))^2     >0
\eeq 
where $\nu, v$ are related to $\theta$ as in the definition of $G_{\caR}$, in particular $v \in \spann \caR$.  $c(\caR)  \geq c(r)$ follows because the number of collections $\caR$ is $C(r)$. 
 \end{proof}
 
 \begin{lemma}\label{lem: stay in mobility plane}
 Let $(\theta_{n})_{n\ge 1}$ be the sequence defined from \eqref{condition 2 sequence eta}.
We fix $k,k' \in \bbN$ with $k < k'$ and we
assume that $\theta_{k}\in {\caZ} (m)$ for some $ 1 \leq m \leq  \str V\str$ by virtue of $\{\rho^1, \dots , \rho^m\}$. 
If 
\begin{equation}\label{distance not too big}
|\theta_{n} - \theta_k| \le \frac{ L}{2r} L^{m-1} M^{-{\gamma_2}} \quad \text{for every $ n $ with} \quad k\leq n \leq k', 
\end{equation}
then
 $$\theta_{k'}-\theta_{k} \in G_{\rho^1, \dots , \rho^{m}}$$
  \end{lemma}
\begin{proof}
By \eqref{distance not too big}, for every $\rho \in \caP \setminus \mathrm{span} \{ \rho^{1}, \dots , \rho^{m} \}$,
\begin{equation}\label{useful for identifying constant}
|\langle \theta_{n}, \rho \rangle|
\; \ge \;  
|\langle \theta_k, \rho \rangle | - |\langle \theta_{n} - \theta_{k}, \rho \rangle |
\; > \; 
L^m M^{-{\gamma_2}} -   \frac{ \str \rho \str}{2r} L^{m} M^{-{\gamma_2}} 
\; \ge \;
M^{-{\gamma_2}} ,
\end{equation}
because $\str \rho \str \leq r$ and $L \geq 2$. 
It follows therefore that condition \eqref{eq: theta mobility condition} does not hold for $\rho$, hence
\begin{equation*}
\eta_{n+1}
\; = \;
\eta_{n} + \rho_{n} \quad \text{with} \quad \rho_{n} \in \mathcal \caP \cap \spann\{ \rho^1, \dots , \rho^m \}.   
\end{equation*}
and, since this holds for any $n$ with $k\leq n \leq k'$, we have 
\begin{equation} \label{eq: diff eta etaprime}
\eta' \; = \; \eta + v  \quad \text{for some} \quad
 v  \in  \mathrm{span}\{ \rho^1, \dots , \rho^m \}.
\end{equation}
Since the function $t \mapsto t^{q-1}$ is strictly increasing for $t \geq 0$, this implies that for each $x$
\begin{equation*}
(\eta'(x))^{q-1} \; = \; (\eta(x))^{q-1} + \nu(x) v(x)  \qquad \text{for some $\nu(x) \geq 0$}
\end{equation*}
which proves the claim because $\theta_k(x)=(\eta(x)/M)^{q-1}, \theta_{k'}(x)= (\eta'(x)/M)^{q-1} $
\end{proof}

\begin{lemma}\label{lem: leave mobility plane}
Let $\theta_k,\theta_{k'}$ be as in Lemma \ref{lem: stay in mobility plane} and assume that all conditions of Lemma \ref{lem: stay in mobility plane} hold. If additionally 
\begin{equation*}
|\theta_{k'} - \theta_k| \; \ge \; \frac{ L}{4r} L^{m-1} M^{-{\gamma_2}}, 
\end{equation*} 
then  $\theta_{k'} \in {\caZ} (m')$ for some $m' < m$.
\end{lemma}
\begin{proof}
The inequalities in \eqref{useful for identifying constant} already show that, for $L$ large enough, 
\begin{equation}\label{a think to check rrr}
|\langle \theta_{k'}, \rho \rangle| \;  > \; L^{m-1} M^{-{\gamma_2}}
\quad \text{for any} \quad
\rho \in \mathcal P \setminus \mathrm{span} \{ \rho^{1}, \dots , \rho^{m} \}.
\end{equation}
So to conclude, by \eqref{eq: has to be bigger}, 
it is enough to find  $u \in \mathrm{span} \{ \rho^{1}, \dots , \rho^{m} \}, \str u \str=1$ 
for which \eqref{eq: has to be bigger} is violated (with $\theta=\theta_{k'}$).
Using \eqref{eq: has to be bigger} for $\theta_k$, we have
\begin{equation*}
|\langle \theta_{k'} , u \rangle| \; \ge \; |\langle \theta_{k'} - \theta_k, u \rangle | - |\langle \theta_{k}, u \rangle|
\; \ge \; 
|\langle \theta_{k'} - \theta_k, u \rangle | -C_1(r) L^{m-1} M^{-{\gamma_2}}
\end{equation*}
and hence,  by choosing $L$ large, it suffices  to choose $u$ such that, for some $c(r) > 0$,
\begin{equation}\label{finally to show for facebookgroup analysis}
|\langle \theta_{k'} - \theta_k, u \rangle | \; \ge \; c(r)\, | \theta_{k'} - \theta_k |,
\end{equation}
This is possible by Lemma \ref{lem: acrobatics g} because, by Lemma \ref{lem: stay in mobility plane},   $\theta_{k'}-\theta_{k}\in G_{\rho^1,\ldots, \rho^m}$. 
\end{proof}
We are now ready to conclude the 
\subsubsection{Proof of Proposition \ref{prop: reduced mobility}}

We pick a sequence $(\theta_n)_{1\leq n\leq l}$ defined from \eqref{condition 2 sequence eta}. By Lemma \ref{lem: condition on gammathree},   $\theta_n \not\in {\caZ} (0)$ unless, possibly, for $n=l$. 
To analyse this sequence, we inductively construct sequences $n_j,m_j$ with $j=1,\ldots, s$ for some $s<\infty$. 

Set $n_1=1$ and choose $0\leq m_1 \leq \str V\str$ to be such that $\theta_1=\theta_{n_1} \in {\caZ} (m_1)$.
If $n_1=l$, then we are done, i.e.\ $s=1$. Otherwise, $m_1 \neq 0$ and we continue.  Assume that $n_j,m_j$ and $\theta_{n_j}$ have already  been chosen for some $j\geq 1$ such that  $n_j < l$ and $\theta_{n_j} \in \caZ(m_j)$ with $0< m_j < \str V\str$.    Then,  let $n_{j+1}$ be the smallest number larger than $n_j$ such that at least one of the following occurs
\begin{enumerate}
\item[$a)$]    $n_{j+1}=l$,
\item[$b)$]   $ |\theta_{n_{j+1}}- \theta_{n_{j}}| \; \ge \; \frac{ L}{4r} L^{m_j-1} M^{-{\gamma_2}}$.
\end{enumerate}
 If $a)$ occurs then we stop the sequence, i.e.\ $s:=j+1$ and we set (a dummy) $m_{j+1}:=m_j$.  If $b)$ occurs but not $a)$, then we choose  $0< m_{j+1}<m_j$  such that $\theta_{n_{j+1}} \in {\caZ}(m_{j+1})$, which is possible by 
 Lemma \ref{lem: leave mobility plane}, and we continue.   In both cases it holds that 
 \beq
  |\theta_{n_{j+1}} - \theta_{n_j}|  \leq  \frac{L^{m_j}}{2r}M^{-\gamma_2}
 \eeq
 Indeed, if $b)$ did not occur, this is trivial and if $b)$ did occur then we derive it from the fact that $b)$ had not occured for $n= n_{j+1}-1$ and from the fact that
 $\str \theta_{n+1}-\theta_n \str \leq C(r) M^{-1}$ (from a simple explicit calculation). \\
 Since $l<\infty$ and $n_j$ is strictly increasing, this procedure ends at some step. 
Collecting all the bounds we get
\beq
|\theta_{l} - \theta_1| \leq   \sum_{j=1}^{s-1}  |\theta_{n_{j+1}} - \theta_{n_j}|   \leq   \sum_{m=1}^{\str V\str} \frac{L^{m-1}}{2r}M^{-\gamma_2} \leq   \frac{L^{\str V\str}}{2r}M^{-\gamma_2}.
\eeq
which ends the proof, as explained before Lemma \ref{lem: condition on gammathree}. 
\qed

\subsection{Locality of the partition $\caF$}

The aim of this section is to control  the set of moves $\caP^{(V)}_A(\eta)$ \emph{locally in $A \subset V$}, i.e.\ without knowing the configuration $\eta$ outside of $A$. 
As such, this is impossible because $\caP^{(V)}_A(\eta)$ is determined globally in $V$, as we in a striking way in Figure \ref{figure: motion flemish mountain}. However, we can still achieve this control if we impose a condition on the boundary of the set $A$, roughly saying that no element of $\caP^{(V)}(\eta)$ is supported there.  This is the content of Lemma \ref{lem: second locality}. Such locality statements become powerful when combined with an argument that tells us that it is easy to find regions $A$ such that this condition on the boundary holds, which we will do in Section \ref{sec: smallness moves}. 

To state a convenient boundary condition, we introduce a set $ \caP'_A(\eta)$ that is bigger than $ \caP^{(V)}_A(\eta)$ but easier to control. Let
\beq \label{def: capprime}
    \caP'_A(\eta) := \bigcup_{\eta' :  \str \theta(\eta)- \theta(\eta') \str \leq C_0(r) M^{-\gamma_2}   }     \{\rho:  \, s(\rho) \subset A \, \text{and}\,\,  \exists \eta'': \eta' \mathop{\sim}\limits_{\rho} \eta''\}
 \eeq
 with $C_0(r)$ as in Proposition \ref{prop: reduced mobility}.
Proposition \ref{prop: reduced mobility} immediately yields 
   \beq
   \caP^{(V)}_A(\mu) \subset \caP_A'(\eta), \qquad \text{    for any $\eta \in \mu, \mu  \in \caF^{(V)}$ and any $V \supset A$ with $\str V \str \leq (2r)^{2d}$}. \label{eq: inclusion cap in capprime}
    \eeq
Note that  $    \caP'_A(\eta)$ is defined locally, which is the reason that it will be indeed easy to control.  
  
  For subvolumes $A \subset V$, we write $\eta_A$ for the restriction of $\eta$ to $A$ and we introduce the boundary set
\beq
\partial_k A:= \{ x \in A, \distance(x, A^c) \leq k\}. 
\eeq
 Next, we consider two volumes $V,V'$ with $\str V \str, \str V'\str \leq (2r)^{2d}$. 
 \begin{lemma} \label{lem: second locality}
 Assume $A \subset V \cap V'$. Let $\eta \in \Om_V, \eta' \in \Om_{V'}$ be such that $\eta_A=\eta'_A$ and
\beq  \label{eq: empty boundary}
\caP_{ \partial_{2r} A}'(\eta)= \emptyset
\eeq
Then, 
\beq \label{eq: equality fbg}
\caP^{(V)}_{A}(\eta)= \caP^{(V')}_{A}(\eta'). 
\eeq
 \end{lemma}

\begin{proof}
From  \eqref{eq: empty boundary} and \eqref{eq: inclusion cap in capprime}, we get
\beq \caP^{(V)}_{ \partial_{2r} A}(\eta) = \emptyset, \qquad  \caP^{(V')}_{ \partial_{2r} A}(\eta') = \emptyset. \label{eq: again symmetric diff empty}      \eeq
Call $\tilde A:= A \setminus  \partial_{r} A $, then for any $\rho \in  \caP^{(V)}_{A}(\eta)$,
\beq
s(\rho) \subset \tilde A, \qquad \text{or}\quad s(\rho) \subset \tilde A^c \quad (\text{where $\tilde A^c:= (\tilde A)^c$}).
\eeq
 because of \eqref{eq: empty boundary} and $\str s(\rho)\str \leq r$. 
As already used below Proposition \ref{prop: reduced mobility}, the claim $\rho \in  \caP^{(V)}_{A}(\eta)$ is equivalent to the existence of a finite sequence $\rho_1,\ldots, \rho_l $ with $\rho_l=\rho$ such that 
\beq \label{eq: condition sequence}
\eta_{n} \mathop{\sim}\limits_{\rho_n} \eta_{n+1}, \qquad \text{ for $n=1,\ldots, l$, and with  $\eta_1=\eta$}
\eeq
(and hence $\eta_{n>1}$ determined by $\rho_n$ via $\eta_{n+1}=\eta_{n}+\rho_{n}$). We observe that we can in fact always find such a sequence with $s(\rho_n)  \subset \tilde A$ for $n=1,\ldots,l$.   Indeed, the validity of the relation
\beq \label{eq: validity}
\eta_{n} \mathop{\sim}\limits_{\rho_n} \eta_{n+1} 
\eeq
 depends on the values of $\eta_n$ in the region $s(\rho_n)$ only, therefore the presence of a $\rho_{n'}, n' <n$ in the sequence with $s(\rho_{n'}) \subset \tilde A^c$ (which influences the configuration $\eta_n$ in $\tilde A^c$ only)  does not influence the validity of \eqref{eq: validity}.  Hence, one can omit all $\rho_n$ with $s(\rho_{n}) \subset \tilde A^c$ and obtain a shorter sequence that still satisfies \eqref{eq: condition sequence}. 
Hence, we now assume that the sequence $\rho_1,\ldots, \rho_l$ was chosen such that  $s(\rho_{n}) \subset \tilde A$ . 
For such a sequence we check that 
\beq \label{eq: condition sequence prime}
\eta'_{n} \mathop{\sim}\limits_{\rho_n} \eta'_{n+1}, \qquad \text{ for $n=1,\ldots, l$ and with $ \eta_1=\eta'$.}
\eeq
(and hence $\eta'_{n>1}$ determined by $\rho_n$ via $\eta'_{n+1}=\eta'_{n}+\rho_{n}$).
Indeed, since the validity of $\eta'_{1} \sim_{\rho_1} \eta'_{2}$ depends on the configurations $\eta'_1$ in $s(\rho_1) \subset \tilde A$ only, and since $(\eta_1)_{\tilde A}=(\eta_1')_{\tilde A} $ and $\eta_{1} \sim_{\rho_1} \eta_{2}$, we see that $\eta'_{1} \sim_{\rho_1} \eta'_{2}$ holds and moreover  $(\eta_2)_{\tilde A}=(\eta'_2)_{\tilde A} $. 
We can iterate this argument to obtain \eqref{eq: condition sequence prime} together with $(\eta_n)_{\tilde A}=(\eta'_n)_{\tilde A} $ for $n=1,\ldots, l$. 
Hence we have proven in particular $\rho \in \caP^{(V')}_{A}(\eta')$, hence  $\caP^{(V')}_{A}(\eta') \subset  \caP^{(V')}_{A}(\eta')$. The opposite inclusion follows in the same way (there is a symmetry between primed and unprimed variables). 
\end{proof}

\subsection{Smallness of $\caP(\mu)$} \label{sec: smallness moves}
We already established that the components $\mu$ are small, but that in itself does not yet capture the intuition of `sparse resonant spots'. That intuition is however made precise now: We show in Lemma \ref{lem: bad is sparse} that  for \emph{most} components $\mu$, the union of sets $S(\rho), \rho \in \caP^{(V)}(\mu)$ is sparse in $V$. Such components are called `good'. 

 In addition to the sets $\caP,\caP'$, we define also 
\beq  \label{def: capdoubleprime}
    \caP''(\eta) := \bigcup_{\eta': \str \theta(\eta)- \theta(\eta') \str \leq 2C_0(r) M^{-\gamma_2}   }     \{\rho\, :\,   \exists \eta'': \eta' \mathop{\sim}\limits_{\rho} \eta''\}
 \eeq
 By Proposition \ref{prop: reduced mobility}, we have (here $\caP'=\caP'_V$)
 \beq \label{eq: inclusion prime doubleprime}
 \cup_{\eta' \in \mu(\eta)} \caP'(\eta') \subset \caP''(\eta)
 \eeq
   Since we need to count configurations $\eta$, it is useful to introduce the counting probability measure $\mathbb{P}^{(M)}$  on $\Om_{{V}}^{(M)}$. 
We abbreviate $\mathbb{P}=\mathbb{P}^{(M)}$.   Also, from now on, we do not keep track of specific exponents like $\gamma_1,\gamma_2,\ldots$ and we simply write $c$. 

\begin{lemma} \label{lem: locally sparse}
\beq \label{eq: locally sparse}
\bbP( \rho \in \caP''(\eta))  \leq  C(r) M^{-c}
\eeq
\end{lemma}
\begin{proof}
If $\rho \in \caP''(\eta)$, then there are $\eta',\eta''$ as in \eqref{def: capdoubleprime}, i.e.\ such that $\str \theta(\eta)-\theta(\eta') \str  \leq  2C_0(r) M^{-\gamma_2}$ and, by  Lemma \ref{lem: condition on gammathree},  $\str \langle \theta(\eta'),\rho \rangle\str   \leq M^{-\gamma_2} $. Since $\str \rho \str \leq C(r)$, we then conclude   
\beq  \rho \in \caP''(\eta) \quad  \Rightarrow \quad   \str\langle \theta(\eta),\rho \rangle  \str  \leq C(r) M^{-\gamma_2}. \label{condition on capdoubleprime} \eeq 
Hence we are led to estimate 
\begin{align}
 \bbP( \str  \langle \theta(\eta),\rho \rangle  \str \leq C(r) M^{-\gamma_2} ) & = M^{-\str V \str} \sum_{\eta}  \chi(    \str  \langle \theta(\eta),\rho \rangle  \str \leq C(r) M^{-\gamma_2}   ) \\[1mm]
  &\leq     \int _{[0,1]^{\str V\str}}\d (\eta/M)       \chi(    \str \langle \theta(\eta),\rho \rangle  \str  \leq C(r) M^{-\gamma_2}   )  \label{eq: restricted integral theta}
\end{align}
To get the inequality, we replaced the sum by an integral at the cost of adjusting $C(r)$.  This is easily justified by observing that
\beq
\str \theta(\eta+t) - \theta(\eta) \str \leq C(r)M^{-1}, \qquad \text{for $t \in [0,1]^{\str V\str}$.}
\eeq
 Obviously, we can restrict  \eqref{eq: restricted integral theta} to the subvolume $s(\rho)\subset V$ without change. By a change of variables, \eqref{eq: restricted integral theta} then equals
\beq
 \int _{[0,1]^{\str s(\rho) \str}}\d \theta   \,  J(\theta)      \chi(   \str  \langle \theta,\rho \rangle  \str  \leq C(r) M^{-\gamma_2}   ), \qquad \text{with $J(\theta)=  (\prod_{x\in s(\rho)} \theta(x))^{-\tfrac{q-2}{q-1}}$.}  
 \eeq
which is bounded by 
$C(r) M^{-c}$ by a  H{\"o}lder inequality.

\end{proof}
 The following definition of 'good' components $\mu$ depends on a constant $c_1$ that will be chosen to be small enough in Lemma \ref{lem: frame} below.

\begin{definition}[Good partitioning sets] \label{def: good} \emph{
A $\mu \in \caF^{(V)}$ is 'good' if the collection of subsets of $V$ 
\beq \label{eq: collection to be covered}
\{ S(\rho) \cap V: \,  \rho \in \caP_V'(\eta)\, \text{for some}\,\,  \eta \in \mu   \}
   \eeq
can be covered by $c_1 r$ sets such that each of those sets has diameter $r$. 
 We let  $\caF^{(V)}_{\mathrm{g}}\subset \caF^{(V)}$ be the collection of good $\mu$. }
\end{definition} 
We now deduce that configurations $\eta$ such that $\mu(\eta)$ is not `good',  have small probability.  
\begin{lemma} \label{lem: bad is sparse}
\beq \label{eq: bad is sparse 2}
\bbP( \mu \not \in \caF^{(V)}_{\mathrm{g}})  \leq  C(r) M^{-cr}
\eeq
\end{lemma}
\begin{proof}
If the collection \eqref{eq: collection to be covered} cannot be covered by $n$ sets with diameter $r$, then, for any  $\eta \in \mu$, there are at least $m=cn$ moves
$\rho^{1}, \ldots, \rho^{m} \in \caP''_V(\eta)$ with mutually disjoint supports, i.e.\ $s(\rho^{i}) \cap s(\rho^j) =\emptyset$ for any $\rho^{i} \neq \rho^j$. To pass from $\caP'$ to $\caP''$, we used \eqref{eq: inclusion prime doubleprime}.  We will now prove that 
\beq \label{eq: bound on collection}
 \bbP\big(\rho^{1}, \ldots, \rho^{m} \in \caP''_V(\eta)  \big) =  \prod_{i=1}^m  \mathbb{P} ( \rho^{i} \in \caP''(\eta) ) \leq   C(r) M^{-c m }
\eeq
First, because of the locality of the definition of $    \caP_V''(\eta) $ and the fact that $\bbP$ is a product measure, the 
events
\beq
\rho^{i} \in  \caP_V''(\eta), \qquad    \rho^{j} \in  \caP_V''(\eta), \qquad  \text{for}\,\,  s(\rho^{i}) \cap s(\rho^j) =\emptyset
\eeq
are $\bbP$-independent.  This is the equality in \eqref{eq: bound on collection}.  The inequality is Lemma \ref{lem: locally sparse}. 
The claim \eqref{eq: bad is sparse 2} now follows from \eqref{eq: bound on collection}  by  choosing $n$ (hence $m$) proportional to $r$ and noting that the number of collections of $m$ distinct $\rho$'s is bounded by $C(r)$.
\end{proof}

%
%


\section{Analysis of the resonant Hamiltonian: Left-Right splitting} \label{sec: left right splitting}
As announced, we split the Hamiltonian $Z$ into a left and a right part, $Z_\links$ and $Z_\rechts$, such that these parts have a sparse commutator.  The main result is in Proposition \ref{proposition: commutator left right}.
\subsection{Preliminary definitions}

Recall  the (restricted) hyperplane $\bbH_a = \{ x: x_1 =a\}$ and the strip 
\beq
\bbS= \bbS_{a,r^2}=  \{ x \in \bigvolume, \str x_1 -a \str < r^2 \}. 
\eeq
 For convenience we gather $y=(x_2,x_3,\ldots, x_d) \in \bbZ^{d-1}$. Sums over $y$ are understood to range over the set $\{y: (a,y) \in \bigvolume\}$,  and we define the regions ($\tilde B_y$ will be used only in Section \ref{sec: proofs})
\begin{align}
B_y & :=\{x :\,  \str x- (a,y)\str \leq 2r^2 \} \cap \bbS_{a,r^2}. \nonumber\\[2mm]
\tilde B_y &:= \{ x :  \str x - (a,y) \str\leq (2r)^{2d} \} \cap \bbS_{a,r^2} \label{def: balls}
\end{align}
We also abbreviate $\caP^{(y)}= \caP^{(B_y)}$  and  $\caF^{(y)} = \caF^{(B_y)}$.

First, we define a procedure that assigns to any $\mu \in \caF^{(y)} $ a decomposition of $B_{y}$ into a left and right region $\mathrm{L}(\mu)$ and $\mathrm{R}(\mu)$. This is the `slaloming' between resonant spots that was discussed in Section \ref{sec: sketch of the proof}. 
 \begin{definition}[Left-right decomposition] \label{def: left right} \emph{
 Fix $y$ and $\mu \in \caF^{(y)}$. 
Let $K^{(y)}_1,K^{(y)}_2, \ldots$ be the connected components of the collection 
\beq \{(S(\rho) \cap B_y) :  \rho \in \caP^{(y)}(\mu) \}, \eeq i.e.\
\beq
\cup_j K^{(y)}_j = \cup_{\rho \in \caP^{(y)}(\mu) } (S(\rho) \cap B_y),\qquad \textrm{and}\quad  j \neq j' \Rightarrow  \quad K^{(y)}_j \cap K^{(y)}_{j'} = \emptyset
\eeq
Then, the left, resp.\ right region is 
\beq
\links(\mu) := \{x \in B_y: x_1 \leq a \} \bigcup_{j: K_j  \cap  \{x_1 \leq a\}  \neq \emptyset } K^{(y)}_j, \qquad   \rechts(\mu):=B_{y} \setminus \links(\mu)
\eeq }
\end{definition}
Note that, if $\caP^{(y)}(\mu)=\emptyset$, then $\links(\mu) = \{x \in B_y: x_1 \leq a \}$, i.e.\ the left-right splitting is the most obvious one.  The intuition is that for `good' $\mu$, the $\links(\mu)$ deviates from $\{x \in B_y: x_1 \leq a \}$ only in a few places, and in particular, $\links(\mu)$ can be determined locally. This is established in the following lemma. The reader might find it helpful to consult Figure \ref{fig: leftright1}, even though the latter is not meant to illustrate the full generality of  Lemma \ref{lem: frame}
\begin{lemma} \label{lem: frame}   Let $r$ be large enough and the constant $c_1$ in Definition \ref{def: good}  small enough.
Fix $y,y'$.  Let the triple  $\{F_0,F_1,F_2\} $  form a partition of $B_y \cup B_{y'}$  such that 
\beq
\distance(F_1,F_2) \geq r^2/4, \qquad \text{and}\qquad
  B_y \setminus F_2 =   B_{y'} \setminus F_2 =:F_{01}
  \eeq
  and hence $F_{01}=F_0 \cup F_1 \subset (B_y \cap B_{y'})$.
Choose $\mu  \in \caF^{(y)}, \mu'  \in \caF^{(y')}  $  such that at least one of them is good, i.e.\ 
  $\mu \in \caF^{(y)}_{\mathrm{g}} $ or $  \mu' \in \caF^{(y')}_{\mathrm{g}}$ and such that
 \beq \label{eq: overlap fbgroups}
 \eta_{F_{01}}= \eta'_{F_{01}} \quad \text{for some}\,\,  \eta \in \mu,  \eta' \in \mu'
 \eeq
  Then,
  \beq
  \caP^{(y)}_{F_1}(\mu)=     \caP^{(y')}_{ F_1}(\mu')
  \eeq
  and  
\beq \label{eq: equality lefts}
 \links(\mu) \cap F_1 =  \links(\mu') \cap F_1 
\eeq 
\end{lemma}
\begin{proof} 
For concreteness, let us assume that $\mu$ is good. 
From  straightforward geometric considerations, using that $\mu$ is good, that $\distance(F_1,F_2) \geq r^2/4$ and that $r$ is chosen large enough and $c_1$ small enough, we can construct a partition $(\tilde F_0,\tilde F_1,\tilde F_2) $ of $B_y \cup B_{y'}$ such that 
\begin{enumerate}
\item  \beq  \tilde F_0 \subset F_0,   \qquad   F_1 \subset \tilde F_1, \qquad   F_2 \subset \tilde F_2 \eeq
\item   \beq  \tilde F_0 \cap S(\rho)=\emptyset \qquad \text{for any}\,\,  \rho \in \mathop{\cup}\limits_{\eta \in \mu} \caP'_{B_y}(\eta).  \label{eq: no sets in tilde f zero}   \eeq
\item \beq
\distance(\tilde F_1, \tilde F_2) > 2r. 
\eeq 
\end{enumerate}
We put 
\beq
 \tilde F_{01} :=     B_y \setminus \tilde F_2 = B_{y'} \setminus \tilde F_2
\eeq
(and hence $\tilde F_{01}=\tilde F_0 \cup \tilde F_1$).  
From $1)$ we get $ \tilde F_{01} \subset  F_{01}$ and from $3)$, we get  $\partial_{2r} \tilde F_{01}\subset \tilde F_0 =\tilde F_0 \cap B_{y''} $ with $y''=y,y'$. Therefore,  $2)$ implies   
\beq
\caP'_{\partial_{2r} \tilde F_{01}}(\eta)=\emptyset \quad \text{for any}\,\, \eta \in \mu \eeq
Take now $\eta,\eta'$ as in \eqref{eq: overlap fbgroups}, i.e.\ in particular $\eta_{\tilde F_{01}}= \eta'_{\tilde F_{01}}$. For these $\eta,\eta'$,  we can apply Lemma \ref{lem: second locality} with $V=B_y, V'=B_{y'}, A=  \tilde F_{01}$  to conclude that 
\beq \label{eq: equality in tilde f zero one}
  \caP^{(y)}_{\tilde F_{01}}(\mu)=   \caP^{(y')}_{\tilde F_{01}}(\mu')
\eeq
We now show \eqref{eq: equality lefts}.   Let us consider the connected components $K^{(y)}_j, K^{(y')}_i $ from Definition \eqref{def: left right} for $y,y'$, respectively.   From \eqref{eq: no sets in tilde f zero}  we get (by \eqref{eq: inclusion cap in capprime}) that 
\beq \label{eq: nothing in border}  S(\rho) \cap \tilde F_0 =\emptyset, \qquad \text{for any $\rho \in \caP^{(y)}(\mu)$}. \eeq 
  Hence, $\tilde F_0$ does not intersect any of the components $K_j^{(y)}$ and therefore any one of them is either contained in $\tilde F_1$ or in $\tilde F_2$.
This need not be true for the components $K_i^{(y')}$, but nevertheless, we can still deduce  that no $K_i^{(y')}$ can intersect both $\tilde F_1$ and $\tilde F_2$.
Indeed, if a given $K_i^{(y')}$ would intersect  $\tilde F_1$ and $\tilde F_2$, than, 
since $\diam(S(\rho)) \leq 2r$ for any $\rho$ and $\distance(\tilde F_1,\tilde F_2) >2r$, we conclude that there must be a $\rho \in \caP^{y'}({\mu'})$ with $S(\rho) \in \tilde F_{01}$. 
However, this is in contradicition with 
(\ref{eq: equality in tilde f zero one}-\ref{eq: nothing in border}).  Now we conclude by \eqref{eq: equality in tilde f zero one} that the connected components $K^{(y)}_j$  contained in $\tilde F_{01}$ coincide with the connected components $K^{(y')}_i$  contained in $\tilde F_{01}$
(and moreover, these components are in fact contained in $\tilde F_1$).  This implies \eqref{eq: equality lefts}. 
%
%

\end{proof}

\subsection{Definition of Left-right splitting}

For notational reasons, we associate, in an arbitrary way,  to any subset $A \subset \bbS$ with $\str A \str \leq r$ a unique coordinate $y=y(A) \in \bbZ^{d-1}$ such that
\beq
y(A) \in \proj_{2,\ldots, d-1} A
\eeq
where we used the coordinate projections: if $x=(x_I,x_{I^c})$ with $I$ a subset of $\{1,\ldots, d\}$, then $\proj_I x= x_I $.

The definition of the left-right splitting is
\begin{align}
Z_{\mathrm{L},y}  & : =   \sum_{\mu \in \caF^{(y)}}  P_\mu  \left(     
\sum_{\scriptsize{\begin{array}{c}  A \subset \bbS: y(A)=y \\      A \subset \links(\mu)             \end{array} }  }
 \scrD(A) +   
\sum_{\scriptsize{\begin{array}{c}  \rho \in \caP^{(\bbS)}: y( {s(\rho)}) =y \\      {s(\rho)} \subset \links(\mu)                \end{array} }  }   W_{\rho}     \right)\\
Z_{\mathrm{R},y}  & : =   \sum_{\mu \in \caF^{(y)}}   P_\mu  \left( \sum_{\scriptsize{\begin{array}{c}  A \subset \bbS: y(A)=y \\      A \not\subset \links(\mu)  \end{array} }}    \scrD(A)   +   \sum_{\scriptsize{\begin{array}{c}  \rho \in \caP^{(\bbS)}: y( {s(\rho)}) =y \\      {s(\rho)} \not\subset \links(\mu)                \end{array} }  }     W_{\rho}     \right)   \label{def: z links and z rechts}
  \end{align}
 It immediately follows that
\beq
 Z_{\mathrm{L},y}+Z_{\mathrm{R},y} = \sum_{A: y(A)=y}   \scrD(A) +    \sum_{\rho \in \caP^{(y)}: y( {s(\rho)}) =y  }      W_{\rho}  
\eeq
Since the region $B_y$ is 'broader' than the strip $\bbS$, any $\rho \in \caP^{(\bbS)}$ satisfies $ \rho \in  \caP^{(y( {s(\rho)}))}$ and hence $\rho$ appears exactly once in the sum on the right hand side. 
Therefore,  we have indeed defined a splitting of $Z$: 
\beq \label{def: zway}
Z= \sum_{y \in \bbZ^{d-1}: (a,y) \in \bigvolume } Z_y, \qquad  Z_y:=    Z_{\mathrm{L},y}+Z_{\mathrm{R},y} 
\eeq
We will now describe the good properties of this splitting.  To translate the sparseness of a collection of configurations $\eta$ into a bound on operators, we introduce the normalized trace of operators $O$ that are restricted to low-energy, i.e.\ $\caP_{\leq M}(O)=O$, by
\beq
\tr^{(M)}_A (O ):=  \frac{\Tr_{A} (O)}{\Tr_A (\lone_{\leq M})},\qquad \text{whenever}\,\,  s(O) \subset A
\eeq
where we used $\lone_{\leq M} = \otimes_{x \in A} \chi(N_x \leq M)$.
We  note that the right-hand side does not depend on the set $A$, provided that $s(O) \subset A$ and so we can write $\tr^{(M)}(O)$ without ambiguity.   For example, 
we will use the projections onto good configurations:
\beq \label{def: projections on bad}
P_{\caF^{(y)}_{\mathrm{g}}} =\sum_{\mu \in \caF^{(y)}} P_\mu, \qquad \bar P_{\caF^{(y)}_{\mathrm{g}}}= \lone -P_{\caF^{(y)}_{\mathrm{g}}}
\eeq
acting on $\caH_{B_{y}}$, then, with the normalized trace, we can restate \eqref{eq: bad is sparse 2} as 
\beq  \label{eq: bad is sparse 3} 
\tr^{(M)}(\bar P_{\caF^{(y)}_{\mathrm{g}}}) = \bbP(\mu(\eta) \not \in \caF^{(y)}  )    \leq   C(r) M^{-cr}
\eeq 
We will also need  the associated Hilbert-schmidt norm
\beq \label{def: hs norm trace}
\norm O\norm^2_{\tr^{(M)}}:= \tr^{(M)}(O^*O)
\eeq
and the following bound (from straightforward manipulations using cyclicity of the trace): for an orthogonal projection $P$,
\beq \label{eq: noncomm hoelder}
\norm O_1PO_2 \norm_{\tr^{(M)}}  \leq (\tr^{(M)} (P))^{1/2}   \norm O_1\norm \norm O_2\norm.
\eeq
\begin{figure}[h!] 
\vspace{0.5cm}
\begin{center}
\def\svgwidth{8cm}
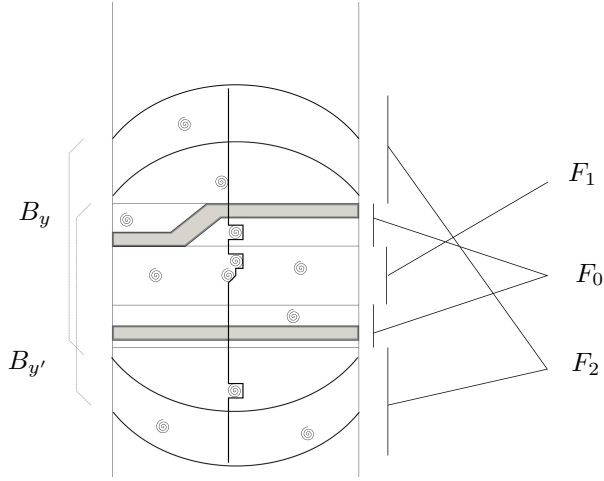
\caption{Spatial structure of the sets $F_0, F_1,F_2$ as used in Case C of the proof of Proposition \ref{proposition: commutator left right}. 
The shaded area is the set $\tilde F_0= \tilde F_0 \cap B_y = \tilde F_0 \cap B_{y'}$. The small spirals indicate the sets $S(\rho)$ with $\rho \in \caP'(\eta)$ or $\rho \in \caP'(\eta')$.}
\label{fig: leftright1} 
\end{center}
\end{figure}

\begin{proposition}\label{proposition: commutator left right}
Let $ Z_{\links,y}, Z_{\rechts,y}$ be as described above. If $r$ is chosen large enough, the following properties hold: 
\begin{enumerate}
\item  The supports satisfy $s(Z_{\links, y}),  s(Z_{\rechts, y}) \subset B_y$ and the  operators are bounded as
\beq
\norm Z_{\links, y} \norm  \leq C(r) M^C, \qquad    \norm Z_{\rechts, y} \norm \leq  C(r) M^C
\eeq
\item The 'left' and 'right' operators commute up to a sparse term:
\beq
\norm [Z_{\links,y}, Z_{\rechts,y'}]  \norm_{\tr^{(M)}}  \leq   C(r) M^{C-cr}  
\eeq
\item The 'left' part does not go too far to the right:  for any operator $O_\rechts$ such that $\proj_1(s(O_\rechts)) >a+ r^2/2$,
\begin{align}
 & \norm [Z_{\links,y}, O_\rechts ] \norm_{\tr^{(M)}}  \leq 
C(r) M^{C-c r}  \norm O_\rechts \norm    \label{eq: z links with o rechts}
\end{align}
 Analogously, for any $O_\links$ such that $\proj_1(s(O_\links)) < a- r^2/2$,
\begin{align}
 & \norm [Z_{\rechts,y}, O_\links ] \norm_{\tr^{(M)}}  \leq 
C(r) M^{C-c r}  \norm O_\links \norm      \label{eq: z rechts with o links}
\end{align}
\end{enumerate}
\end{proposition}
\begin{proof}[Proof of Propostion \ref{proposition: commutator left right}]
Let us start with some easy remarks.
\begin{itemize}
\item [$a)$] If $\mu \in \caF^{(y)},\mu' \in \caF^{(y')}$, then   $[P_\mu,P_{\mu'}]=0$ because both projections are diagonal in the same basis.   For the same reason
$[P_\mu, \scrD(A)]=0$.
\item [$b)$] For $\rho$ such that ${s(\rho)} \subset B_{y}$, $[P_{\mu}, W_\rho]=0$ for any $\mu \in \caF^{(y)}$, by the definition of $\caF^{(y)}$, i.e.\ Definition \ref{def: partition}. 
\item [$c)$] If ${s(\rho)} \cap S(\rho')= \emptyset$ and $s(\rho') \cap S(\rho)=\emptyset$, then $[W_\rho,W_{\rho'}]=0$.
\end{itemize}
 The bounds in 1) follow trivially from the bounds following Definition \ref{def: moves}.  
We now prove point 2). 

To estimate the commutator $[Z_{\mathrm{L},y}, Z_{\mathrm{R},y'}] $, we first consider the term
\beq \label{eq: comm to be considered}
[ \sum_{\mu \in \caF^{(y)}}  P_\mu      \sum_{\scriptsize{\begin{array}{c} \rho \in \caP^{(y)}: y( {s(\rho)}) =y \\   {s(\rho)} \subset \links_y(\mu) \end{array} } }    W_{\rho},   \sum_{\mu' \in \caF^{(y')}}  P_{\mu'}      \sum_{\scriptsize{\begin{array}{c} \rho' \in \caP^{(y')}: y( {s(\rho')}) =y'  \\   {s(\rho')} \not\subset \links_{y'}(\mu') \end{array} } }         W_{\rho'}      ]
\eeq
where we added primes to the variables in the second term of the commutator for clarity.
We consider four cases: \\ 

\noindent \textbf{Case A:} $y=y'$\\[1mm]
If $\mu\neq \mu'$, then \eqref{eq: comm to be considered} vanishes by  $b)$ above and $P_\mu P_{\mu'}=P_\mu \delta_{\mu,\mu'}$.  To see that 
\beq
[P_\mu W_\rho, P_{\mu} W_{\rho'}] =0, \qquad   \textrm{if} \, \,  {s(\rho)} \subset \links(\mu),   {s(\rho')} \not \subset \links(\mu)
\eeq
we use that, by the definition of the Left-Right decomposition, $S(\rho) \cap S(\rho')=\emptyset$ and therefore $[W_\rho,W_{\rho'}] =0$, by $c)$ above.\\

\noindent \textbf{Case B:}  $ 3r < \str y-y'\str \leq 2r^2-r$\\[1mm]
We have that  ${s(\rho)} \subset B_{y'}$ and ${s(\rho')} \subset B_{y}$. Therefore, by $a),b)$ above, all commutators  involving a projection $P_\mu$ vanish, just as in \textbf{case A}, and hence it suffices to consider the commutator
\beq \label{eq: contributing commutators}
 [W_{\rho}, W_{\rho'}]
\eeq
Because $\str y -y' \str > 3r$, this trivially vanishes by $c)$ above. \\

\noindent \textbf{Case C:}  $ 0<  \str y-y'\str \leq 3r $\\[1mm]
By the same reasoning as in \textbf{Case B}, it suffices to consider the commutator \eqref{eq: contributing commutators}.
Define
\begin{align}
F_0 & := \{ x \in B_y \cup B_{y'}:    \,        r^2/4 < \str y(x)- \frac{y+y'}{2} \str \leq  r^2/2    \},     \\[2mm]
 F_1 & := \{ x \in B_y \cup B_{y'} :  \,   \str y(x)- \frac{y+y'}{2} \str \leq  r^2/4    \},  \\[2mm]
  F_2 & :=  (B_y \cup B_{y'}) \setminus (F_0 \cup F_1) 
\end{align}
If either $\mu$ or $\mu'$ is good, then the conditions of Lemma \ref{lem: frame} are satisfied.  Moreover, $S(\rho),S(\rho') \subset F_1$.  If $S(\rho)\cap S(\rho') =\emptyset$, then there is nothing to prove because \eqref{eq: contributing commutators} vanishing trivially.  If $S(\rho)\cap S(\rho') \neq \emptyset$, then Lemma \ref{lem: frame} tells us that either both $S(\rho), S(\rho') $ are included in the 'left' $\links$-set, or both are not included in the $\links$-set.  Therefore, their commutator does not appear in \eqref{eq: comm to be considered}.

So we conclude that the only non-vanishing contribution to \eqref{eq: comm to be considered} originates from pairs $(\mu,\mu')$ such that  none of them is good. 
We recast the sum over such pairs, for fixed $\rho,\rho'$, as
\beq
\sum_{\mu \in \caF^{(y)} \setminus \caF^{(y)}_{\mathrm{g}}} \sum_{\mu' \in \caF^{(y')} \setminus \caF^{(y')}_{\mathrm{g}}} [P_\mu         W_{\rho},   P_{\mu'}    W_{\rho'} ] =  [\bar P_{\caF^{(y)}_{\mathrm{g}}}         W_{\rho},   \bar P_{\caF^{(y')}_{\mathrm{g}}}      W_{\rho'} ] 
\eeq
with the projections $\bar P_{\caF^{(y)}_{\mathrm{g}}} $ (on non-good configurations) as in \eqref{def: projections on bad}.  Then we bound  by \eqref{eq: noncomm hoelder}:
\beq \label{eq: bounding comm in norm}
 \norm [\bar P_{\caF^{(y)}_{\mathrm{g}}}         W_{\rho},   \bar P_{\caF^{(y')}_{\mathrm{g}}}      W_{\rho'} ]  \norm_{\tr^{(M)}}  \leq   2 \norm  W_{\rho'} \norm  \norm  W_{\rho} \norm  \norm \bar P_{\caF^{(y')}_{\mathrm{g}}} \norm  \,     (\tr^{(M)} (\bar P_{\caF^{(y)}_{\mathrm{g}}} ) )^{1/2}
\eeq
  The norms on the RHS are estimated as $ \norm W_\rho \norm \leq  C(r) M^{C} $ by the properties stated following Definition \ref{def: moves}, $ \norm \bar P_{\caF^{(y')}_{\mathrm{g}}} \norm \leq 1 $ using that $\bar P_{\caF^{(y')}_{\mathrm{g}}}$ is a projection,  and  the trace  
is bounded by   $C(r) M^{-cr}$, see  \eqref{eq: bad is sparse 3}.  Finally the sum over pairs $\rho,\rho'$ is bounded by $C(r)$, the number of terms in the sum.\\

\noindent\textbf{Case D:}  $ 2r^2-r \leq \str y-y'\str \leq 2r^2$\\[1mm]
In this case there is no reason for  $ [ W_{\rho}, P_{\mu'}] $ to be small since  it can happen that ${s(\rho)} \not \subset B_{y'}$. 
Instead we recast \eqref{eq: comm to be considered} as
\beq \label{eq: comm with proj}
 \sum_{\rho,\rho': y( {s(\rho)}) =y, y( {s(\rho')}) =y' } [ P_{\caF^{(y)}(\rho \to \links)}    W_{\rho},   \bar P_{\caF^{(y')}(\rho' \to \links)}     W_{\rho'}      ]
\eeq
with
\beq
P_{\caF^{(y)}(\rho \to \links)} :=  \sum_{\mu \in \caF^{(y)}:  {s(\rho)} \subset \links(\mu)}  P_\mu,\qquad   \bar P_{\caF^{(y)}(\rho \to \links)}  :=   \lone - P_{\caF^{(y)}(\rho \to \links)} 
\eeq
Since $S(\rho) \cap S(\rho') =\emptyset$,  \eqref{eq: comm with proj} is rewritten as a sum over $W_\rho [ P_{\caF^{(y)}(\rho \to \links)}  ,    W_{\rho'}      ] + [W_{\rho},   P_{\caF^{(y')}(\rho' \to \rechts)}  ]W_{\rho'}$. Let us look more generally at an expression of the form
\beq  \label{eq: template far frontier}
[ P_{\caF^{(y)}(\rho \to \links)}  ,  O   ] , \qquad \distance(s(O), s(\rho)) \geq r^2/2
\eeq
Then we claim 
\beq \label{eq: vanishing commutator}
[ P_{\caF^{(y)}(\rho \to \links)}  ,   P_{\caF^{(y)}_{\mathrm{g}}}  O  P_{\caF^{(y)}_{\mathrm{g}}}      ] =0.
\eeq
To check this, we first write  
\beq   \label{eq: decompose p into eta}  P_{\caF^{(y)}_{\mathrm{g}}}O P_{\caF^{(y)}_{\mathrm{g}}}= \sum_{\eta,\eta': \mu(\eta), \mu(\eta') \in \caF^{(y)}_{\mathrm{g}}} P_\eta OP_{\eta'},
\eeq
and, using the formula
\beq \label{eq: projectors}
P_{\caF^{(y)}(\rho \to \links)}P_{\eta} =P_{\eta}  P_{\caF^{(y)}(\rho \to \links)} =  \chi\big( \mu(\eta) \in \caF^{(y)}:  {s(\rho)} \subset \links(\mu) \big)\, P_{\eta},
\eeq
 we note that any nonzero contribution to \eqref{eq: vanishing commutator} has to come from pairs $(\eta,\eta')$ in \eqref{eq: decompose p into eta} such that
\begin{enumerate}
\item $\eta_{s(O)^c}= \eta'_{s(O)^c}$.
\item  $\mu(\eta),\mu'(\eta')$ are good.
\item  $s(\rho) \subset \links(\mu(\eta)), s(\rho) \not\subset \links(\mu(\eta')) $ or vice versa (see \eqref{eq: projectors}).
\end{enumerate}
Since $\distance (s(O),s(\rho)) >r^2/2$,  we can construct a partition $\{F_0,F_1,F_2\}$ as in Lemma \ref{lem: frame} (applied with $y'=y$) such that ${S(\rho) \cap B_y} \subset F_1$ and $s(O) \subset F_2$.  Then Lemma \ref{lem: frame} implies that there are no pairs that can contribute to the commutator in  \eqref{eq: vanishing commutator} and hence  \eqref{eq: vanishing commutator} holds. 
  Therefore, 
\beq
[ P_{\caF^{(y)}(\rho \to \links)}  ,    O      ] =    [ P_{\caF^{(y)}(\rho \to \links)}  ,   \bar P_{\caF^{(y)}_{\mathrm{g}}}  O    P_{\caF^{(y)}_{\mathrm{g}}}  + P_{\caF^{(y)}_{\mathrm{g}}}  O   \bar P_{\caF^{(y)}_{\mathrm{g}}}  + \bar P_{\caF^{(y)}_{\mathrm{g}}} O   \bar P_{\caF^{(y)}_{\mathrm{g}}}  ] 
\eeq
and proceeding as in \eqref{eq: bounding comm in norm},  this is bounded by $\norm O \norm M^{-cr}$.  Plugging in $O=W_\rho$, we obtain again the bound $C(r) M^{-cr}$.\\

\noindent \textbf{Case E:}  $\str y-y'\str >2 r^2$\\[1mm]
 In this case, the commutator vanishes obviously since the two terms act on disjoint regions. \\

Finally, if we consider instead of the commutator \eqref{eq: comm to be considered}, one of the terms containing $\scrD(A)$, then we can repeat the above reasoning (though with many simplifications) and obtain the same result.  This proves point $2)$.  \\  

And now to point $3)$:  The estimates \eqref{eq: z links with o rechts} and   \eqref{eq: z rechts with o links} are of course analogous and we consider only the former. 
  Again we first consider only the terms $W_\rho$ in $Z_{\links,y}$.  
Let us first assume that $\proj_1(s(\rho)) \geq a+r^2/8$.  
Since  $s(\rho) \subset \links(\mu)$, this means that the left $\links(\mu)$ set extends far to the right, and hence, if $c_1$ in Definition \ref{def: good} is small enough, this implies that $\mu$ cannot be good. Just as in the estimates above, non-good $\mu$ give a contribution that is sufficiently small in $\norm \cdot \norm_{\tr^{(M)}}$ (see \eqref{eq: bounding comm in norm}) and hence we can disregard them. 
Then, we are left with $\rho$ such that $\proj_1(s(\rho)) < a+r^2/8$. In that case  $[W_{\rho}, O_\rechts]=0$ because of disjoint supports and we only have to estimate 
\beq
[ P_{\caF^{(y)}(\rho \to \links)}  ,  O_\rechts   ], \qquad  \distance(s(\rho),s(O_\rechts)) \geq 3r^2/8
\eeq
The same expression was considered in \eqref{eq: template far frontier} in \textbf{Case D} above, except that there we had the distance $r^2/4$ instead of $3r^2/8$.  However, the same reasoning applies with $3r^2/8$ as well, provided that $r$ is large enough. 

The terms originating from $\scrD(A)$ are treated analogously. 
\end{proof}

\section{Proof of Theorems \ref{thm: splitting of current} and \ref{thm: vanishing conductivity} }\label{sec: proofs}

The crux of the proof of Theorem \ref{thm: splitting of current} is a left-right splitting of the Hamiltonian in the strip $\bbS$. This will be achieved below following Proposition \ref{thm: splitting of hamiltonian} which itself relies heavily on Proposition \ref{proposition: commutator left right}.   This left-right decomposition will allow us to study the local currents as sums of operators with small variance. We do this in Section \ref{sec: decomposition of the current}.   The proof of Theorem \ref{thm: vanishing conductivity} is then a simple consequence, and it is given in the final Section \ref{sec: proof of main theorem}. 
In the present section, we will also remove the high-energy cutoff.  The terms in the Hamiltonians and currents (e.g.\ he first term on the right hand side in \eqref{eq: splitting of h in strip} that appear because of this, can always be treated in a simple-minded way, exploiting the fact that the Gibbs state gives very small weight to high-energy states. 

\subsection{States and Hilbert-Schmidt norms}\label{sec: states and norms}
We recall/introduce the states
\beq \label{def: recall and introduce states}
\om_{\be, \bigvolume} (O) :=  \frac{\Tr \e^{-\be H_\bigvolume} O}{\Tr  \e^{-\be H_\bigvolume}},\qquad   \om_{\be, \bigvolume,0} (O) :=  \frac{\Tr \e^{-\be H^{(0)}_\bigvolume} O}{\Tr  \e^{-\be H^{(0)}_\bigvolume}}, \qquad O \in \caB(\caH_\bigvolume)
\eeq
In what follows, we suppress the dependence on $\be,\bigvolume$ since all of our estimates will be uniform in $\bigvolume$ and in $\be$ whenever $\be$ is small enough, hence we simply write $\om(\cdot), \om_0(\cdot)$.    
We define the associated Hilbert-Schmidt norms
\beq
\norm O \norm_{\om}:= \om(O^* O)^{1/2}, \qquad   \norm O \norm_{\om_0}:= \om_0(O^* O)^{1/2},
\eeq
completely analogous to the the norm \eqref{def: hs norm trace} that was the Hilbert-Schmidt norm associated to the state $\tr^{(M)}(\cdot)$. 
We denote the covariance of observables by 
\beq
\om(O_1 ; O_2) := \om(O_1 O_2)- \om(O_1)\om(O_2)  = \om ((O_1-\om(O_1)) (O_2-\om(O_2)))
\eeq

\subsection{Decomposition of the Hamiltonian in a strip} \label{sec: splitting of ham in strip}
First we split the Hamiltonian $H=H_\bigvolume$ as 
\beq
H= H^{(\links)}_{\bbS^c}+ H^{(\rechts)}_{\bbS^c}  + H_\bbS
\eeq
where
\beq
H^{(\links)}_{\bbS^c}=\sum_{x: \,x_1< a-r^2+2}H_x  ,\qquad   H^{(\rechts)}_{\bbS^c}=  \sum_{x:\, x_1 > a+ r^2-2}H_x
\eeq
and\footnote{Note that $x$ is included in $H_{\mathbb{S}}$ if the star domain $\{x': \str x-x' \str \leq 1\}$ is in $\bbS$,  to be consistent with the definition of the  resonant strip Hamiltonian in Section \ref{sec: analysis invariant subspaces} }
\beq
H_\bbS = \sum_{x: \, a-r^2+2 \leq  x_1 \leq a+r^2-2 } H_x
\eeq

From the technical point of view, the main step in our analysis is the splitting of the Hamiltonian  $H_\bbS$ as a sum of three terms.   It is accomplished in the next proposition, which relies crucially on Proposition \ref{thm: renormalization} and  Proposition \ref{proposition: commutator left right}.    
Recall the regions $\tilde B_y$ introduced in \eqref{def: balls}. 
\begin{proposition}[Splitting of Hamiltonian]\label{thm: splitting of hamiltonian}
Let $M= \be^{-(1+c(r))/q}$. 
For sufficiently large $r \in \bbN$,  there is a decomposition
\beq \label{eq: splitting of h in strip}
 H_{\bbS} =  \sum_{x \in \bbS} (1-\caP_{\leq M}) H_x  +    \pot_{\bbS}(\scrA)  +  \sum_{y \in \bbZ^{d-1}:\,  (a,y) \in \bigvolume} \tilde Z_{y}
\eeq
such that 
\begin{enumerate}
\item The potentail $\scrA$ consists only of low-energy terms, $\caP_{\leq M}(\scrA)=\scrA$ and
\beq
\interleave\scrA \interleave_{C \be^{-c},1}  \leq C(r) \be^{-C+cr}
\eeq
\item The operators $\tilde Z_{y}$ can be split as
\beq
\tilde Z_y = \tilde Z_{\links, y} + \tilde Z_{\rechts, y}
\eeq
such that 
\beq s(\tilde Z_{\links, y}) \subset   \tilde B_{y},  \qquad  \caP_{\leq M}(\tilde Z_{\links, y})=\tilde Z_{\links, y},   \qquad \norm  \tilde Z_{\links, y} \norm  \leq  C(r) \be^{-C}, \eeq idem for $\tilde Z_{\rechts, y}$, and 
\beq \label{eq: mother of all commutators}
\norm  [\tilde Z_{\links,y}, \tilde Z_{\rechts,y'}] \norm_{\om_0}  \leq C(r) \be^{cr-C}
\eeq
Furthermore, the left/right part is  supported on the left/right side of the strip $\bbS$ in the following sense:
\begin{align}
 & \norm  [\tilde Z_{\links, y}, O_\rechts] \norm_{\om_0} \leq  C(r) \be^{cr-C} \norm O_\rechts \norm      \label{eq: comm with boundary}  \\[2mm]
 &\norm  [\tilde Z_{\rechts, y}, O_\links] \norm_{\om_0} \leq  C(r) \be^{cr-C} \norm O_\links \norm      \label{eq: comm with boundary 2} 
\end{align}
for any operators $O_\rechts, O_\links$ satisfying  $\str s(O_\rechts)\str, \str s(O_\links) \str \leq C$ and 
\beq
 a+(3/4)r^2  < \proj_1(s(O_\rechts)), \qquad     a-(3/4)r^2  > \proj_1(s(O_\links)). 
\eeq
\end{enumerate}
 \end{proposition}

\begin{proof}
We recall  the operator $Z_y=Z_{\links,y}+ Z_{\rechts,y} $ as defined in \eqref{def: zway} and we also  use the notation $Z_y$ to denote the potential $\scrA_{Z_y}$ defined by 
\beq
\scrA_{Z_y}(A) := \begin{cases}  Z_y & A= B_y \,\, \text{for some $y$}  \\     0 & \text{otherwise} \end{cases}
\eeq
First, we will show that 
\beq \label{eq: reassembled ham}
H_{\bbS} =  \sum_{x \in \bbS} (1-\caP_{\leq M}) H_x  + \pot_\bbS\big( \caK(\hat \scrF )) + \sum_{y} \pot_\bbS\big( \caK(Z_{ y}) \big)
\eeq
where (also below) sums over $y$ are understood to range over the set $\{y \in \bbZ^{d-1}:\,  (a,y) \in \bigvolume\}$. 
By the decomposition in  \eqref{eq: hat decomp}  and the relation $\caK \caK^{-1}=1$ from Section \ref{sec: unitary trafo}, 
\beq    \label{eq: silly one} H_{\bbS}-\sum_{x \in \bbS} (1-\caP_{\leq M})(H_x)= \pot_{\bbS}(\scrE) = \pot_{\bbS}(\caK(\scrD+\hat\scrG+ \hat \scrF))   \eeq
Furthermore,  
\beq \label{eq: silly two}
 \pot_\bbS\big( \caK(\scrD+ \hat\scrG) \big) =    U_{\bbS}\pot_\bbS(\scrD+ \hat\scrG)U^*_{\bbS} =  U_{\bbS}ZU^*_{\bbS} = \sum_y U_{\bbS}Z_yU^*_{\bbS}  =   \sum_{y} \pot_\bbS\big( \caK(Z_{ y}) \big)
\eeq
The first equality is  \eqref{eq: cak and unitary},  the second is the definition of $Z$ (see  \eqref{eq: def z}), the third is $Z=\sum_y Z_y$ (see \eqref{def: zway}) and the fourth is \eqref{eq: cak and unitary restricted} with $O=Z_y$.
Hence, \eqref{eq: reassembled ham} follows from \eqref{eq: silly one} by the equality of the first and last expression in \eqref{eq: silly two}.

Next, we subtract from the terms with $Z_{y}  $ a part that is small enough to be included into $\scrA$: 
We split, according to \eqref{eq: decomp caic and u}, 
\beq
H_\bbS( \caK(Z_{ y} ))   = U_{\tilde B_y}   Z_{ y}  U^*_{\tilde B_y}  +   H_\bbS( \caI_{\tilde B^c_y }\caK( Z_{y} )) 
\eeq
with $\caI_{B^c}$  for a set $B$ the restriction of potentials introduced in Section \ref{sec: unitary trafo}. 
Since $s( Z_{ y})=B_y$, we have $\distance( \tilde B^c_y, s( Z_{y}) ) \geq r^2 \geq  c  \str s( Z_{ y}) \str $, and therefore we conclude from \eqref{eq: bound cutoff pot} that
\beq   \label{eq: bound for small part in z}
\interleave \caI_{\tilde B^c_y }\caK( Z_{y} ) \interleave_{M^c,1} \leq C(r) M^{-c r^2+C}
\eeq
This estimate was in fact the reason to choose  $(2r)^{2d}$ as the radius of the balls $\tilde B_y$. 
As announced, we now define
\beq
\scrA :=  \caK(\hat\scrF) + \sum_{y}  \caI_{\tilde B^c_y }\caK( Z_{ y} )
\eeq
and we check that it satisfies the bound claimed in item 1.\ of the Proposition.  Indeed, for the second  term on the right hand side this is the bound  in \eqref{eq: bound for small part in z}. For the first term, the bound follows from  \eqref{eq: cak preserves bounds} and item 2.  of Proposition \ref{thm: renormalization} (upon plugging $\delta=M^{-\gamma_1}$). \\

 Now we move to the left-right splitting. Set 
\beq
V:= U_{\tilde B_y} , \qquad \tilde Z_{\links, y} :=  V   Z_{\links, y}  V^* , \qquad    \tilde Z_{\rechts, y} :=V   Z_{\rechts, y}  V^*.
\eeq
By the unitarity of $V$, 
\begin{align}
\norm [\tilde Z_{\links, y}, \tilde Z_{\rechts, y'}] \norm_{\tr^{(M)}} =   \norm [ Z_{\links, y}, Z_{\rechts, y'}] \norm_{\tr^{(M)}}
\end{align}
We can bound the right hand side by Proposition \ref{proposition: commutator left right} 2) and we thus obtain the bound \eqref{eq: mother of all commutators}, but in the norm $\norm \cdot \norm_{\tr^{(M)}}$ rather than $\norm_{\om_0}$. In Lemma \ref{lem: converting sparse bounds} below, we explain how to relate these norms.

Finally, we  have to control $ [\tilde Z_{\links, y}, O_\rechts] $ with $O_\rechts$ (as given in the statement of the proposition). Set
\beq 
W : = U_{B'_y}, \qquad \text{with}\,\, B'_y = \tilde B_y \cap  \{x :\, \distance(x, s(O_\rechts) ) \leq r \}.
\eeq
and calculate 
\begin{align}
  [\tilde Z_{\links, y}, O_\rechts] &= V  [ Z_{\links, y},  V^*  O_\rechts  V  ] V^*  \nonumber   \\[2mm]
    &= V [ Z_{\links, y},  W^* O_\rechts  W   ] V^*  + V [ Z_{\links, y},  (V^* O_\rechts  V-  W^* O_\rechts  W)   ] V^*    \label{eq: comm acro 3}
\end{align}
We bound the last line as follows, using unitarity of $V,W$,  
\begin{align}
\norm   [\tilde Z_{\links, y}, O_\rechts]  \norm_{\tr^{(M)}}  &\leq   \norm [ Z_{\links, y}, W^* O_\rechts  W ]\norm_{\tr^{(M)}}    +   
2\norm Z_{\links, y} \norm \norm  V^* O_\rechts  V-  W^* O_\rechts  W  \norm      \label{eq: comm acro 4}
\end{align}
For the first term, we use that the operator $W^* O_\rechts  W$ is supported to the right of $a+(3/4)r^2-r$ and  hence Proposition \ref{proposition: commutator left right} 3) gives, for large enough $r$,   the bound $C(r) M^{C-cr} \norm W^* O_\rechts  W \norm$. 
For the second term, we use, from Lemma \ref{lem: local unitaries}, that
\beq \label{eq: comm acro 1}
\norm  V^* O_\rechts  V-  W^* O_\rechts  W \norm \leq C(r) M^{-cr } \norm O_\rechts \norm 
\eeq 
and the bound $\norm Z_{\links, y} \norm  \leq C(r)M^C$ from Proposition \ref{proposition: commutator left right} 1. Putting these bound together, we get
%
%
\beq
\norm  [\tilde Z_{\links, y}, O_\rechts] \norm_{\tr^{(M)}} \leq  C(r) M^{C-cr} \norm O_\rechts \norm
\eeq
and analogously for  $ [\tilde Z_{\rechts, y}, O_\links]$. 
To finish the proof of item 2., it remains to argue that the bounds on operators in $\norm \cdot \norm_{\tr^{(M)}}$ can be converted to bounds in the norm $\norm \cdot \norm_{\om_0}$:
\begin{lemma} \label{lem: converting sparse bounds}
For any  operator $O$ with $\caP_{\leq M}(O)=O$ and $M= \be^{-(1+\gamma_c)/q} $,  with $\gamma_c>0$, 
\beq \label{eq: comparison of states}
\norm O \norm_{\om_0} \leq ( C \be^{-\gamma_c/q})^{\str s(O) \str}   \norm O \norm_{\tr^{(M)}}, \qquad   
\eeq
\end{lemma}
\begin{proof}
Since the density matrices of the state $\om_0$ and $\tr^{(M)}$ are both product and diagonal in the same basis, this boils down to the estimate
\beq
M (\sum_{\eta \geq 0} \e^{-\beta \eta^{q}}  )^{-1}  \leq  C M  \be^{1/q} \leq  C \be^{-\gamma_c/q}. 
\eeq
\end{proof}
Hence, to get the bounds (\ref{eq: mother of all commutators}, \ref{eq: comm with boundary}, \ref{eq: comm with boundary 2}) from the corresponding bounds with $\norm \cdot \norm_{\tr^{M}}$ and Lemma \ref{lem: converting sparse bounds}, we have to choose the exponent ${\gamma_c/q}$ such that $ cr - {(\gamma_c/q)}r^{2d} > c'r $ (for some $c'>0$ depending on $c$). This is achieved by decreasing  the cut-off exponent $\gamma_c=\gamma_c(r)$ sufficiently fast as $r$ grows. 
\end{proof}

Finally, we define a left-right decomposition of the full strip Hamiltonian
\beq \label{eq: first tilded splitting}
 H_{\mathbb{S}}=  \tilde H^{(\links)}_{\mathbb{S}}+ \tilde H^{(\rechts)}_{\mathbb{S}} 
\eeq
Since only the splitting of the $\tilde Z_{y}$-terms in the Hamiltonian matters in the end, we can simply assign all other terms to, say, the right part, and define
by setting simply
\beq
 \tilde H^{(\links)}_{\mathbb{S}}: =   \sum_{y} \tilde Z_{\links, y} , \qquad   \tilde H^{(\rechts)}_{\mathbb{S}} = H_{\mathbb{S}} -\tilde H^{(\links)}_{\mathbb{S}} 
\eeq
The tildes in this expression serve to distinguish this splitting from the 'naive' left-right decomposition  $H_\bbS=  H^{(\links)}_{\mathbb{S}}+  H^{(\rechts)}_{\mathbb{S}}$ with
\beq
  H^{(\links)}_{\mathbb{S}} := \sum_{x: \,  a-r^2+2 \leq  x_1 \leq a} H_x, \qquad     H^{(\rechts)}_{\mathbb{S}} :=  \sum_{x: \,  a < x_1\leq a+r^2-2} H_x
\eeq
corresponding to the 'naive' left-right decomposition of the total Hamiltonian $H=  H^{(\links)}+  H^{(\rechts)}$
\beq
H^{(\links)}=    H^{(\links)}_{\mathbb{S}}+   H^{(\links)}_{\mathbb{S}^c},\qquad  H^{(\rechts)}=    H^{(\rechts)}_{\mathbb{S}}+   H^{(\rechts)}_{\mathbb{S}^c}
\eeq
which was already introduced in Section \ref{sec: currents} to define the current.
\subsection{Decomposition of the current}\label{sec: decomposition of the current}

Our goal is to estimate 
\beq
J = J_{\bbH} = \i  [H,H^{(\links)}]
\eeq
where $H=H^{(\links)}+H^{(\rechts)}$ by using the left-right decomposition of the Hamiltonian constructed in Section \ref{sec: splitting of ham in strip}. 
Namely, we set
\beq
\tilde H^{(\links)} =  \tilde H_{\bbS}^{(\links)} +  H_{\bbS^c}^{(\links)},\qquad  \tilde H^{(\rechts)} =  \tilde H_{\bbS}^{(\rechts)} +  H_{\bbS^c}^{(\rechts)}
\eeq
such that
\beq
H=\tilde H^{(\links)} +\tilde H^{(\rechts)}.
\eeq
Hence by splitting
\beq
H^{(\links)} =   \tilde O_{\bbS}  +  \tilde H^{(\links)},\qquad  \tilde O_{\bbS}:=  (H^{(\links)}_{\bbS} -   \tilde H^{(\links)}_{\mathbb{S}}  ) 
\eeq
we get 
\begin{align}
-\i J = -\i J^{(1)}- \i J^{(2)}  &:= [H,  \tilde O_{\bbS} ]  +  [H, \tilde H^{(\links)}]   \label{eq: precise current decomp} \\
 &= [H,  \tilde O_{\bbS} ]  +  [\tilde H^{(\rechts)}, \tilde H^{(\links)}]   \\
&= [H,  \tilde O_{\bbS} ] +    [H^{(\rechts)}_{\bbS^c}, \tilde H_{\bbS}^{(\links)}] +   [\tilde H^{(\rechts)}_{\bbS}, \tilde H_{\bbS}^{(\links)}]  +   [\tilde H^{(\rechts)}_{\bbS}, H_{\bbS^c}^{(\links)}]  
\end{align}
For future use,  note that $\om(J)=\om(J^{(1)})=0$ by stationarity, and therefore also $\om(J^{(2)})=0$.  
We recognize the expression for $J$ in Theorem \ref{thm: splitting of current} since the commutators on the right-hand side are sums of local operators. 
For further discussion, let us make this explicit by defining
\begin{align}
V^{(\rechts,j)}_A & := \left\{\begin{array}{lll}  (1-\caP_{\leq M})(H_x)   &  \text{if}\,\,  A=\{x': \, \str x' - x \str \leq 1 \} \subset \bbS    & j=1  \\[2mm]
\scrA(A) &  \text{if}\,\,  A \subset \bbS    & j=2  \\[2mm]  
 \tilde Z_{\rechts,y} \qquad &  \text{if}\,\,   A= \tilde B_y \,\,   \qquad & j=3  \\[2mm] 
 \caP_{\leq M}(H_x)  \qquad &    \text{if}\,\,  A=\{x': \, \str x' - x \str \leq 1 \} \,\, \text{with $x_1-a =r^2-1\,\, \text{or}\,\, r^2 $}   \qquad & j=4   \\[2mm] 
(1-\caP_{\leq M})(H_x)  &    \text{if}\,\,  A=\{x': \, \str x' - x \str \leq 1 \} \,\, \text{with $x_1-a =r^2-1\,\, \text{or}\,\, r^2 $}   \qquad &  j=5 
\end{array} \right. \nonumber  
\end{align}
for some $x,y$ and  $V^{(\rechts,j)}_A=0$ in all other cases. Similarly,
\begin{align}
V^{(\links,j)}_A & := \left\{\begin{array}{lll}  
 \tilde Z_{\links,y} \qquad &  \text{if}\,\,   A= \tilde B_y \,\,   \qquad & j=3  \\[2mm] 
 \caP_{\leq M}(H_x)  \qquad &    \text{if}\,\,  A=\{x': \, \str x' - x \str \leq 1 \} \,\, \text{with $-(x_1-a) =r^2-1\,\, \text{or}\,\, r^2 $}   \qquad & j=4   \\[2mm] 
(1-\caP_{\leq M})(H_x)  &    \text{if}\,\,  A=\{x': \, \str x' - x \str \leq 1 \} \,\, \text{with $-(x_1-a) =r^2-1\,\, \text{or}\,\, r^2 $}   \qquad &  j=5 
\end{array} \right. \nonumber  
\end{align}
for some $x,y$ and  $V^{(\links,j)}_A=0$ in all other cases.  The assymetry between $\rechts$ and $\links$ in these formulas is due to the arbitrary choice, made following \eqref{eq: first tilded splitting}, to assign all nonessential terms to the $\rechts$ part. 
Next, we set
\beq
-\i \tilde I^{(j,j')}_{A,A'} : =[V^{(\rechts,j)}_A, V^{(\links,j')}_{A'} ]
\eeq
such that, from the decomposition of Proposition \ref{thm: splitting of hamiltonian} and the definition of $J^{(2)}$ above, we indeed have
\beq
J^{(2)}  = \sum_{A,A': A \cap A' \neq \emptyset} \sum_{j,j'= 1, \ldots, 5}  \tilde I^{(j,j')}_{A,A'}
\eeq
This is straightforward to check; the terms with $j=4,5$ are those originating from terms $H_x$ such that the star domain $\{x': \, \str x' - x \str \leq 1 \}$ has overlap both with $\bbS$ and $\bbS^c$. Those terms are included in $H^{(\links)}_{\bbS^c}$ or $H^{(\rechts)}_{\bbS^c}$ but they do not commute with $H_x$ for $x$ inside the strip, hence they contribute to the current.  The terms with $j=1,2,3$ correspond to the three terms on the right hand side of \eqref{eq: splitting of h in strip}. 
The tilde on $\tilde I^{(j,j')}_{A,A'}$ is to distinguish it from the centered operators $I^{(j,j')}_{A,A'}:=\tilde I^{(j,j')}_{A,A'}-\om(I^{(j,j')}_{A,A'})$ that will be used later.

Next, we  establish the desired properties of these local terms.

\subsection{Classification of current operators and proof of Theorem \ref{thm: splitting of current}}\label{subsec: Classification of current operators}
We classify the `current' operators  $
\tilde I^{(j,j')}_{A,A'}$ introduced above. 
\begin{lemma}\label{lem: all currents are k}
For any $j,j'$ and $A,A'$, the operator
$
\tilde I^{(j,j')}_{A,A'}$  can be written as 
\beq \label{eq: i as sum}
\tilde I^{(j,j')}_{A,A'} = \sum_{i=1}^{C}  K^{(j,j',i)}_{A,A'}
\eeq
where each of the operators $K^{(j,j',i)}_{A,A'}$ is of the $K$-type introduced in Section \ref{sec: observables} (Appendix), such that, for any of these operators $K=K^{(j,j',i)}_{A,A'}$, we have $s(K) \subset A \cup A'$ and 
\beq \label{eq: desired bound on k}
w(K) \leq C(r) \be^{-C+c r+c(r) \str A \cup A' \str  }.
\eeq
with $w(K)$ as defined in \eqref{eq: definition of w k} (Appendix). 
Similarly, for any pair of the operators   $
\tilde I^{(j,j')}_{A,A'}, \tilde I^{(j'',j''')}_{A'',A'''} $ (not necessarily distinct) with $(A \cup A') \cap (A'' \cup A''') \neq \emptyset $,  the product
\beq \label{eq: products of i}
(\tilde I^{(j,j')}_{A,A'})^*  \tilde I^{(j'',j''')}_{A'',A'''}
\eeq
is a sum of $C$ operators of the  $K$-type satisfying  $s(K) \subset A \cup A' \cup A'' \cup A'''$ and 
\beq
w(K) \leq C(r) \be^{-C+c r+c(r) \str A \cup A' \cup A'' \cup A''' \str  }.
\eeq
\end{lemma}

Before proceeding with the proof, let us try to clarify the meaning of this lemma.  Let us choose one term  $\tilde I=\tilde I^{(j,j')}_{A,A'}$ contributing to $J^{(2)}$. Then, the bounds \eqref{eq: products of i} and \eqref{eq: desired bound on k} tell us that, for some operators $K_1,\ldots, K_C$
\beq
\om(\tilde I^*  \tilde I) = \om(K_1)+\ldots +\om(K_C) \leq w(K_1)+\ldots +w(K_C)  \leq C(r) \be^{-C+c r+c(r) \str A \cup A' \str  }
\eeq
where the first inequality follows from Theorem \ref{thm: correlation decay}.      Hence, $\tilde I$ is small in the $\norm \cdot \norm_{\om}$ norm.  This fact is of course a crucial ingredient of the intuition that correlations of the current are small in the thermal state.  Furthermore, the lemma stresses the fact that these operators $K_i$ are of the $K$-type and this takes most of the effort in the proof.   This is important because operators of $K$-type are the operators for which we can prove spatial decorrelation estimates and estimate the $\norm \cdot \norm_{\om}$. The philosophy of estimating $\norm K\norm_{\om}$ in Theorem \ref{thm: correlation decay} consists in essence in relating $\norm K\norm_{\om}$ to $\norm K\norm_{\om_0}$. 

Then, let us give the main (quite simple intuition) why $\norm \tilde I \norm_{\om_0}=\om_0( \tilde I^* \tilde I )$ is small, taking for granted that this can then be translated to the $\norm \cdot \norm_\om$ norm.  Recall that $\tilde I$ is a commutator of the form $\tilde I=[V^{\rechts}, V^{\links}]$ for some operators $V^\rechts, V^\links$. The most simple-minded bound is (by Cauchy-Schwarz)
\beq \label{eq: cs for commutants}
\str \om_0([V^\rechts, V^\links])\str \leq  \str \om_0(V^\rechts V^\links)\str +\str \om_0(V^\links V^\rechts )\str   \leq  2\norm V^\rechts \norm_{\om_0} \norm V^\links \norm_{\om_0},
\eeq
 hence it suffices to show that at least one of the $V$-operators is small in the $\norm \cdot \norm_{\om_0}$ and the other not too big.  Of course, it also suffices if this is true in the $\norm \cdot \norm$-norm since 
\beq
\norm V \norm_{\om_0} \leq \norm V \norm.
\eeq
 Let us now apply this to the problem at hand. The operators $V^{\links, j},V^{\rechts, j}$ are small in $\norm \cdot \norm_{\om_0}$ in the cases $j=1,2,5$, and not too big in all other cases $j=3,4$, meaning that the norm is bounded by $C(r)M^C$.  Hence any commutator involving $j=1,2,5$ is obviously small. That leaves the commutators
 $[V^{\links, 3},V^{\rechts, 3}]$,   $[V^{\links, 3},V^{\rechts, 4}]$  and  $[V^{\links, 4},V^{\rechts, 3}]$. Those commutators cannot be controlled by the simple bound \eqref{eq: cs for commutants}.  Instead, the first of these commutators is small by the bound \eqref{eq: mother of all commutators} in Proposition \ref{thm: splitting of hamiltonian} (This was the main result achieved in the previous sections) and the second and third are small because of \eqref{eq: comm with boundary} and  \eqref {eq: comm with boundary 2} in Proposition \ref{thm: splitting of hamiltonian}. The reason that these bounds apply is that
   $V^{\rechts, 4}$ is situated `far to the right' and $V^{\links, 4}$ `far to the left', because they are terms of the Hamiltonian that are situated at the boundary of the strip $\bbS$.

\begin{proof} 
We consider the cases for $j,j'$ separately and we give the proof of \eqref{eq: i as sum},\eqref{eq: desired bound on k} for some exemplary cases, the others being simplifications of the former.    The bounds on \eqref{eq: products of i} are then also obtained analogously and therefore we skip them entirely. 
\\ [2mm]
\noindent  \textbf{The case $j=3,j'=3$.} This is the most intuitive case.   The bound \eqref{eq: mother of all commutators} in Proposition \ref{thm: splitting of hamiltonian}  gives immediately, with $O:= \tilde I^{(3,3)}_{A,A'}$
\beq
\norm O \norm_{\om_0} \leq C(r) \be^{cr -C}, \qquad  \caP_{\leq 2M}(O)=O
\eeq
and hence $O$ is of $K$-type, and the bound on $w(K)$ follows since  $\str A \cup A' \str = C(r)$. \\[1mm] 

\noindent  In what follows, we let  $h_x$ stand for one of the following three operators $a_x,a^*_x, N^q_x$.  

\noindent  \textbf{The case $j=3,j'=4$.} 
Here, $\tilde I^{(3,4)}_{A,A'}$ is a sum of terms of the form 
\beq
[Z_{\rechts, y},E_{x_1}] \otimes E_{x_2},\qquad    E_{x_i}= \caP_{\leq M}(h_{x_i}), \quad  x_1 \in A, x_2 \not \in A'.
\eeq
Since $x_1$ is necessarily on the left boundary of the strip $\bbS_a$, the
 bound \eqref{eq: comm with boundary} in Proposition \ref{thm: splitting of hamiltonian} yields 
\beq 
\norm O \norm_{\om_0}  \leq C(r) \be^{-C+cr}, \qquad \text{with}\, \,   O:=[Z_{\rechts, y},E_{x_1}]
\eeq
Then, $O'= O \otimes E_{x_2}$ is of $K$-type. By the Cauchy-Schwarz inequality
\beq
 \norm O_1 O_2 \norm^2_{\om_0} =     \om_0(O_1 O_2 O^*_2O^*_1) \leq   \norm O_1 \norm_{\om_0}   \norm  O_2 O^*_2O^*_1 \norm_{\om_0} \leq  \norm O_1 \norm_{\om_0}   \norm  O_2\norm^2 \norm O_1 \norm
\eeq
 applied to $O_1=O', O_2= E_{x_2}$, we then get the desired bound on $w(K)$, because $\str A \cup A' \str = C(r)$.\\ [2mm]

\noindent \textbf{The case $j=2,j'=5$.}
Then  $O :=V^{(\rechts, 2)}= \scrA_{2}(A)$ satisfies $\caP_{\leq M}(O)=O$ and $V^{(\links, 5)}$ is (a sum of) operators of the form $(1-\caP_{\leq M})(h_{x})$ or $(1-\caP_{\leq M})(h_{x_1}h_{x_2})$ for some $x,x_1,x_2$.  Let us first do the simpler case $(1-\caP_{\leq M})(h_{x})$.   If $x \not \in s(O)$ then the commutator vanishes so we assume $x \in s(O)$. 
We split
\beq \label{eq: prop ax}
(1-\caP_{\leq M})(h_x) = \caP_{>2 M}(h_x)  + E_x, \qquad    \caP_{\leq 2 M}(E_x)  =(E_x), \qquad  \norm E_x \norm \leq M^C 
\eeq
Obviously, $O \caP_{>2M}(B) = \caP_{>2M}(B) O=0$  for any $B$ so it suffices to consider the $E_x$-term.   We set 
\beq \label{eq: first definition otilde}
O':= [O, E_x], \qquad   \caP_{\leq 2M} (O')=O',
\eeq
such that $K:=O'$ is of $K$-type, 
and we estimate, using the information on $O$ from Theorem \ref{thm: splitting of hamiltonian} 1), 
\beq \label{eq: linfty linfty}
\norm O' \norm \leq 2 \norm E_x  \norm \norm O \norm \leq  C(r) \be^{-C+cr +c \str  A \str}
\eeq
Since $\norm O' \norm_{\om_0} \leq \norm O' \norm$ and $  \str A \cup A' \str  \geq c \str A \str$, the desired estimate \eqref{eq: desired bound on k} holds. 

Next, let us consider the case $(1-\caP_{\leq M})(h_{x_1}h_{x_2})$ and we again consider $x_1,x_2 \in s(O)$. 
We can split
\begin{align}
(1-\caP_{\leq M})(h_{x_1}h_{x_2}) & =  (1-\caP_{\leq M})(h_{x_1}) (1-\caP_{\leq M})(h_{x_2})  \\
& + (1-\caP_{\leq M})(h_{x_1}) \caP_{\leq M}(h_{x_2})  \\
& +  \caP_{\leq M}(h_{x_1}) (1-\caP_{\leq M})(h_{x_2}) 
\end{align}
and then 
\beq
(1-\caP_{\leq M})(h_{x_i}) = \caP_{>2 M}(h_{x_i})  + E_{x_i}
\eeq
with $E_{x_i}$ the same properties as in \eqref{eq: prop ax}.  Terms with $ \caP_{>2 M}(h_{x_i})$ vanish again such that all non-vanishing terms consist of operators invariant under $\caP_{\leq 2M}$ whose norm is estimated as in \eqref{eq: linfty linfty}, so also in this case we get operators of $K$-type with the desired estimate on $w(K)$.

In the case where $x_1 \in s(O), x_2 \not \in s(O)$, we define $O':= [O, E_{x_1}]$. 
We  split $h_{x_2}= \caP_{> M}(h_{x_2})+(1- \caP_{> M})(h_{x_2})$. Taking the first term, we obtain the operator
\beq
O' \otimes \caP_{> M}(h_{x_2})
\eeq
which is  of $K$-type, and the desired bound on $w(K)$ follows by the bounds on $O'$ above. For the second term, we now set $E_{x_2}:=(1- \caP_{> M})(h_{x_2})$ and we have again $\caP_{\leq 2M}(E_{x_2})=E_{x_2} $ so that we obtain terms of the type
\beq
O'' = O' \otimes E_{x_2}, \qquad \caP_{\leq 2M}(O'')=O''
\eeq
which is of $K$-type, and the bound on $w(K)$ follows by $\norm O'' \norm_{\om_0} \leq \norm O'' \norm_{} \leq \norm O' \norm \norm E_{x_2} \norm \leq  \be^{-C+cr+c\str  A \str}$ and, again $  \str A \cup A' \str  \geq c \str A \str$.  \\ [2mm]
As already indicated above, the other cases follow analogously. 
\end{proof}

Analogously to Lemma \ref{lem: all currents are k}, we have to check
\begin{lemma}\label{lem: terms in oscillating part}
The operators $\tilde O_{\bbS}$ introduced in Section \ref{sec: decomposition of the current} 
can be written as 
\beq
\tilde O_{\bbS}= \sum_{A \subset \bbS} \tilde O_A, \qquad     \tilde O_A= \sum_{i=1}^C K^{i}_{A} 
\eeq
 where each of the operators $K^{i}_{A} $ is of the $K$-type introduced in Section \ref{sec: observables} (Appendix) such that, for any $K=K^{i}_{A} $, we have $s(K) \subset A $  and
\beq  \label{eq: desired bound on k for o}
w(K) \leq C(r) \be^{ -C +c(r) \str s(K) \str }
\eeq
Similarly, for any pair of the operators   $
\tilde O_A, \tilde O_{A'} $ (not necessarily distinct) with $A \cap A' \neq \emptyset $,  the product   $
\tilde O_A \tilde O^*_{A'} $
is a sum of $C$ operators of  $K$-type satisfying \eqref{eq: desired bound on k for o}. 
\end{lemma}
This is proven using the same ideas as in Lemma \ref{lem: all currents are k}, though there are much less terms to consider. Therefore, we skip the proof. 

\begin{proof}[Proof of Theorem \ref{thm: splitting of current}]
We put
\beq
\tilde I_{A}= \sum_{A_1,A_2: A_1 \cup A_2} \tilde I^{(j,j')}_{A_1,A_2}
\eeq
and from Lemma \ref{lem: all currents are k} and Theorem \ref{thm: correlation decay} 1), and noting that the number of terms in the above sum is bounded by $C^{\str A \str}$,  we get that (for $A \cap A' \neq \emptyset$)
\begin{align}
\om(\tilde I_{A}) & \leq    C(r) \be^{-C+c r+c(r) \str A \str  }  \label{eq: good prop i tilde one}  \\[1mm]
\om(\tilde I_{A}^* \tilde I_{A'}) & \leq     C(r) \be^{-C+c r+c(r) \str A \cup A' \str  }     \label{eq: good prop i tilde two}
\end{align}
We now put 
$ I_A: = \tilde I_A - \om(\tilde I_A) $. Since  $\om( J^{(2)})=0$, we get that 
\beq
J^{(2)}= \sum_A  \tilde I_A \qquad \Rightarrow \qquad J^{(2)}= \sum_A  I_A
\eeq
Moreover,  \eqref{eq: good prop i tilde one}, \eqref{eq: good prop i tilde two} are still valid  for  $I_A,I_{A'}$ replacing $\tilde I_A, \tilde I_{A'}$. 
Hence we have shown that the operators $I_A$ have all properties claimed in Theorem \ref{thm: splitting of current} 

Analogously, from Lemma \ref{lem: terms in oscillating part} and Theorem \ref{thm: correlation decay} 1), we get   (for $A \cap A' \neq \emptyset$)
\begin{align}
\om(\tilde O_{A}) & \leq    C(r) \be^{-C+c(r) \str A \str  }  \label{eq: good prop o tilde one}  \\[1mm]
\om(\tilde O_{A}^* \tilde O_{A'}) & \leq     C(r) \be^{-C+c(r) \str A \cup A' \str  }     \label{eq: good prop o tilde two}
\end{align}
we put  $ O_A: = \tilde O_A - \om(\tilde O_A) $, then we have
\beq
[H, \sum_{A} \tilde O_A] =  [H, \sum_{A}O_A]
\eeq
and the bounds \eqref{eq: good prop o tilde one}, \eqref{eq: good prop o tilde two} are still valid for $O_A, O_{A'}$ replacing $\tilde O_A, \tilde O_{A'}$.
\end{proof}

\subsection{Proof of Theorem \ref{thm: vanishing conductivity} }\label{sec: proof of main theorem}
At this point, we reinstate the dependence on the hyperplane position $a$, writing $J_{\bbH_a}^{(j)}$. 
Let us now define
\beq
\caJ^{(j)}_\tau =  \frac{1}{\sqrt{\tau \str \bigvolume \str }} \int^\tau_0  \d t  \sum_{a }   J^{(j)}_{\bbH_a}(t)
\eeq
and recall from \eqref{eq: precise current decomp} that $\caJ_\tau= \caJ^{(1)}_\tau+\caJ^{(2)}_\tau$. 
By Cauchy-Schwarz, 
\beq
\om( \caJ_\tau   \caJ_\tau) \leq 2 \sum_{j=1,2}   \om( \caJ^{(j)}_\tau   \caJ^{(j)}_\tau )
\eeq
And hence we can estimate $j=1,2$ separately. 
\subsubsection{The current $\caJ^{(1)}_\tau$}\label{sec: first current}
For $j=1$, we use, with $O_{\bbS_a}= \tilde O_{\bbS_a}-\om(\tilde O_{\bbS_a})$, 
\beq
\caJ^{(1)}_\tau  =  \frac{1}{\sqrt{\tau \str \bigvolume \str }}    \sum_{ a}   (O_{\bbS_a}(\tau) -  O_{\bbS_a}(0) )
\eeq
and hence, by using again Cauchy-Schwarz and  the invariance of  $\om$ under the dynamics, 
\begin{align}
\om( \caJ^{(1)}_\tau  \caJ^{(1)}_\tau )  &  \leq    \frac{4}{\tau \str \bigvolume \str} \sum_{a,a'}     \om( O_{\bbS_a} O_{\bbS_{a'}} ) \nonumber  \\[2mm]
& =  \frac{4}{\tau \str \bigvolume \str} \sum_{a,a'}     \om( \tilde O_{\bbS_a} ;\tilde O_{\bbS_{a'}} )  \nonumber  \\[2mm]
&=   \frac{4}{\tau \str \bigvolume \str} \sum_{a,a',A,A'}     \om( \tilde O_{A,a} ;\tilde O_{A',{a'}} )  \chi(A \subset \bbS_{a,r^2}) \chi(A' \subset \bbS_{a',r^2}).  \label{eq: decorrelating sum o}
  \end{align}
Using Theorem \ref{thm: correlation decay} 2),  Lemma \ref{lem: terms in oscillating part}, and particular the fact that all our estimates are uniform in the hyperplane position $a$, we see that  $ \om(  \tilde O_{A,a} ;   \tilde O_{A',a'}  )$ decays exponentially in  $\distance(A,A')$. 
and we bound \eqref{eq: decorrelating sum o} by $\be^{-C}/\tau$. 
\subsubsection{The current $\caJ^{(2)}_\tau$} \label{sec: second current}
For $j=2$, we proceed somewhat differently; by Cauchy-Schwarz,
\beq
\om( \caJ^{(2)}_\tau   \caJ^{(2)}_\tau  )  =   \frac{1}{\str \bigvolume \str} \sum_{a,a'} \int^{\tau}_0 \d t  \int^{\tau}_0\d t'     \om(  J^{(2)}_{\bbH_{a}}(t) J^{(2)}_{\bbH_{a'}}(t') ) \leq    \frac{\tau}{\str \bigvolume \str}\sum_{a,a'}     \om(  J^{(2)}_{\bbH_a}(t) J^{(2)}_{\bbH_{a'}}(t)) \label{eq: variance caj two}
\eeq
By  time-translation invariance we can drop $t$ in the argument. 
Since  $\om(J^{(2)}_{\bbH_a})=0$,   the last expression equals a connected correlation function
\beq \label{eq: corr current two}
 \frac{\tau}{\str \bigvolume \str}\sum_{a,a'} \om(  J^{(2)}_{\bbH_a} ;  J^{(2)}_{\bbH_{a'}}  )  =  \frac{\tau}{\str \bigvolume \str}\sum_{a,a'} \sum_{A,A'}    \om(  \tilde I_{A,a} ;   \tilde I_{A',a'}  )  \, \chi(A \subset \bbS_{a, r^2+2} )  \chi(A' \subset \bbS_{a', r^2+2} )
\eeq
Using Theorem \ref{thm: correlation decay} 2),  Lemmas \ref{lem: all currents are k}, and particular the fact that all our estimates are uniform in the hyperplace position $a$, we see that  $ \om(  \tilde I_{A,a} ;   \tilde I_{A',a'}  )$ decays exponentially in $\distance(A,A')$, and in particular, we  bound \eqref{eq: variance caj two}  by
\beq \label{eq:  bound decorrelating sum}
\tau \be^{-C+cr}
\eeq

\subsubsection{Bound on $\ka_{\tau}(\be)$ }

Combining the conclusions from Sections \ref{sec: first current} and \ref{sec: second current}
\beq
\ka_{\tau}(\be) = \be^2 \om( \caJ_\tau   \caJ_\tau) \leq 2\be^2 \sum_{j=1,2}   \om( \caJ^{(j)}_\tau   \caJ^{(j)}_\tau ) \leq     \frac{\be^{-C}}{\tau}+   \tau \be^{-C+cr}
\eeq
Taking now $ \tau = \be^{-m}$, we get Theorem \ref{thm: vanishing conductivity}.

\appendix
 \section{Appendix: Decay of correlations}  

%

In this section, we prove some clustering properties of the high-temperature states in our model.  Recall  the states 
$
\om_{\be, \bigvolume}(\cdot)$ and $ \om_{\be, \bigvolume,0}(\cdot)$ defined in \eqref{def: recall and introduce states}.
In what follows, we again suppress the dependence on $\be,\bigvolume$ since all of our estimates will be uniform in $\bigvolume$ and in $\be$ whenever $\be$ is small enough, hence we simply write $\om(\cdot), \om_0(\cdot)$.    

\subsection{Result} \label{sec: observables}

Recall from Section \ref{sec: energy cutoff} the projection operators $\caP_{\leq M}$ and $\caP_{>M}$ acting on potentials and operators.  Throughout this section, we set
\beq
M=  \be^{-(1+\gamma_c)/q}, \qquad   \textrm{for some}\,\,  0< \gamma_c < q/(q-1)-1
\eeq
We specify two classes of observables. 
The first class consists of low-energy operators $O$, satisfying
\beq
O= \caP_{\leq 2M} (O), \qquad   \text{and} \qquad \str s(O) \str < \infty
\eeq 
The second class of observables is defined starting from monomials $Y$ in creation/annihilation operators
\beq
  Y=    a^{\sharp}_{x_m} \ldots  a^{\sharp}_{x_2} a^{\sharp}_{x_1},
\eeq
for some $x_1,x_2,\ldots, x_m \in \bbZ^d$, $m \in \bbN$ and $a^{\sharp}_x$ either $a^{*}_x$ or $a_x$.  Moreover, we assume the polynomial to be normal-ordered, i.e.\ all $a_x$ appear to the right of $a^*_x$. 
 We let $\degree(Y):=m$, i.e.\  the degree of $Y$ and, for any $x \in s(Y)=\{x_1,\ldots, x_m\}$, we define  $\degree_x(Y)$ as the number of $j \in \{1, \ldots, m \}$ such that $x_j=x$. Then  
 \beq \sum_{x \in s(Y)} \degree_x(Y)=\degree(Y). \eeq  
Given a low-energy observable $O$ and a monomial $Y$  as above, with $s(O) \cap s(Y) =\emptyset$,  we consider
\beq \label{eq: def k}
K= O \otimes  \caP_{>M}(Y) 
\eeq
allowing that $O=\lone$ or $Y=\lone$, corresponding to $s(O)=\emptyset $ and $ \degree(Y)=0$.    We will refer to operators of the form \eqref{eq: def k}  'observables of $K$-type'. 
This class of operators is chosen so that it matches our needs as closely as possible, but it is of course in no sense the maximal one for which a result like the upcoming theorem can be proven.

\bet[Correlation decay at high temperature] \label{thm: correlation decay} Assume that $q>1$  and 
 fix a parameter $\alpha$ such that $0 < 2 \al <1-1/q$.   Let us abbreviate, for monomials $Y$ as above, 
 \beq
v(Y)  :=     \be^{- \degree(Y)/2}  \,    (\e^{-\be^{-\gamma_c/2}})^{\str s(Y) \str  }     \prod_{x \in s(Y)}  \degree_x(Y)! 
\eeq
provided that  $Y \neq \lone$ and $v(Y)=1$ if  $Y =\lone$. 
 There is a $\be_c>0$ such that for $\be<\be_c$, the following hold, for all observables $K,K'$ of $K$-type with $O,Y$ as in \eqref{eq: def k}, 
   \ben
 \item
 \beq \label{eq: definition of w k}
\str \om(K)  \str \leq    C^{\str s(O)\str+ \degree(Y)}  \norm O \norm_{\om_0}   v(Y)  =: w(K)
\eeq
\item 
\beq \label{eq: decorrelation}
\str \om(K; K')  \str \leq  w(K) w(K')   \sum_{x \in s(K),x' \in s(K')}\be^{\al\str x-x'\str}, \qquad \text{for}\,\,  s(K) \cap s(K') =\emptyset
\eeq
\een
The constant $C$ in \eqref{eq: definition of w k} depends only on $\al,\gamma_c$, the exponent $q$ and the spatial dimension $d$. 
\eet
From the estimate in \eqref{eq: decorrelation} and inspecting the range of values for $\al$, we could guess the behaviour of the correlation length  as a function of $q,\beta$
\beq
\xi_{\mathrm{corr.}} \propto \frac{ q}{(q-1)} \str\ln\be\str^{-1}, \qquad \text{for small $\be$}.
\eeq
We see that $\xi_{\mathrm{corr.}}$
diverges as $q \to 1$. This is consistent with the fact that for $q=1$, the system is harmonic and the correlation length  is seen to be independent of $\be$. In contrast,  as the above formula shows, for $q>1$, our upper bound for the correlation length decreases with decreasing $\be$.  

There is an extensive literature on exponential decay of correlations at high temperature, i.e.\ results like Theorem \ref{thm: correlation decay}.  However, we did not find any existing result that fits our needs. This is due to 1) the fact that our one-site space is unbounded and 2) the necessity to have bounds in terms of the Hilbert-Schmidt norm (or some other norm that can capture the sparseness) of the observables, as we have on the right hand side of the inequality \eqref{eq: definition of w k} and hence the right hand side of \eqref{eq: decorrelation}. 
The work \cite{abdesselamprocacciscoppola} addresses the first point, in that it treats unbounded spin systems, and \cite{ciprianipra} gets close to addressing the second point, but we have not found any combination of these results.   In classical spin systems, the approach to decay of correlations via the logarithmic Sobolev inequality or Poincare inequality provides just the type of bounds we need, see e.g.\ \cite{bodineauhelffer}, but, as far as we know, this approach has not been fully adapted to the quantum case yet.

 Therefore, we set up a cluster expansion to prove Theorem \ref{thm: correlation decay}, following to some extent \cite{netocnyredig}.  This is organized as follows. In Section \ref{sec: polymer rep}, we give the general setup which is not specific to our model and which contains some basic results and philosophy from cluster expanions. In Section \ref{sec: bounds on polymer weights}, we prove bounds on so-called polymer weights needed to carry through the cluster expansion. It is this part where we deal with the unboundedness of the on-site Hilbert space, and, more generally, where we need the observables to be of $K$-type. 
  In the short Section \ref{sec: proof of correlation decay}, we combine the bounds of Section \ref{sec: bounds on polymer weights} with the machinery of Section \ref{sec: polymer rep} to give the proof of Theorem \ref{thm: correlation decay}. We should stress that the material in Sections \ref{sec: polymer rep} and \ref{sec: proof of correlation decay} is completely standard, therefore we present proofs in those sections in a compact way.

\subsection{Polymer representations and cluster expansion}\label{sec: polymer rep}
To decompose the Hamiltonian, we will use 'plaquettes' $B$. Each plaquette is defined to consist of a finite set $s(B) \subset \bbZ^d$ and, for each $x \in s(B)$, a pair of variables $(\si_{x,+},\si_{x,-}) \in \bbN \times \bbN $.   To such a plaquette $B$, we associate  the operator
\beq
V_B= \prod_{x \in s(B)}   (a^*_x)^{\si_{x,+}} (a_x)^{\si_{x,-}}
\eeq
Moreover, we restrict ourselves to the case where, $ s(B)= \{x\}$ or $s(B)=\{x,x'\}$ for some $x,x'$ with $\str x-x'\str=1$, and 
$$\sum_{x \in s(B)} (\si_{x,+}+\si_{x,-}) \leq 2. $$
If needed, we indicate that $\si_{x,\pm}$ are associated to a plaquette by writing $ \si_{x,\pm}(B)$.

Then, our  Hamiltonian can be written as
\beq
H_\bigvolume= H^{(0)}_\bigvolume + \sum_{B: s(B) \subset \bigvolume} g(B) V_B
\eeq
where the sum is over plaquettes $B$, and $g(B)$ is a coupling constant satisfying $\str g(B)\str \leq 1$. 
We consider finite collections  $\Ga$ of  pairs $(B,\tau)$ with  $ \tau \in [0,\be] $ and we write them  as ordered sequences
\beq \label{def: biggamma}
\Ga =  ((B_1, \tau_1), \ldots, (B_n, \tau_n)), \qquad \text{with  $n =\str \Ga \str$}
\eeq
such that $(\tau_1,\ldots, \tau_n) $ is in the simplex
\beq
\Delta_n(\be) = \{0 \leq  \tau_1 \leq \tau_2 \ldots \leq \tau_n \leq \be  \}
\eeq
The ambiguity in \eqref{def: biggamma} that occurs when $\tau_j=\tau_{j+1}$ will be irrelevant as we will mostly  integrate $\tau_1,\ldots, \tau_n$ with the Lesbegue measure.
For convenience, we also define the collection of plaquettes appearing in \eqref{def: biggamma};
\beq
\caB(\Ga):=\{ B_j: j=1,\ldots, n  \},
\eeq
and $n_B$ the multiplicity with which a plaquette $B$ appears, such that  
\beq
\sum_{B \in \caB(\Ga)} n_B = n
\eeq

\subsubsection{Polymer representation of the partition function}

For a sequence $\Ga$ as in  \eqref{def: biggamma}, we set, for $\Ga\neq \emptyset$, 
\beq
R(\Ga):= V_{B_n}(\tau_n) \ldots V_{B_2}(\tau_2) V_{B_1}(\tau_1), \qquad \text{with}\,\,  V_B(\tau) = \e^{\tau H_\bigvolume^{(0)}}  V_B \e^{-\tau H_\bigvolume^{(0)} }   \label{eq: formula polymer}
\eeq
and $R(\emptyset):=\lone$. 
Then we can represent the partition function
\beq
Z_\bigvolume=Z_\bigvolume(\be) = \Tr \e^{-\be H_\bigvolume}
\eeq
as a  series;
\beq \label{eq: expansion partition function}
\frac{Z_\bigvolume}{Z_{\bigvolume,0}} =\int_{} \,  \d\Ga \,  \om_0( R(\Ga)) , \qquad 
\eeq
where we used the shorthand
\beq
\int_{} \,  \d\Ga \, \ldots \; =\;  \sum_{n \geq 0}   \,\,  \sum_{\substack{B_1, \ldots, B_n \\   s(B_j) \subset \bigvolume   }}   \,\,    \mathop{ \int}\limits_{\Delta_{n}(\be)} \d \tau_1 \ldots \d \tau_n \quad \ldots
\eeq
and it is understood that for  $n=0$, the sums/integrals are absent. For example, the right hand side of \eqref{eq: expansion partition function} starts with the term $\om_0(R(\emptyset))=\omega_0(\lone)=1$.  
Formally, the identity \eqref{eq: expansion partition function} follows readily by the Duhamel expansion. To establish this rigorously, one first checks that  the series on the right hand side is absolutely convergent, uniformly for $g(B) \in \{ z \in \bbC: \, \str z \str \leq 1 \} $. This is not explicitly proven here but one can easily deduce it from the bounds derived in Section \ref{sec: bounds on polymer weights}.  Therefore, the right hand side of   \eqref{eq: expansion partition function} is the Taylor series of an analytic function in $g(B)$. By explicit calculation, one checks that it coincides with the Taylor series of the left hand side.

For two finite sets $S,S' \subset \bbZ^d$, we define the adjacency relation
\beq
S \sim S' \quad \Leftrightarrow   \quad S \cap S' \neq \emptyset
\eeq
and we call a collection $\caS$ of sets $S$  connected if the collection is connected by the adjacency relation $\sim$.   A connected collection will below also be called a cluster.
We say that $\Ga$ is connected iff.\ the collection $ \caS(\Ga):= \{s(B): \, B \in \caB(\Ga)\}$ is connected. 
If $\Ga$ is not connected, then we can decompose $\caS(\Ga)$ in a unique way into maximally connected components, and this induces a decomposition  $\Ga_1, \ldots, \Ga_m$, such that
\beq
\Ga= \Ga_1 \cup \ldots  \cup \Ga_m,
\eeq
We then obtain the factorization
\beq \label{eq: factorization}
\om_0( R(\Ga))= \prod_{j=1}^m  \om_0( R(\Ga_j))
\eeq
because $\om_0$ is a product state, i.e.\ $\om_0(O_1O_2)= \om_0(O_1)\om_0(O_2)$ whenever $s(O_1) \cap s(O_2) = \emptyset$. 
It is now advantageous to reorganize the expansion \eqref{eq: expansion partition function} by collecting the contributions of connected $\Ga$ corresponding to  the same domain  $s(\Ga):=\cup_{B \in \caB(\Ga)} s(B)$. 

 To that end, we define, for a finite, nonempty $S$,  
\begin{align}
\varrho(S)   &  := \mathop{\int}\limits_{\substack{s(\Ga)=S \\ \Ga \,  \textrm{connected}}}  \d\Ga \,      \omega_{0} (R(\Ga))
\end{align}
Let us denote by $\frB_{\bigvolume}$ the set of all finite collections $\caS$ of sets $S \subset \bigvolume$ and we call such a collection $\caS \in \frB_{\bigvolume}$ admissible iff., for any two different $S,S' \in \caS$, $S \nsim S'$. 
Then our polymer representation for the partition function reads
\beq \label{eq: expansion partition function again}
\frac{Z_\bigvolume}{Z_{\bigvolume,0}} = \sum_{\substack{\caS \in \frB_{\bigvolume} \\  \caS\,  \textrm{admissible}} }  \prod_{S \in \caS}  \varrho(S)  
\eeq
where the term with $\caS =\emptyset$ is defined to be $1$.
To check \eqref{eq: expansion partition function again}, one relies on \eqref{eq: factorization} and a similar factorization property for the sums/integrals abbreviated by $\int \d \Ga$.

\subsubsection{Abstract cluster expansion}\label{sec: abstract cluster}
In this section, it is convenient to take an abstract point of view.   Consider complex weights  $\varpi(S) \in \bbC$ for finite sets $S \subset \bbZ^d$. 
Define
\beq
\Upsilon_{\bigvolume} : =  \sum_{\substack{\caS \in \frB_{\bigvolume} \\  \caS\,  \textrm{admissible}} }  \prod_{S \in \caS}  \varpi(S)  \eeq
For a collection $\caS$, we introduce `truncated weights'
 \beq \label{def: truncated weights}
 \varpi^T(\caS)   = \sum_{\scrG \in \frG^c(\caS) }  (-1)^{\str  \scrE(\scrG) \str}    \prod_{\{S, S'\} \in  \scrE(\scrG)} 1_{[S \sim S']}  \prod_{S'' \in \caS} \varpi(S'')      
 \eeq
where the sum runs over $\frG^c(\caS)$, the set of connected graphs with vertex set $\caS$,
 $\scrE(\scrG) $ is the edge set of the graph $\scrG$ (there are no self-edges), and the first product runs over the edge set  $\scrE(\scrG) $. 
Note that if $\caS$ is not a cluster, then $\varpi^T(\caS)=0$. 

Next, we state the basic result of cluster expansions, cfr.\ (eq. 4) in \cite{ueltschi}.

\begin{theorem}   \label{thm: basic cluster expansion result}
Assume there is  ${a} >0$  such that, for any $x$,
\beq
\sum_{S \subset \bigvolume: S \sim \{x\}}   \e^{{a} \str S \str}  \str \varpi(S)\str   \leq {a}.     \label{eqkotecky preiss abstract}
\eeq
Then $\Upsilon_\bigvolume \neq 0$, 
 \beq
\log \Upsilon_\bigvolume  =  \sum_{\caS \in \frB_\bigvolume}    \varpi^T(\caS),
 \eeq
and, for any $x$,
\beq
 \sum_{\caS \in \frB_\bigvolume:   \caS \sim \{x\}  }    \left \str \varpi^T(\caS) \right\str    \leq    {a}   \label{eqbound on clusters touching something}
\eeq
where the condition $\caS \sim S'$ means that there is a $S \in \caS$ such that $S \sim S'$.
\end{theorem}
In what follows, we use the notation $\varrho^T(\cdot)$, defined from weights $\varrho(\cdot)$, as in the abstract case above.  

\subsubsection{Expansion for observables and correlations} \label{sec: expansion observables correlations}

We have already defined the notion of connectedness for sequences $\Ga$ as connectedness for the collection $\caS(\Ga)$. Given a nonempty set $A$, we say that $\Ga$ is $A$-connected
if the collection $\caS(\Ga) \cup_j \{A_j\}$ is connected, with $A_j$ the connected components of $A$. 

Consider an operator $K$ with $\str s(K) \str <\infty$. 
 For a finite set $S \subset s(K)^c$, we define formally (because we do not address here the convergence of the series on the right hand side)
\beq \label{def: varrho k}
\varrho_{K}(S) :=    \mathop{\int}\limits_{\substack{ \Ga  \text{ $s(K)$-connected} \\ s(\Ga) \cap s(K)^c =S }}    \d \Ga \,    \om_0(R(\Ga) K)  \,  
\eeq
The contribution to the right hand side from $\Ga=\emptyset$ is  $\om_0(K)$ whenever $s(K)$ is connected, and $0$ whenever $s(K)$ is not connected.  
Note that for $S=\emptyset$, the constraint in \eqref{def: varrho k} reads simply $s(\Ga) \subset s(K)$ whenever $s(K)$ is connected,  and then  $\varrho_{K}(\emptyset) $ does in general not vanish, whereas  $\varrho_{K}(\emptyset)=0$ whenever $s(K)$ is not connected. Note also that, if $\varrho_{K}(S)\neq 0$ and $S\neq \emptyset$, then  $S$ has distance $1$ to any of the connected components of $s(K)$. 
  Let us for the time being, until the end of Section \ref{sec: expansion observables correlations}, assume that $s(K), s(K')$ are connected. 
By mimicking the steps leading to \eqref{eq: expansion partition function again}, we then obtain the following polymer representation for $\om(K) $
\begin{align}
\om(K)  & = \frac{Z_0}{Z} \sum_{S_0 \subset \bigvolume}  \varrho_{K}(S_0)   \sum_{\substack{\caS \in \frB_{\bigvolume \setminus (S_0 \cup s(K)) } \\  \caS\,  \textrm{admissible}} }  \prod_{S \in \caS}  \varrho(S) \\[3mm]
& = \frac{Z_0}{Z} \sum_{S_0 \subset \bigvolume}  \varrho_{K}(S_0)   \sum_{\substack{\caS \in \frB_{\bigvolume }  \\  \caS\,  \textrm{admissible}}}   \prod_{S \in \caS}  \varrho(S)  \chi(S \nsim (S_0 \cup s(K)))   \label{eq: modified sum over admissible cas} \end{align}

Let us now assume that the  criterion \eqref{eqkotecky preiss abstract} of Theorem \ref{thm: basic cluster expansion result} is satisfied for some ${a}$, then we can apply Theorem \ref{thm: basic cluster expansion result}  both  to the quotient of partition functions in \eqref{eq: expansion partition function again} and to each term in the $S_0$-sum in \eqref{eq: modified sum over admissible cas} to obtain
\begin{align}
\log \frac{Z}{Z_0} &= \mathop{\sum}\limits_{\caS \in \frB_\bigvolume}  \varrho^T(\caS) \\[3mm] 
\log \sum_{\substack{\caS \in \frB_{\bigvolume }  \\  \caS\,  \textrm{admissible}}}   \prod_{S \in \caS}  \varrho(S)  \chi(S \nsim (S_0 \cup s(K)))  &=  \mathop{\sum}\limits_{\caS \in \frB_\bigvolume} \chi[\caS \nsim (S_0 \cup s(K))]  \varrho^T(\caS)
\end{align}
Therefore, we can write. 
\begin{align}
\om(K)  & = \sum_{S_0}  \varrho_{K}(S_0) \,   \e^{ -\mathop{\sum}\limits_{\caS \in \frB_\bigvolume}  \varrho^T(\caS) }\e^{ \mathop{\sum}\limits_{\caS \in \frB_\bigvolume} \chi[\caS \nsim (S_0 \cup s(K))]  \varrho^T(\caS) } \nonumber \\
& = \sum_{S }  \varrho_{K}(S) \,  \e^{- f (S \cup s(K) )}  \label{eq: final expression omega k}
\end{align}
where it is understood (also below) that $S,S_0$ range over subsets of $\bigvolume$  and we used the  shorthand (up to now only with $m=1$)
\beq
f(A_1,A_2, \ldots, A_m) :=  \mathop{\sum}\limits_{\caS \in \frB_\bigvolume} \chi(\caS \sim A_1,\caS \sim A_2, \ldots, \caS \sim A_m)  \varrho^T(\caS)
\eeq

Take now $K,K'$ such that  $\distance(s(K), s(K')) > 1$ and both $s(K), s(K')$ are connected.  Then $s(KK')= s(K) \cup s(K')$ is not connected.  Mimicking again all the above steps,  and using the definition of $\varrho_{KK'}(\cdot)$, 
we can then derive 
\begin{align}
\om(KK')  & =    \sum_{\substack{S,S'   \\   (S \cup s(K))  \nsim  (S' \cup s(K'))  }}   \varrho_{K}(S)\varrho_{K}(S') \,   \e^{- f (S \cup S' \cup s(KK')) }  +   \sum_{S}  \varrho_{K K'}(S)  \,   \e^{- f (S \cup s(K K')) } 
\end{align}
such that, after some algebra involving in particular the identity 
\beq
f(A_1 \cup A_2)= f(A_1)+ f(A_2)- f(A_1,A_2) 
\eeq
we obtain
\begin{align}
\om(K;K')  & =  \sum_{S,S'}  \varrho_{K}(S)\varrho_{K'}(S') \e^{- f (S \cup S' \cup s(KK')) }  \,  \left(  \e^{  f (S \cup s(K), S' \cup s(K')) } - 1\right) \nonumber \\[2mm]
 & -     \sum_{\substack{S,S'   \\   (S \cup s(K))  \sim  (S' \cup s(K'))  }}       \varrho_{K}(S)\varrho_{K'}(S') \,   \e^{- f (S \cup S' \cup s(KK')) }   \nonumber  \\[2mm]
 &+ \sum_{S}  \varrho_{K K'}(S)  \,   \e^{- f (S \cup s(K K') )  }     \label{eq: final expression omega two k}
\end{align}
where we used the shorthand $\om(K;K')=\om(KK')-\om(K)\om(K') $. 
This formula can be used to exhibit some decay of the correlation $\om(K;K')$ in $\distance(s(K),s(K'))$, as we explain now. 
\begin{lemma}\label{eq: cluster decay}
Assume that the criterion \eqref{eqkotecky preiss abstract} is satisfied for the weights $\varrho_{\theta}(S):=  \theta^{-\str S \str}  \varrho(S)$ for some $a>0$ and $0<\theta<1$. Then 
\beq \label{eq: bounds on f one and two}
\str f(A) \str\leq a \theta \str A \str, \qquad   \str f(A,A') \str \leq a \sum_{x \in A, x' \in A'}\theta^{\str x-x'\str}
\eeq
Let $K,K'$ be observables such that $s(K),s(K')$ is connected, but $\distance(s(K),s(K'))>1$.  For an observable $\tilde K$, let $\conn(\tilde K)$ be the number of connected components of the set $s(\tilde K)$ and let
\beq
b_{\tilde\theta}(\tilde K):=\str \varrho_{\tilde K}(\emptyset)\str + \sum_{S: S \neq \emptyset}   \tilde\theta^{-(\str S \str+ \conn(\tilde K))} \, \str S \str \str \varrho_{\tilde K}(S)\str, \qquad 0 < \tilde \theta<1.
\eeq
Then 
\begin{align}  
\str \om(K) \str  &\leq     \theta_1^{-\str s(K) \str} b_{ \theta_1}(K).  \label{eq: stuff with bk one}  \\[2mm]
\str \om(K;K') \str  &\leq   C(1+a)  \theta_1^{-\str s(K) \str-\str s(K')\str}  (b_{\theta \theta_1}(K)b_{\theta \theta_1}(K') + b_{\theta \theta_1}(KK')) \,  \sum_{x \in s(K), x' \in s(K')} \theta^{\str x -x' \str}.   \label{eq: stuff with bk}
\end{align}
with $\theta_1= \e^{-2 a \theta}$ and with the constant $C$ independent of $\theta, a$. 
\end{lemma}
\begin{proof}
Note first, by inspection of \eqref{def: truncated weights}, that 
\beq
\varrho^T_{\theta}(\caS) =    \theta^{-\sum_{S \in \caS} \str S \str} \varrho^T(\caS).
\eeq
 Therefore the estimate
\begin{align}
\sum_{\caS}\chi(\caS \sim \{x\},\caS \sim \{x'\})  \str \varrho^T(\caS)\str & \leq   \theta^{\str x-x'\str} \sum_{\caS}  \chi(\caS \sim \{x\})  \str \varrho^T_\theta(\caS)\str  \leq a   \theta^{\str x-x'\str}. 
\end{align}
follows from the simple observation  $\inf_{\caS: \caS \sim \{x\},\caS \sim \{x'\}}   \theta^{\sum_{S \in \caS} \str S \str} \leq    \theta^{\str x-x'\str}  $ and Theorem \ref{thm: basic cluster expansion result} applied with $\varpi=\varrho_\theta$. Summing over $x \in A,x' \in A'$, this yields the second claim in  \eqref{eq: bounds on f one and two}, whereas the first one follows more directly from Theorem \ref{thm: basic cluster expansion result}. 

With the estimates   \eqref{eq: bounds on f one and two} in hand, the proof of \eqref{eq: stuff with bk} is a lengthy but straightforward calculation starting from \eqref{eq: final expression omega two k}.  The stated bound is proven for the three terms on the right hand side of \eqref{eq: final expression omega two k}. Let us do the first term which is the most complicated one. 
Using the bounds $\str \e^{z}-1 \str \leq \str z \str \e^{\str z \str} $ for $z \in \bbC$, and the bounds in \eqref{eq: bounds on f one and two}, we get 
\begin{align}
 & \Big \str \sum_{S,S'}  \varrho_{K}(S)\varrho_{K}(S') \e^{- f (S \cup S' \cup s(KK')) }  \,  \left(  \e^{  f (S \cup s(K), S' \cup s(K')) } - 1\right) \Big \str\nonumber  \\[2mm]
 & \leq      a  \sum_{S,S'}  \str \varrho_{K}(S) \str \str \varrho_{K}(S')\str \,   \e^{2a \theta (\str S \str+ \str S' \str + \str s(K)\str+\str s(K')\str )} \,\,  \times \nonumber  \\[2mm]
& \qquad   \Big( \sum_{x \in S, x' \in S'} + \sum_{x \in S, x' \in s(K')} +\sum_{x \in s(K), x' \in S'} + \sum_{x \in s(K), x' \in s(K')}  \Big) \Big( \theta^{\str x-x'\str}\Big)  \label{eq: split in four}
\end{align}
We again split this into four terms corrsponding to the four sums in the last expression. The fourth sum gives, upon summing $S,S'$
\beq
a\e^{2a \theta ( s(K)+s(K')} b_{\theta_1}(K) b_{\theta_1}(K')    \sum_{x \in s(K), x' \in s(K')}  \theta^{\str x-x'\str}.
\eeq
The third sum gives, upon summing $S$,  
\begin{align}
 & a\e^{2a \theta ( s(K)+s(K'))} b_{\theta_1}(K)   \sum_{S'} \sum_{x \in s(K), x' \in S'}  \theta^{\str x-x'\str}   \str \varrho_{K}(S')\str \, \theta^{(\str S' \str+1)}   (\theta\theta_1)^{ -(\str S' \str+1) } \\[2mm]
 &  \leq   a\e^{2a \theta ( s(K)+s(K'))} b_{\theta_1}(K)  b_{\theta_1\theta}(K')   \sum_{x \in s(K)}  \theta^{\distance(x,s(K'))}    \end{align}
 where we used the triangle inequality $\str S' \str+1 + \str x-x'\str \geq \distance(x,s(K'))$ to get the last inequality.  The first and second sums in \eqref{eq: split in four} are similar. Hence, since $b_{\theta'}(K) \leq b_{\theta''}(K)$ for $\theta' \geq \theta''$, we have obtained the desired bound on \eqref{eq: split in four}, namely \eqref{eq: stuff with bk}. 
 In the two remaining terms of \eqref{eq: final expression omega two k}, we always estimate $\e^{f(A)}$ by $\e^{a \theta \str A \str}$. The smallness comes then from the constraint on $S,S'$.  The bound on \eqref{eq: stuff with bk one}  is obtained analogously, but simpler. 

\end{proof}

\subsection{Bounds on polymer weights} \label{sec: bounds on polymer weights}

The following lemma contains estimates on the weights $\varrho$, from which Theorem \ref{thm: correlation decay} will easily follow.  Throughout this section, we assume that $\be$ is taken small enough and we not repeat this at every step.

\begin{lemma} \label{lem: bound varrho} 
Fix a parameter $\al=\al(q)$ satisfying $0<2\al<1-1/q$.
Recall the weights $\varrho(S), \varrho_K(S)$ from Section \ref{sec: polymer rep}. Then 
\begin{enumerate}
\item   \beq \str \varrho(S) \str \leq   (C\be)^{\al\str S \str} \eeq
\item   Consider an observable $K= O \otimes \caP_{\geq M}(Y)$ `of $K$-type', as defined in Section \ref{sec: observables}.  Then,   
 \beq \str \varrho_K(S) \str \leq w(K)\times \begin{cases} (C\be)^{ \al (\str S \str+\conn(K))}  & S \neq \emptyset \\ 1 & S=\emptyset  \end{cases}, \eeq
 with $w(K)$ as defined in Theorem \ref{thm: correlation decay} and $\conn(K)$ the number of connected components of $s(K)$.  
\end{enumerate}
 uniformly in $\bigvolume$, provided that $S, s(K) \subset \bigvolume$.  
\end{lemma}
Note that replacing $(C\beta)$ by $\beta$ in the above lemma yields an equivalent claim upon adjusting $\al$. The same will be true often in the proof, below in Section  \ref{sec: proof of lemma varrho}, but we prefer to keep the constants to avoid repeated readjusting of exponents.    However, we do need to readjust constants, in particular the constant in the definition of $w(K)$. 
Before giving the lengthy proof of Lemma \ref{lem: bound varrho}, let us first use it to give the 
\subsubsection{Proof of Theorem \ref{thm: correlation decay}} \label{sec: proof of correlation decay}
We give the proof in the case where the sets $s(K),s(K')$ are connected (because  Lemma \ref{eq: cluster decay} is restricted to this case). The general case follows by the same reasoning.\\

\noindent\emph{Step $1$} For any $\al'$ satisfying $0<2\al'<1-1/q$,  the criterion \eqref{eqkotecky preiss abstract} is satisfied for the weights $\varrho_{\theta}(S):=  \theta^{-\str S \str}  \varrho(S)$ with $a=1$ and $\theta=\be^{\al'}$.    To see this, we combine Lemma \ref{lem: bound varrho} 1) for some $\al'' >\al'$ with
 the geometrical fact
\beq \label{eq: geo fact}
\sum_{S: S \ni x} \chi(S \, \textrm{connected})    c^{\str S \str}  \leq 1, \qquad \text{for small enough}\, c  
\eeq

\noindent \emph{Step $2$} For any $\al'$ satisfying $0<2\al'<1-1/q$ and observable $K$ of $K$-type, we establish
\beq
b_\theta(K) \leq  w(K), \qquad \text{with $\theta= \be^{\al'}$}.
\eeq
  This is a straightforward consequence of  Lemma \ref{lem: bound varrho} 2) for some $\al''>\al'$, using again the geometrical fact \eqref{eq: geo fact} (and keeping in mind that $\distance(S,s(K))=1$ whenever $\varrho_K(S)\neq 0$ and $S \neq \emptyset$).\\
  
\noindent \emph{Step $3$}    
The two claims of Theorem \ref{thm: correlation decay}  follow by the results \eqref{eq: stuff with bk one} and \eqref{eq: stuff with bk}  of Lemma \ref{eq: cluster decay}, using Steps $1$ and $2$ above with $\al'>\al$ and noting that, for $\theta=\be^{\al}$, the quantity $\theta_1=\e^{-2a \theta}$ in Lemma \ref{eq: cluster decay} can made arbitrarily close to $1$ by taking $\beta$ large enough, and that $w(KK')=w(K)w(K') $ whenever $s(K) \cap s(K')=\emptyset $.

%
%

\subsubsection{Proof of Lemma \ref{lem: bound varrho}} \label{sec: proof of lemma varrho}

Let us first fix some additional notation. For a given $\Ga$, we set
 $$     \si_{x}(B) := \si_{x,+}(B)+\si_{x,-}(B), \qquad  n(x):=  \sum_{B_j : s(B_j) \ni x}   \si_{x}(B_j),   $$
  and  
  \beq
  \caN(\Ga) := \prod_{x \in s(\Ga)} n_x !
  \eeq
  We introduce a 'cut-off state'
  \beq
  \om_{0,2M}(O) :=     \om_{0}(  \caP_{\leq 2M}(O)) 
  \eeq
The following lemma is a purely combinatorial bound.  Recall the quantity
\beq
v(Y)  = \be^{- \degree(Y)/2}  \,    (\e^{-\be^{-\gamma_c/2}})^{\str s(Y) \str  }     \prod_{x \in s(Y)}  \degree_x(Y)! 
\eeq
for a monomial $Y$ (introduced in Theorem \ref{thm: correlation decay}).  
\begin{lemma} \label{lem: bounds on gammas}
Fix a parameter  $\kappa$ such that $1>\ka>1/q$. Then, for any $\Ga$ and monomial $Y$, 
\begin{enumerate}
\item
\begin{align}
  \str \om_{0} (R( \Ga) \str   & \leq    \caN(\Ga)^{1/2}   \prod_{x \in s(\Ga)}    C^{n(x)}   \be^{-n(x) \ka/2}    \label{eq: bound by max one}
\end{align}
\item
\begin{align}
  \str \om_{0,2M} (R(\Ga) R^*(\Ga) \str  
  & \leq  \caN(\Ga)     \prod_{x \in s(\Ga)}     C^{n(x)}  \be^{-n(x)\ka} 
\end{align}
\item 
\begin{align}
\str \om_{0} (R( \Ga) \caP_{\geq M}(Y) \str   & \leq   \caN(\Ga)^{1/2} C^{\degree(Y)} v(Y)     \prod_{x \in s(\Ga)}    C^{n(x)}   \be^{-n(x) \ka/2}    \label{eq: bound by max two}
\end{align}

\end{enumerate}
\end{lemma}
\begin{proof}
Consider a sequence $\eta_0, \eta_1, \ldots, \eta_n$ in $\Om_{s(\Ga)}$ with $n=\str \Ga \str$.  We note, 
 by inserting decompositions of identity and using cyclity of the trace,   that
\beq
\str  \om_0( R(\Ga)) \str \leq  \sum_{\substack{\eta_0,\eta_1,\ldots, \eta_{n-1}\\[1mm] \eta_n=\eta_0     } }  \e^{\sum_{j=1}^n\tau_j(E(\eta_j)-E(\eta_{j-1}))}  
 \prod_{j=1}^n \str  \langle \eta_{j},V_{B_j} \eta_{j-1}\rangle\str \label{eq: naive bound paths}
\eeq 
Since $0 \leq \tau_1 \leq \ldots \leq \tau_n \leq \beta$, we can bound the exponent as
\begin{align}
\sum_{j=1}^n\tau_j(E(\eta_j)-E(\eta_{j-1})  & = \int_0^{\beta} \d \tau  \,  \tau \frac{\partial e}{\partial \tau} = \int_0^{\beta} \d \tau  \frac{\partial (\tau e) }{\partial \tau}- \int_0^{\beta} \d \tau \, e  \nonumber\\[1mm] 
 &   \leq \beta  e(\beta)-\beta \inf_{\tau} e(\tau)     
= \be E(\eta_n)-\min_{j}  \be E(\eta_j) = \be E(\eta_0)-\min_{j} \be E(\eta_j)  \nonumber \\[1mm] 
 & \leq \be  \sum_x \big( E(\eta_0(x))-\min_{j} E(\eta_j(x))\big)  \nonumber \\[1mm]        
  & \leq \be  \sum_x \big( E(\eta_0(x))-E(\eta_0(x)-n(x)/2)\big)  \label{eq: path bound}   
\end{align}
where we let the function $e(\tau)$ on $[0,\be]$ be the linear interpolation of $\tau_j \mapsto E(\eta_j)$ with $e(0)=E(\eta_0)$ and we adopted the convention that $E(\xi)=\xi^q$ for $\xi>0$ and $E(\xi)=0 $ for $ \xi \leq 0$. 
The last inequality follows by using that $n(x)$ is the number of field $a_x/a^*_x$ operators appearing on site $x$, and $\eta_0(x)=\eta_n(x)$.
Combining \eqref{eq: naive bound paths} and \eqref{eq: path bound}, using the basic bound $\str \langle \eta(x), a^*_x (\eta(x)-1) \rangle \str \leq \sqrt{\eta(x)}$  and abbreviating 
\beq  Z_0(\be)= \sum_{\xi \in \bbN}  \e^{-\be E(\xi) }, \eeq we get   
\begin{align}
  \str \om_{0} (R( \Ga) \str
   & \leq \sum_{\eta_0 \in \Om_{s( \Ga)}}  \om_0(P_{\eta_0}) \prod_{x \in s( \Ga) }     \frac{\lfloor \eta_0(x)+ n(x)/2 \rfloor! }{\eta_0(x)!}  \e^{\be (E(\eta_0(x))- E(\eta_{0}(x)- n(x)/2 )  )} 
   \label{eq: basic one path}    \\[2mm]
  & \leq   \prod_{x \in s( \Ga)}  \sum_{\eta(x)}   \frac{\lfloor \eta(x)+ n(x)/2 \rfloor! }{(\eta(x))!} \, \times \,  \frac{\e^{-\be E(\eta(x)- n(x)/2 )}}{Z_0(\beta)} \label{eq: binomial} 
\end{align}
For any $0 \leq z \leq  1$, we can use the bound $\frac{m!}{p!(m-p)!} \leq z^{-p}(1-z)^{-(m-p)}$ to get
\begin{align}
 \frac{1 }{ \prod_x \lfloor n_x/2 \rfloor! } \str \om_{0} (R( \Ga) \str  &  \leq    \prod_{x \in s( \Ga)}    \sum_{\eta(x)}  (1-z)^{-n(x)/2}  z^{-\eta(x)}  \, \times \,  \frac{\e^{-\be E(\eta(x)- n(x)/2 )}}{Z_0(\beta)}
\end{align}
Similarly (take $z=1/2$), we have
\beq
\prod_{x \in s(\Ga)} C^{n_x} \lfloor n_x/2 \rfloor! \geq  \caN(\Ga)^{1/2}
\eeq
Let us now choose $z=e^{-\be^{\kappa}}$. 
 For sufficiently small $\be$, we can then estimate
\beq
(1-z)^{-n} \leq  2^n \be^{-n \kappa}
\eeq
and we obtain 
\begin{align}
 \caN(\Ga)^{-1/2} \str \om_{0} (R( \Ga) \str   & \leq   \prod_{x \in s(\Ga)}  \sum_{\eta(x)}  C^{n(x)}  \be^{-n(x) \kappa}    \, \times \,  \frac{\e^{-\be E(\eta(x)- n(x)/2 ) + \be^{\kappa} \eta(x)}}{Z_0(\beta)}    
  \label{eq: bound by max}
\end{align}
Since  $\kappa \geq 1/q$,  we can bound
\beq
  \sum_{\eta(x)}  \frac{\e^{-\be E(\eta(x)- n(x)/2 ) + \be^{\kappa} \eta(x)}}{Z_0(\beta)}   \leq   C^{n(x)}
\eeq
by using the explicit form of $E(\cdot)$. 
Hence
\begin{align}
 \caN(\Ga)^{-1/2} \str \om_{0} (R( \Ga) \str   & \leq   \prod_x    C^{n(x)}   \be^{-n(x) \kappa/2}   
\end{align}
The claim of $1)$ now follows since $n(x) \geq 1$ for any $x \in s(\Ga)$. \\[1mm]

To get  $3)$, we first restrict ourselves to the case $s(Y) \subset s(\Ga)$.
We view for notational convenience the $a_x/a_x^*$-operators in $Y=    a^{\sharp}_{x_m} \ldots  a^{\sharp}_{x_2} a^{\sharp}_{x_1}$ as additional plaquettes $B_{i=1,\ldots,m}$ with $ s(B_i)= \{x_i\}$, $\si_{x_i,+/-}=1/0$ if $a^\sharp_{x_i}= a^*_{x_i}$ and $\si_{x_i,+/-}=0/1$ if $a^\sharp_{x_i}= a_{x_i}$.  We define the ordered sequence
\beq
\Ga'=\big (           (B_1,\tau_1), \ldots, (B_m,\tau_m), (B_{m+1},\tau_{m+1}), \ldots   (B_{m+n},\tau_{m+n}) \big)
\eeq
where $\tau_{1,\ldots,m}=0$ and $(B_{m+1},\tau_{m+1}), \ldots   (B_{m+n},\tau_{m+n})$ are the (ordered) elements of $\Ga$ with renamed indices. 
Now we apply the same reasoning as in the proof of $1)$ to get the analogue of \eqref{eq: bound by max}, which now reads
\begin{align}
 \caN(\Ga')^{-1/2} \str \om_{0} (R( \Ga) \caP_{\geq M}(Y)) \str   & \leq   \prod_{x \in s(\Ga)}  C^{n'(x)}   \be^{-n'(x) \ka}    \sum_{\substack{\eta(x) \geq 0\, \textrm{for}\,  x \not \in s(Y)\\[1mm]
 \eta(x) \geq M\, \textrm{for}  \, x \in s(Y)   }}  \frac{   \e^{-\be E(\eta(x)- n'(x)/2 ) + \be^{\kappa} \eta(x)}    }{Z_0(\beta)}    
  \label{eq: bound by max three} 
  \end{align}
  with $n'(x)$ corresponding to $\Ga'$.   For $x \not \in s(Y)$, we bound the the sum over $\eta(x)$ by $C^{n'(x)}$, as in 1), and, for $x  \in s(Y)$, we bound it as
\beq
\sum_{ \eta(x) \geq M}  \frac{   \e^{-\be E(\eta(x)- n'(x)/2 ) + \be^{\kappa} \eta(x)}    }{Z_0(\beta)}      \leq C^{n'(x)}\e^{-\be^{-\gamma_c/2}}
\eeq  
Hence, altogether, we bound \eqref{eq: bound by max three} as
\begin{align}
   \str \om_{0} (R( \Ga) \caP_{\geq M}(Y)) \str & \leq   \caN(\Ga')^{1/2}     (\e^{-\be^{-\gamma_c/2}})^{\str s(Y) \str  }    \prod_{x \in s(\Ga)}  C^{n'(x)}   \be^{-n'(x) \ka}    
\end{align}
Since $n'(x)=n(x)+\degree_x(Y)$, we can bound
\beq
 \caN(\Ga') \leq  \caN(\Ga)  2^{\degree(Y)} \prod_{x \in s(Y)} \degree_x(Y)!
\eeq
and we get the claim of point $3)$ for the restricted case $s(Y) \subset s(\Ga)$.
In the  general case, we split $Y=Y_1Y_2$ such that $s(Y_1) \subset s(\Ga), s(Y_2) \cap s(\Ga)= \emptyset$, and we use the fact that $\om_0$ is a product state:
\beq
 \om_{0} (R( \Ga) \caP_{\geq M}(Y)) =  \om_{0} (R( \Ga) \caP_{\geq M}(Y_1))  \, \om_{0} (\caP_{\geq M}(Y_2)) 
\eeq
For the first factor we use the bound above (for the restricted case). For the second factor, we show, by analogous but simpler reasoning, that it is bounded by $v(Y_2)$.  Since $v(Y)=v(Y_1)v(Y_2)$, this proves the full claim of $3)$.  \\[1mm] 

For $2)$, we mimick the derivation of  \eqref{eq: basic one path}   to arrive at 
\begin{align}
    \str \om_{0,2M} (R(\Ga) R^*(\Ga) \str      &  \leq    \prod_{x \in s( \Ga)}  \sum_{\eta(x) \leq 2M}   
    \frac{\e^{-\be E(\eta(x))}}{Z_0(\beta)} \,   \frac{(\eta(x)+ n(x))! }{(\eta(x))!}  \e^{2\be (E(\eta(x))- E(\eta_{0}(x)- n(x) )  )} \label{eq: double}
\end{align}
The main difference with the argument  in 1) is that every perturbation term appears twice now (therefore we have now $n(x)$ instead of $n(x)/2$ in the argument of  the factorial) and that we had to apply the bounds of \eqref{eq: path bound} twice. Proceeding as in \eqref{eq: binomial} and \eqref{eq: bound by max}, we bound  \eqref{eq: double} by 
\begin{align}
    \prod_{x \in s( \Ga)}     \sum_{\eta(x) \leq 2M}   \frac{\e^{-\be E(\eta(x))}}{Z_0(\beta)} \,  C^{n(x)}  \be^{-n(x) \kappa_1} \e^{ \be^{\kappa_1} \eta(x)}   \e^{2\be (E(\eta(x))- E(\eta(x)- n(x) ))} 
\end{align}
for any $0< \kappa_1 <1$.  To deal with the right-most  exponential, we note that 
\beq  \label{eq: acro with cutoff}
\sup_{0< 2\be < 1} \sup_{0\leq \tilde\xi\leq \xi \leq 2M} \e^{2\be( \xi^q- \tilde\xi^q)} (2\be)^{\kappa_2(\xi-\tilde\xi)} \leq C(\kappa_2)  <\infty,     
\eeq
for  any $\kappa_2>0$.  Indeed,  for  $0< 2\be < 1$, the function $\xi \mapsto f(\xi)= \e^{2\be \xi^q} (2\be)^{\kappa_2 \xi}$ is decreasing on the interval $[0, (\tfrac{\kappa_2 \str \ln 2\beta \str}{2q \beta})^{\tfrac{1}{q-1}} ]$ and, since   $M =  \be^{-(1+\gamma_c)/q}$ with $0<\gamma_c< q/(q-1)-1$,  we see that  $2M$ lies in this interval.   We use \eqref{eq: acro with cutoff} with $\xi=\eta(x), \tilde \xi= \eta(x)-n(x)$ to obtain, for  $\kappa_1-1/q \geq 0$,
\begin{align}
\caN(\Ga)^{-1}  \str \om_{0,2M} (R(\Ga) R^*(\Ga) \str   
  & \leq        \prod_{x \in s( \Ga)}   \sum_{\eta(x) \leq 2M} \frac{ \e^{-\be E(\eta(x))+ \be^{\kappa_1} \eta(x)}}{Z_0(\beta)}   C^{n(x)}  \be^{-n(x) \kappa_1}   \be^{-2\kappa_2 n(x)}   \\[2mm]
    & \leq      \prod_{x \in s( \Ga)}     C^{n(x)}  \be^{-n(x) (\kappa_1+2 \kappa_2)} 
\end{align}
where the last inequality uses the explicit form of $E(\eta(x))$ to perform the sums over $\eta(x)$.   The claim of  $2)$ follows by choosing $\kappa_1,\kappa_2$ such that 
$\kappa_1 +2 \kappa_2=\kappa$, taking advantage of the fact that $\kappa_2$ can be chosen arbitrarily small. 
\end{proof}
To perform the sum/integral over the sequences $\Ga$, we will need to exploit the smallness of the Lesbegue mass over the simplex $\Delta_{m}(\be)$ for large $m$
\begin{lemma}  \label{lem: combi bound}  For any $\ka' < 1/2$
\beq \label{eq: combi bound}
\int_{s(\Ga)=S} \d\Ga  \,  (c\be)^{-\ka' \sum_{x } n_x } \caN(\Ga)^{1/2} \leq   (C\be)^{\str S\str (1/2 -\ka')}
 \eeq
\end{lemma}
\begin{proof}
We first establish, for any $\Ga$, 
\beq \label{eq: basic bound}
\big(  \prod_{x \in s(\Ga)} \sqrt{n_x!}  \big)   \leq    C^{\sum_{x} n_x } \big(  \prod_{B \in \caB(\Ga)}  n_B!  \big)
 \eeq
 In the remainder of the proof, it is understood (unless mentioned otherwise) that $x$ ranges over $s(\Ga)$ and $B$ over $\caB(\Ga)$. 
To prove \eqref{eq: basic bound}, note that 
\beq n_x = \sum_{B: s(B) \ni x}  \sigma_x(B)n_B  \eeq  where  the maximal number of nonzero terms on the right hand side is $C=C(d)$, hence
\beq \label{eq: prod n b}
n_x!  \leq  C^{n_x}  \prod_{B: s(B) \ni x}    (\sigma_x(B) n_B)!
\eeq
and
\begin{align}
 \prod_x n_x!  &\leq  C^{\sum_x n_x}    \big(\prod_{B: \str s(B) \str=1}  \sigma_x(B) n_B)! \big)     \big(\prod_{B: \str s(B) \str=2}   n_B! \big)^2  \\[2mm]
&\leq  C^{\sum_x n_x}   2^{\sum_{B: \str s(B) \str=1  } 2 n_B }   \big(\prod_{B}   n_B! \big)^2  \label{eq: basic bound three}
\end{align}
where the first inequality follows because $\sigma_x(B) \leq 1$ whenever $\str s(B) \str=2$ and because every factor with $\str s(B) \str=2$ appears twice in the product in \eqref{eq: prod n b}. The second inequality follows from $\sigma_x(B) \leq 2$. 
Observing that (for now, we use only the first inequality)
\beq \label{eq: nb and nx}
\sum_{x}n_x \geq  \sum_B n_B  \geq  (1/2) \sum_{x}n_x,
\eeq
we get \eqref{eq: basic bound} from \eqref{eq: basic bound three}.

Since $\caN(\Ga)$ does not depend on the $\tau$-variables, we can perform all $\d \tau$-integrals on the LHS of \eqref{eq: combi bound}. This gives the products of the Lesbegue measure of the simplices;
\beq
 \prod_{B}  \frac{\be^{n_B}}{n_B !}  
\eeq
Using this bound, the definition of $\caN(\Ga)$,  \eqref{eq: basic bound} and \eqref{eq: nb and nx}, we bound the LHS of \eqref{eq: combi bound} by 
\beq
 \sum_{\caB: \mathop{\cup}\limits_{B \in \caB} s(B)=S} \,\,   \sum_{\substack{n_B \geq 1: B \in \caB \\[1mm] n_\caB \geq \str S \str/2 }} \,\,  (C \be)^{(1-2\ka')n_\caB }, \qquad \text{with $n_\caB : = \sum_{B \in \caB} n_{B}$} 
\eeq
where the sum is now over collections $\caB$ of plaquettes.   Using that the number of terms in the leftmost sum is bounded by $C^{\str S \str}$, the claim follows by straightforward combinatorics.
\end{proof}

\begin{proof}[Proof of Lemma \ref{lem: bound varrho}]
To prove $1)$, we write for any $1 > \ka >1/q$, 
\begin{align}
\str \varrho(S)\str \leq & \mathop{\int}\limits_{\substack{s(\Ga)=S \\ \Ga \,  \textrm{connected}}}  \d\Ga \,     \str \omega_{0} (R(\Ga)) \str   \\[1mm] 
\leq &     \mathop{\int}\limits_{s(\Ga)=S}  \d \Ga \, \str \om_0(R(\Ga)) \str     \\[1mm]
      \leq & \mathop{\int}\limits_{s(\Ga)=S}  \d \Ga    \,   \caN(\Ga)^{1/2}   \prod_x    C^{n(x)}   \be^{-n(x)\ka/2}        \\[1mm]
            \leq &   \,   \be^{\str S \str(1/2-\ka/2)}   C^{\str S \str}
\end{align}
where the third inequality follows from Lemma \ref{lem: bounds on gammas} 1) and the fourth from  Lemma  \ref{lem: combi bound} with $\kappa'=\kappa/2$.  The claim follows by setting $\ka=1-2 \al$.\\[1mm]
Next, we prove $2)$:   For any set $S$ and plaquette $B$, we can consider the reduced plaquette $B_S$ with $s(B_S):= s(B) \cap S$ and $\sigma_{x,\pm}(B_S):=\sigma_{x,\pm}(B) $ whenever $x \in s(B_S)$.   Then given a collection $\Ga$, we define the collection $\Ga_{S}$  
\beq
\Ga_{S} := \{ (B_S, \tau): \,    (B,\tau) \in \Ga \,\, \text{and}\,\,  s(B_S)\neq \emptyset \}
\eeq
Note that
\beq \label{eq: disjoint union}
s(\Ga)= s(\Ga_S) \cup s(\Ga_{S^c}), \qquad     s(\Ga_S) \cap s(\Ga_{S^c}) = \emptyset. 
\eeq
and, for $x \in S$, $n(x)$ as defined from $\Ga_S$ equals $n(x)$ as defined from $\Ga$. In particular, we have 
\beq \label{eq: factorzation can}
\caN(\Ga_S) \caN(\Ga_{S^c}) =  \caN(\Ga).  
\eeq
 We apply this below with $S=s(O), s(O)^c$. 
 We have
\begin{align}
 \om_0(R(\Ga) K) &=    \om_0(R(\Ga_{s(O)}) O)   \om_{0} (R( \Ga_{s(O)^c})  \caP_{\geq M}(Y)  ) \\[2mm]
 & =   \om_{0}(\caP_{\leq 2M}(R(\Ga_{s(O)})) \caP_{\leq 2M}(O)) \, \times \,   \om_{0} (R( \Ga_{s(O)^c})  \caP_{\geq M}(Y)  ) \label{eq: acro cutoff one}
\end{align}
where the second equality follows because $\caP_{\leq 2M}(O)=O$ and the density matrix of $\om_0$ is diagonal in the $N_x$-basis.    Using Cauchy-Schwarz, positivity of the projector $ \otimes_{x \in A}\chi(N_x \leq 2 M)$ for any $A$, and $(\caP_{\leq 2M}(O))^*= \caP_{\leq 2M}(O^*)$, the first factor can be bounded as
\begin{align}
\str  \om_{0}(\caP_{\leq 2M}(R(\Ga_{s(O)})) \caP_{\leq 2M}(O)) \str^2  &\leq   \om_{0}(\caP_{\leq 2M}(O^*)\caP_{\leq 2M}(O)) \, \times \, \om_0(\caP_{\leq 2M}(R(\Ga_{s(O)})) \caP_{\leq 2M}(R^*(\Ga_{s(O)}))    )  \\[2mm]  
& \leq   \om_{0, 2M}(O^*O) \,  \times \, \om_{0, 2M}(R(\Ga_{s(O)}) R^*(\Ga_{s(O)})   )  \label{eq: acro cutoff two}
\end{align}

We recall the definition of  $\varrho_{K}(S)$ in \eqref{def: varrho k} and we abbreviate
\beq
 \int_{K} \d \Ga \ldots : = \mathop{\int}\limits_{\substack{ \Ga  \text{ $s(K)$-connected} \\ s(\Ga) \cap s(K)^c =S }} \d \Ga \ldots 
\eeq 
We estimate, for any $1 > \ka >1/q$, 
\begin{align}
\str \varrho_{K}(S) \str &  \leq  \int_{K} \d \Ga \,   \str \om_0(R(\Ga) K)\str \\[2mm]
 & \leq \str\om_{0,2M}(O^* O) \str^{1/2} \int_{K} \d \Ga     \,    \str\om_{0,2M} (R( \Ga_{s(O)})R^*( \Ga_{s(O)}) )\str^{1/2}     \,    \str \om_{0} (R( \Ga_{s(O)^c})  \caP_{\geq M}(Y)  )  \str   \\[2mm] 
 & \leq \str\om_{0,2M}(O^* O) \str^{1/2}  C^{\degree(Y)} v(Y)\int_{K} \d \Ga  \str\caN(\Ga_{s(O)}) \caN(\Ga_{s(O^c)})\str^{1/2}     \prod_{x \in s(\Ga)} (c\be)^{-\ka n(x)/2}  \label{eq: splitt o and c}
\end{align}
The first inequality follows from \eqref{eq: acro cutoff one} and \eqref{eq: acro cutoff two} and the second from Lemma \ref{lem: bounds on gammas} $2),3)$, using \eqref{eq: disjoint union}.   Let us now first take $S \neq \emptyset$. Starting from \eqref{eq: factorzation can}, we rewrite and  bound the $\d \Ga$-integral in \eqref{eq: splitt o and c} as (recall that $\conn(K)$ is the number of connected components of $s(K)$), 
\begin{align}
& \int_K    \d\Ga     \,  \caN(\Ga)^{1/2}     \prod_{x \in s(\Ga)} (c\be)^{-\ka n(x)/2} \\[2mm]
&\leq  \sum_{\substack{S': S' \cap s(K)^c=S\\[1mm]    \str S' \cap s(K) \str \geq  \conn(K)  }} \mathop{\int}\limits_{s(\Ga)=S' }   \d\Ga   \,    \caN(\Ga)^{1/2}     \prod_{x \in s(\Ga)} (c\be)^{-\ka n(x)/2}  \\[2mm]
&\leq   \sum_{\substack{S': S' \cap s(K)^c=S\\[1mm]    \str S' \cap s(K) \str \geq  \conn(K)  }} (C\be)^{\al \str S' \str}  \leq   C^{\str s(K)\str} (C\be)^{\al (\str S \str+\conn(K))} \label{eq: resulting bound}
\end{align}
where we used Lemma \ref{lem: combi bound} with $\ka'=\ka/2$ and we set $\ka=1-2 \al$.  Plugging this into \eqref{eq: splitt o and c}  and recalling the definition of $w(K)$ yields the desired claim.   For $S= \emptyset$, the above proof still applies if we drop the constraint $ \str S' \cap s(K) \str \geq  \conn(K)$ in the last lines. Then the resulting bound on the right hand side of \eqref{eq: resulting bound} is simply $C^{\str s(K)\str}$, and we can again conclude by plugging into  \eqref{eq: splitt o and c}.

\end{proof}


\bibliographystyle{ieeetr}
\bibliography{mylibrary11}

\end{document}

%% file: leftright2.eps_tex

\begingroup
  \makeatletter
  \providecommand\color[2][]{%
    \errmessage{(Inkscape) Color is used for the text in Inkscape, but the package 'color.sty' is not loaded}
    \renewcommand\color[2][]{}%
  }
  \providecommand\transparent[1]{%
    \errmessage{(Inkscape) Transparency is used (non-zero) for the text in Inkscape, but the package 'transparent.sty' is not loaded}
    \renewcommand\transparent[1]{}%
  }
  \providecommand\rotatebox[2]{#2}
  \ifx\svgwidth\undefined
    \setlength{\unitlength}{1345.54863281pt}
  \else
    \setlength{\unitlength}{\svgwidth}
  \fi
  \global\let\svgwidth\undefined
  \makeatother
  \begin{picture}(1,0.78480998)%
    \put(0,0){\includegraphics[width=\unitlength]{leftright2.eps}}%
    \put(0.01535285,0.42468838){\color[rgb]{0,0,0}\makebox(0,0)[lb]{\smash{$B_y$}}}%
    \put(-0.0019741,0.17803379){\color[rgb]{0,0,0}\makebox(0,0)[lb]{\smash{$B_{y'}$}}}%
    \put(0.9241696,0.17497611){\color[rgb]{0,0,0}\makebox(0,0)[lb]{\smash{$F_2$}}}%
    \put(0.91907344,0.48923977){\color[rgb]{0,0,0}\makebox(0,0)[lb]{\smash{$F_1$}}}%
    \put(0.9309645,0.32276492){\color[rgb]{0,0,0}\makebox(0,0)[lb]{\smash{$F_0$}}}%
  \end{picture}%
\endgroup